\documentclass[11pt,a4paper]{article}

\usepackage[english]{babel}
\usepackage{german}
\usepackage{amssymb}
\usepackage{amsfonts}
\usepackage{amsmath}
\usepackage{fullpage}
\usepackage{cite}
\usepackage{enumitem}

\newtheorem{theorem}{Theorem}

\newtheorem{corollary}{Corollary}

\newtheorem{definition}{Definition}
\newtheorem{observation}{Observation}

\newtheorem{lemma}{Lemma}

\newenvironment{proof}[1][Proof]{\noindent\textbf{#1.} }{\ \rule{0.5em}{0.5em}}

\usepackage{tabularx}
\newcommand{\problemdef}[3]{
	\begin{center}
		\begin{minipage}{0.95\textwidth}
			\noindent
			#1
			\vspace{5pt}\\
			\setlength{\tabcolsep}{3pt}
			\begin{tabularx}{\textwidth}{@{}lX@{}}
				\textbf{Input:}& #2 \\
				\textbf{Question:}& #3
			\end{tabularx}
		\end{minipage}
	\end{center}
}

\newcommand{\ie}{i.e., }
\newcommand{\eg}{e.g., }

\newcommand{\NP}{\textrm{NP}}

\newcommand{\ML}{\textsc{ML}}

\newcommand{\MALlong}{\textsc{Min.~Aged Labeling}}
\newcommand{\MAL}{\textsc{MAL}}

\newcommand{\MSLlong}{\textsc{Min.~Steiner Labeling}}
\newcommand{\MSL}{\textsc{MSL}}

\newcommand{\MASLlong}{\textsc{Min.~Aged Steiner Labeling}}
\newcommand{\MASL}{\textsc{MASL}}

\newcommand{\ST}{{\textsc{Steiner Tree}}}
\newcommand{\VC}{\textsc{Vertex Cover}}
\newcommand{\MCC}{\textsc{Multicolored Clique}}
\newcommand{\MAXSAT}{\textsc{Monotone Max XOR(3)}}
\newcommand{\MAXSATshort}{\textsc{MonMaxXOR(3)}}

\usepackage{caption} 
\usepackage{subcaption}

\usepackage{graphicx}

\usepackage[capitalise,nameinlink, noabbrev]{cleveref}
\begin{document}

\title{\vspace{-0.5cm}The complexity of computing optimum labelings \\
	for temporal connectivity}

\author{Nina Klobas\thanks{%
Department of Computer Science, Durham University, UK. Email: \texttt{nina.klobas@durham.ac.uk}}
\and George B.~Mertzios \thanks{%
Department of Computer Science, Durham University, UK. Email: \texttt{george.mertzios@durham.ac.uk}} \thanks{Supported by the EPSRC grant EP/P020372/1.} 
\and 
Hendrik Molter\thanks{%
Department of Industrial Engineering and Management, Ben-Gurion~University~of~the~Negev, Israel. Email: \texttt{molterh@post.bgu.ac.il}} \thanks{Supported by the ISF, grant No.~1070/20.}
\and Paul G. Spirakis\thanks{%
Department of Computer Science, University of Liverpool, UK and Computer Engineering \& Informatics Department, University of Patras, Greece. Email: \texttt{p.spirakis@liverpool.ac.uk}} \thanks{Supported by the NeST initiative of the School of EEE and CS at the University of Liverpool and by the EPSRC grant EP/P02002X/1.}}
\date{\vspace{-1.0cm}}

\maketitle

\begin{abstract}
	A graph is temporally connected if there exists a strict temporal path, \ie a path whose edges have \emph{strictly increasing} labels, from every vertex $u$ to every other vertex $v$. 
	In this paper we study \emph{temporal design} problems for undirected temporally connected graphs. 
	The basic setting of these optimization problems is as follows: given a connected undirected graph~$G$, what is the smallest number $|\lambda|$ of time-labels that we need to add to the edges of $G$ such that the resulting temporal graph $(G,\lambda)$ is temporally connected? 
	As it turns out, this basic problem, called \textsc{Minimum Labeling} (\ML), can be optimally solved in polynomial time. 
	However, exploiting the temporal dimension, the problem becomes more interesting and meaningful in its following variations, which we investigate in this paper. 
	First we consider the problem \MALlong\ (\MAL) of temporally connecting the graph when we are given an upper-bound on the allowed \emph{age} (\ie maximum label) of the obtained temporal graph $(G,\lambda)$. 
	Second we consider the problem \MSLlong\ (\MSL), where the aim is now to have a temporal path between any pair of ``important'' vertices which lie in a subset $R\subseteq V$, which we call the \emph{terminals}. 
	This relaxed problem resembles the problem \textsc{Steiner Tree} in static (\ie non-temporal) graphs. However, due to the requirement of \emph{strictly} increasing labels in a temporal path, \textsc{Steiner Tree} is \emph{not} a special case of \MSL. 
	Finally we consider the age-restricted version of \MSL, namely \MASLlong\ (\MASL). 
	Our main results are threefold: we prove that (i)~\MAL\ becomes NP-complete on undirected graphs, while (ii)~\MASL\ becomes W[1]-hard with respect to the number $|R|$ of terminals. 
	On the other hand we prove that (iii)~although the age-unrestricted problem \MSL\ remains NP-hard, it is in FPT with respect to the number $|R|$ of terminals. 
	That is, adding the age restriction, makes the above problems \emph{strictly harder} (unless P=NP or W[1]=FPT).\\
	
\end{abstract}

\section{Introduction}\label{intro-sec}

A temporal (or dynamic) graph is a graph whose underlying topology is subject to discrete changes over time. 
This paradigm reflects the structure and operation of a great variety of modern networks; 
social networks, wired or wireless networks whose links change dynamically, transportation networks, and several physical systems are only a few examples of networks that change over time~\cite{michailCACM,Nicosia-book-chapter-13,holme2019temporal}. 
Inspired by the foundational work of Kempe et al.~\cite{KKK00}, we adopt here a simple model for
temporal graphs, in which the vertex set remains unchanged while each edge is equipped with a set of integer time-labels.

\begin{definition}[temporal graph~\cite{KKK00}]
	\label{temp-graph-def} A \emph{temporal graph} is a pair $(G,\lambda)$,
	where $G=(V,E)$ is an underlying (static) graph and $\lambda :E\rightarrow 2^\mathbb{N}$ is a \emph{time-labeling} function which assigns to every edge of $G$ a set of discrete time-labels.
\end{definition}

Here, whenever $t\in \lambda(e)$, we say that the edge $e$ is \emph{active} or \emph{available} at time $t$. 
Throughout the paper we may refer to ``time-labels'' simply as ``labels'' for brevity. 
Furthermore, the \emph{age} (or \emph{lifetime}) $\alpha(G,\lambda)$ of the temporal graph $(G,\lambda)$ is the largest time-label used in it, \ie $\alpha(G,\lambda) = \max \{t\in \lambda(e) : e\in E\}$. 
One of the most central notions in temporal graphs is that of a \emph{temporal path} (or \emph{time-respecting path}) 
which is motivated by the fact that, due to causality, entities and information in temporal graphs can
``flow'' only along sequences of edges whose time-labels are strictly increasing, or at least non-decreasing.

\begin{definition}[temporal path]
	\label{temp-path-def} 
	Let $(G,\lambda)$ be a temporal graph, where $G=(V,E)$ is the underlying static graph. 
	A \emph{temporal path} in $(G,\lambda)$ is a sequence $(e_1,t_1), (e_2,t_2), \ldots, (e_k,t_k)$, where 
	$(e_1,e_2,\ldots,e_k)$ is a path in $G$, $t_i\in \lambda(e_i)$ for every $i=1,2,\ldots, k$, and $t_1 < t_2 < \ldots < t_k$.
\end{definition}

A vertex $v$ is \emph{temporally reachable} (or \emph{reachable}) from vertex $u$ in $(G,\lambda)$ if there exists a temporal path from $u$ to $v$. If every vertex $v$ is reachable by every other vertex $u$ in $(G,\lambda)$, then $(G,\lambda)$ is called \emph{temporally connected}. 
Note that, for every temporally connected temporal graph $(G,\lambda)$, we have that its age is at least as large as the diameter $d_G$ of the underlying graph $G$. Indeed, the largest label used in any temporal path between two anti-diametrical vertices cannot be smaller than $d_G$.
Temporal paths have been introduced by Kempe et al.~\cite{KKK00} for temporal graphs which have only one label per edge, \ie $|\lambda(e)|=1$ for every edge $e\in E$, and this notion has later been extended by Mertzios et al.~\cite{Mertzios2013Temporal} to temporal graphs with multiple labels per edge. 
Furthermore, depending on the particular application, both variations of temporal paths with non-decreasing~\cite{KKK00,AxiotisF16,KlobasMMNZ21-TemporalDisjoint} and with strictly increasing~\cite{Mertzios2013Temporal,EnrightMMZ21} labels have been studied. In this paper we focus on temporal paths with \emph{strictly increasing} labels. 
Due to the very natural use of temporal paths in various contexts, several path-related notions, 
such as temporal analogues of distance, diameter, reachability, exploration, and centrality have also been  studied~\cite{XuanFJ03,AkridaMNRSZ19,AkridaGMS16,HaagMNR22,AxiotisF16,CasteigtsPS21,DeligkasP20,EnrightMMZ21,enright2021assigning,ErlebachHK15,ErlebachS18,BussMNR20,Mertzios2013Temporal,MichailS-TSP16,AkridaGMS17,MolterRZ21,KlobasMMNZ21-TemporalDisjoint,ThejaswiLG21}. 

Furthermore, some non-path temporal graph problems have been recently introduced too, including
for example temporal variations of maximal cliques~\cite{viardCliqueTCS16,HimmelMNS17}, vertex cover~\cite{AkridaMSZ18-TVC,HammKMS22-TVC}, 
vertex coloring \cite{MertziosMZ21coloring}, 
matching \cite{MertziosMNZZmatching20}, and transitive orientation~\cite{MertziosMRSZ21-transitivity}. 
Motivated by the need of restricting the spread of epidemic, Enright et al.~\cite{EnrightMMZ21} studied the problem of removing the smallest number of time-labels from a given temporal graph such that every vertex can only temporally reach a limited number of other vertices. 
Deligkas et al.~\cite{DeligkasES21-arxiv} studied the problem of accelerating the spread of information for a set of sources to all vertices in a temporal graph, by only using delaying operations, \ie by shifting specific time-labels to a later time slot. 
The problems studied in~\cite{DeligkasES21-arxiv} are related but orthogonal to our temporal connectivity problems. 
Various other temporal graph modification problems have been also studied, see for example~\cite{AxiotisF16,DeligkasP20,MolterRZ21,CasteigtsPS21,enright2021assigning}.

The time-labels of an edge $e$ in a temporal graph indicate the discrete units of time (\eg days, hours, or even seconds) 
in which $e$ is active. However, in many real dynamic systems, \eg in synchronous mobile distributed systems that operate in 
discrete rounds, or in unstable chemical or physical structures, maintaining an edge over time requires energy and thus comes at a cost. 
One natural way to define the \emph{cost} of the whole temporal graph $(G,\lambda)$ is the \emph{total number} of time-labels used in it, \ie the total cost of $(G,\lambda)$ is $|\lambda|=\sum_{e\in E} |\lambda_e|$.

In this paper we study \emph{temporal design} problems of undirected temporally connected graphs. 
The basic setting of these optimization problems is as follows: given an undirected graph~$G$, 
what is the smallest number $|\lambda|$ of time-labels that we need to add to the edges of $G$ such that $(G,\lambda)$ is temporally connected? 
As it turns out, this basic problem can be optimally solved in polynomial time, thus answering to a conjecture made in~\cite{AkridaGMS17}. 
However, exploiting the temporal dimension, the problem becomes more interesting and meaningful in its following variations, which we investigate in this paper. 
First we consider the problem variation where we are given along with the input also an upper bound of the allowed \emph{age} (\ie maximum label) of the obtained temporal graph $(G,\lambda)$. 
This age restriction is sensible in more pragmatic cases, where delaying the latest arrival time of any temporal path incurs further costs, \eg when we demand that all agents in a safety-critical distributed network are synchronized as quickly as possible, and with the smallest possible number of communications among them. 
Second we consider problem variations where the aim is to have a temporal path between any pair of ``important'' vertices which lie in a subset $R\subseteq V$, which we call the \emph{terminals}. 
For a detailed definition of our problems we refer to \cref{prelim-sec}.

Here it is worth noting that the latter relaxation of temporal connectivity resembles the problem \textsc{Steiner Tree} in static (\ie non-temporal) graphs. Given a connected graph $G=(V,E)$ and a set $R\subseteq V$ of terminals, \textsc{Steiner Tree} asks for a smallest-sized subgraph of $G$ which connects all terminals in $R$. Clearly, the smallest subgraph sought by \textsc{Steiner Tree} is a tree. 
As it turns out, this property does not carry over to the temporal case. Consider for example an arbitrary graph $G$ and a terminal set $R=\{a,b,c,d\}$ such that $G$ contains an induced cycle on four vertices $a,b,c,d$; that is, $G$ contains the edges $ab, bc, cd, da$ but not the edges $ac$ or $bd$. Then, it is not hard to check that only way to add the smallest number of time-labels such that all vertices of $R$ are temporally connected is to assign one label to each edge of the cycle on $a,b,c,d$, \eg $\lambda(ab)=\lambda(cd)=1$ and $\lambda(bc)=\lambda(cd)=2$. The main underlying reason for this difference with the static problem \textsc{Steiner Tree} is that temporal connectivity is \emph{not transitive} and \emph{not symmetric:} if there exists temporal paths from $u$ to $v$, and from $v$ to $w$, it is not a priori guaranteed that a temporal path from $v$ to $u$, or from $u$ to $w$ exists.

Temporal network design problems have already been considered in previous works. 
Mertzios et al.~\cite{Mertzios2013Temporal} proved that it is APX-hard to compute a minimum-cost labeling for temporally connecting an input \emph{directed} graph $G$, where the age of the graph is upper-bounded by the diameter of $G$. 
This hardness reduction was strongly facilitated by the careful placement of the edge directions in the constructed instance, in which every vertex was reachable in the static graph by only constantly many vertices. Unfortunately this cannot happen in an undirected connected graph, where every vertex is reachable by all other vertices. 
Later, Akrida et al.~\cite{AkridaGMS17} proved that it is also APX-hard to \emph{remove} the largest number of time-labels from a given temporally connected (undirected) graph $(G,\lambda)$, while still maintaining temporal connectivity. In this case, although there are no edge directions, the hardness reduction was strongly facilitated by the careful placement of the initial time-labels of $\lambda$ in the input temporal graph, in which every pair of vertices could be connected by only a few different temporal paths, among which the solution had to choose. Unfortunately this cannot happen when the goal is to add time-labels to an undirected connected graph, where there are potentially multiple ways to temporally connect a pair of vertices (even if we upper-bound the largest time-label by the diameter).

Summarizing, the above technical difficulties seem to be the reason why the problem of \emph{adding} the minimum number of time-labels with an age-restriction to an \emph{undirected} graph to achieve temporal connectivity remained open until now for the last decade. 
In this paper we overcome these difficulties by developing a hardness reduction from a variation of the problem \textsc{Max~XOR~SAT} (see \cref{thm:NPDiameterStatic} in \cref{MAL-NP-complete-sec}) where we manage to add the appropriate (undirected) edges among the variable-gadgets such that simultaneously~(i) the distance between any two vertices from different variable gadgets remains small (constant) and~(ii) there is no shortest path between two vertices of the \emph{same} variable gadget that leaves this gadget.

\textbf{Our contribution and road-map.} 
In the first part of our paper, in \cref{MAL-NP-complete-sec}, we present our results on \MALlong\ (\MAL). 
This problem is the same as \ML, with the additional restriction that we are given along with the input an upper bound on the allowed \emph{age} of the resulting temporal graph $(G,\lambda)$. 
Using a technically involved reduction from a variation of \textsc{Max~XOR~SAT}, 
we prove that \MAL\ is NP-complete on undirected graphs, even when the required maximum age is equal to the diameter $d_G$ of the input static graph $G$.

In the second part of our paper, in \cref{Steiner-sec}, we present our results on the Steiner-tree versions of the problem, namely on \MSLlong\ (\MSL) and \MASLlong\ (\MASL). 
The difference of \MSL\ from \ML\ is that, here, the goal is to have a temporal path between any pair of ``important'' vertices which lie in a given subset $R\subseteq V$ (the \emph{terminals}). 
In \cref{MSL-NP-complete-subsec} we prove that \MSL\ is NP-complete by a reduction from \textsc{Vertex Cover}, the correctness of which requires showing structural properties of \MSL. 
Here it is worth recalling that, as explained above, the classical problem \textsc{Steiner Tree} on static graphs is \emph{not} a special case of \MSL, due to the requirement of strictly increasing labels in a temporal path. 
Furthermore, we would like to emphasize here that, as temporal connectivity is neither transitive nor symmetric, a straightforward NP-hardness reduction from \textsc{Steiner Tree} to \MSL\ does not seem to exist. 
For example, as explained above, in a graph that contains a $C_4$ with its four vertices as terminals, labeling a Steiner tree is sub-optimal for \MSL.

In \cref{MSL-FPT-subsec} we provide a fixed-parameter tractable (FPT) algorithm for \MSL\ with respect to the number $|R|$ of terminal vertices, by providing a parameterized reduction to \textsc{Steiner Tree}. 
The proof of correctness of our reduction, which is technically quite involved, is of independent interest, 
as it proves crucial graph-theoretical properties of minimum temporal \textsc{Steiner} labelings. 
In particular, for our algorithm we prove (see~\cref{lem:MSLstructure}) that, for any undirected graph $G$ with a set $R$ of terminals, there always exists at least one minimum temporal \textsc{Steiner} labeling $(G,\lambda)$ which 
labels edges either from (i) a tree or from (ii) a tree with one extra edge that builds a $C_4$.

In \cref{MASL-W1-hard-subsec} we prove that \MASL\ is W[1]-hard with respect to the number $|R|$ of terminals. 
Our results actually imply the stronger statement that \MASL\ is W[1]-hard even with respect to the number of time-labels of the solution (which is a larger parameter than the number $|R|$ of terminals).

Finally, we complete the picture by providing some auxiliary results in our preliminary \cref{prelim-sec}. 
More specifically, in~\cref{gossip-polynomial-subsec} we prove that \ML\ can be solved in polynomial time, 
and in~\cref{DAG-subsec} we prove that the analogue minimization versions of \ML\ and \MAL\ on directed acyclic graphs are solvable in polynomial time.

\section{Preliminaries and notation}\label{prelim-sec}

Given a (static) undirected graph $G=(V,E)$, an edge between two vertices $u,v\in V$ is denoted by $uv$, and in this case
the vertices $u,v$ are said to be \emph{adjacent} in $G$. If the graph is directed, we will use the ordered pair $(u,v)$ (resp.~$(v,u)$) to denote the oriented edge from $u$ to $v$ (resp.~from $v$ to $u$). 
The \emph{age} of a temporal graph $(G,\lambda)$ is denoted by $\alpha(G,\lambda) = \max \{t\in \lambda(e) : e\in E\}$. 
A temporal path $(e_1,t_1),(e_2,t_2),\ldots, (e_k,t_k)$ from vertex $u$ to vertex $v$ is called \emph{foremost}, 
if it has the smallest arrival time $t_k$ among all temporal paths from $u$ to $v$. Note that there might be another temporal path from $u$ to $v$ that uses fewer edges than a foremost path. 
A temporal graph $(G,\lambda)$ is \emph{temporally connected} if, for every pair of vertices $u,v\in V$, there exists a temporal path (see \cref{temp-path-def}) $P_1$ from $u$ to $v$ and a temporal path $P_2$ from $v$ to $u$. 
Furthermore, given a set of terminals $R\subseteq V$, the temporal graph $(G,\lambda)$ is \emph{$R$-temporally connected} if, for every pair of vertices $u,v\in R$, there exists a temporal path from $u$ to $v$ and a temporal path from $v$ to $u$; note that $P_1$ and $P_2$ can also contain vertices from $V\setminus R$. 
Now we provide our formal definitions of our four decision problems.

\medskip

\noindent \fbox{ 
	\begin{minipage}{0.471\textwidth}
		\textsc{Min.~Labeling (ML)}
		
		\vspace{1.2mm}
		{\bf{Input:}} A static graph $G = (V,E)$ and \vskip 0pt
		a $k\in \mathbb{N}$. \\
		{\bf{Question:}} Does there exist a temporally \vskip 0pt
		connected temporal graph $(G,\lambda)$, \vskip 0pt
		where $|\lambda |\leq k$?
\end{minipage}} 
\vspace{0,2cm} \noindent \fbox{ 
	\begin{minipage}{0.471\textwidth}
		\textsc{Min.~Aged Labeling (MAL)}
		
		\vspace{1.2mm}
		{\bf{Input:}} A static graph $G = (V,E)$ \\
		and two integers $a,k\in \mathbb{N}$. \\
		{\bf{Question:}} Does there exist a temporally connected temporal graph $(G,\lambda)$, \\
		where $|\lambda|\leq k$ and $\alpha(\lambda)\leq a$?
\end{minipage}}

\vspace{0,2cm} \noindent \fbox{ 
	\begin{minipage}{0.471\textwidth}
		\textsc{Min.~Steiner Labeling (MSL)}
		
		\vspace{1.2mm}
		{\bf{Input:}} A static graph $G = (V,E)$, \\
		a 
		subset $R\subseteq V$ and a $k\in \mathbb{N}$. \\
		{\bf{Question:}} Does there exist an temporally \vskip 0pt
		$R$-connected temporal graph $(G,\lambda)$, \\
		where $|\lambda |\leq k$?
\end{minipage}} 
\vspace{0,2cm} \noindent \fbox{ 
	\begin{minipage}{0.471\textwidth}
		\textsc{Min.~Aged Steiner Labeling (MASL)}
		
		\vspace{1.2mm}
		{\bf{Input:}} A static graph $G = (V,E)$, \\
		a subset $R\subseteq V$, and two integers $a,k\in \mathbb{N}$. \\
		{\bf{Question:}} Does there exist a temporally $R$-connected temporal graph $(G,\lambda)$, \\
		where $|\lambda|\leq k$ and $\alpha(\lambda)\leq a$?
\end{minipage}}

\medskip

Note that, for both problems \MAL\ and \MASL, whenever the input age bound $a$ is strictly smaller than the diameter $d$ of $G$, the answer is always \textsc{NO}. Thus, we always assume in the remainder of the paper that $a\geq d$, 
where $d$ is the diameter of the input graph $G$. 
For simplicity of the presentation, we denote next by $\kappa(G,d)$ the smallest number $k$ for which $(G,k,d)$ is a \textsc{YES} instance for \MAL.

\begin{observation}\label{kappa-bound-obs}
	For every graph $G$ with $n$ vertices and diameter $d$, we have that $\kappa(G,d) \leq n(n-1)$.
\end{observation}

\begin{proof}
	For every vertex $v$ of $G=(V,E)$, consider a BFS tree $T_v$ rooted at $v$, while every edge from a vertex $u\neq v$ to its parent in $T_v$ is assigned the time-label $dist(v,u)$, \ie the length of the shortest path from $v$ to $u$ in $G$. Note that each of these time-labels is smaller than or equal to the diameter $d$ of $G$. 
	Clearly, each BFS tree $T_v$ assigns in total $n-1$ time-labels to the edges of $G$, and thus the union of all BFS trees $T_v$, where $v\in V$, assign in total at most $n(n-1)$ labels to the edges of $G$. 
\end{proof}

\medskip

The next lemma shows that the upper bound of Observation~\ref{kappa-bound-obs} is asymptotically tight as, for cycle graphs $C_n$ with diameter $d$, 
we have that $\kappa(C_n,d)=\Theta(n^2)$.

\begin{lemma} \label{thm:BoundsForCycles}
	Let $C_n$ be a cycle on $n$ vertices, where $n \neq 4$, and let $d$ be its diameter. Then
	\begin{equation*}
		\kappa(C_n,d) = 
		\begin{cases}
			d^2, & \text{when $n = 2d$} \\
			2d^2 + d, & \text{when $n = 2d + 1$.}
		\end{cases}
	\end{equation*}
\end{lemma}

\begin{proof}
	Let $V(C_{n})=\{v_{1},v_{2},\dots ,v_{n}\}$ be the vertices of $C_{n}$. In
	the following, if not specified otherwise, all subscripts are
	considered modulo $n$. We distinguish two cases, depending on the parity of $n$.
	
	First, when $n$ is odd, \ie $n=2d+1$. In this case one can observe that for
	each vertex $v_{i}\in V(C_{n})$ there are exactly two vertices on the
	distance $d$ from it, namely $v_{i+d}$ (on the \emph{right side} of $v_{i}$)
	and $v_{i-d}$ (on the \emph{left side} of $v_{i}$). Therefore, the $(v_{i},v_{i+d}/v_{i-d})$ and $(v_{i+d}/v_{i-d},v_{i})$-temporal paths must
	be labeled using all labels from $1$ to $d$, one per each edge. Note also
	that each edge $v_{i}v_{i+1}$ lies on the $d$ temporal paths when the
	starting vertex $v_{j}$ is on the left side of it ($j\in \{i,i-1,i-2,\dots
	,i-d\}$) and on $d$ temporal paths, when the starting vertex $v_{j^{\prime }}
	$ is on the right side of it ($j^{\prime }\in \{i,i+1,i+2,\dots ,i+d\}$).
	This results in edge $v_{i}v_{i+1}$ admitting all labels. As this is true
	for any edge of $C_{n}$, each edge is labeled with all labels. Therefore we
	need $n\cdot d=2d^{2}+1$ labels to ensure the existence of temporal paths
	among any two vertices in $C_{2d+1}$.
	
	Now let us continue with the case when $n$ is even, \ie$n=2d$. In this case
	each vertex $v_{i}\in V(C_{n})$ has exactly one vertex, $v_{i-d}=v_{i+d}$,
	on the distance $d$ from it and two on the distance $d-1$ from it ($v_{i-d+1}
	$ and $v_{i+d-1}$). Therefore we have to label two disjoint paths starting
	in $v_{i}$, one of length $d$ and the other of length $d-1$. Suppose that we
	chose the following labeling to label the edges of $C_{n}$. Let $i\in
	\{1,2,\dots ,d\}$, if the edge is of form $v_{2i}v_{2i+1}$ then it is
	labeled with all even labels, $\{2,4,6,\dots ,j\}$, where $j\leq d$, and if
	the edge is of form $v_{2i+1},v_{2i}$ then it is labeled with all odd
	labels, $\{1,3,5,\dots ,j^{\prime }\}$, where $j^{\prime }\leq d$. Now
	vertices $v_{2i-1}$ and $v_{2i}$ use the same labels (\ie the same temporal
	paths), to reach all other vertices from the cycle. 
	Namely, the ($v_{2i-1},v_{2i+d-1}$)-temporal path is of length $d$, uses all labels from $1
	$ to $d$ and visits vertices $v_{2i-1},v_{2i},v_{2i+1},\dots ,v_{2i+d-1}$.
	Therefore $v_{2i-1}$ and $v_{2i}$ can reach vertices $\{v_{2i+1},v_{2i+2},\dots ,v_{2i+d-1}\}$. 
	Similarly, the ($v_{2i},v_{2i-d}$)-temporal path is of
	length $d$, uses all labels from $1$ to $d$ and visits 
	vertices $v_{2i},v_{2i-1},v_{2i-2},\dots ,v_{2i-d}$. So $v_{2i-1}$ and $v_{2i}$ can
	reach vertices $\{v_{2i-2},v_{2i-3},\dots ,v_{2i-d}\}$. 
	This implies that that $v_{2i}$ and $v_{2i-1}$ reach all other vertices in the graph. This holds for
	any two endpoints of an edge in $C_{n}$. 
	Therefore we need $d\cdot \frac{n}{2}=d^{2}$ labels to ensure the existence of temporal paths among any two
	vertices in~$C_{2d}$.
\end{proof}

\subsection{A polynomial-time algorithm for \ML}\label{gossip-polynomial-subsec}

As a first warm-up, we study the problem \ML, where no restriction is imposed on the maximum allowed age of the output temporal graph. 
It is already known by Akrida et al.~\cite{AkridaGMS17} that any undirected graph can be made temporally connected by adding at most $2n-3$ time-labels, while for trees $2n-3$ labels are also necessary. 
Moreover, it was conjectured that every graph needs at least $2n-4$ time-labels~\cite{AkridaGMS17}. 
Here we prove their conjecture true by proving that, if $G$ contains (resp.~does not contain) the cycle $C_4$ on four vertices as a subgraph, 
then $(G,k)$ is a \textsc{YES} instance of \ML\ if and only if $k\geq 2n-4$ (resp.~$k\geq 2n-3$). 
The proof is done via a reduction to the gossip problem~\cite{Bumby1981Problem} (for a survey on gossiping see also~\cite{Hedetniemi1988Gossiping}). 

The related problem of achieving temporal connectivity by assigning to every edge of the graph at most one time-label, has been studied by G\"obel et al.~\cite{GobelCV91}, where the relationship with the gossip problem has also been drawn. Contrary to \ML, this problem is NP-hard~\cite{GobelCV91}. 
That is, the possibility of assigning two or more labels to an edge makes the problem computationally much easier. 
Indeed, in a $C_4$-free graph with $n$ vertices, an optimal solution to \ML\ consists in assigning in total $2n-3$ time-labels to the $n-1$ edges of a spanning tree. In such a solution, one of these $n-1$ edges receives one time-label, while each of the remaining $n-2$ edges receives two time-labels. Similarly, when the graph contains a $C_4$, it suffices to span the graph with four trees tooted at the vertices of the $C_4$, 
where each of the edges of the $C_4$ receives one time-label and each edge of the four trees receives two labels. That is, a graph containing a $C_4$ can be temporally connected using $2n-4$ time-labels.

In the gossip problem we have $n$ agents from a set $A$. At the beginning, every agent $x \in A$ holds its own secret. 
The goal is that each agent eventually learns the secret of every other agent. 
This is done by producing a sequence of unordered pairs $(x,y)$, where $x,y \in A$ and each such pair represents one phone call between the agents involved, during which the two agents exchange all the secrets they currently know.

The above gossip problem is naturally connected to \ML. The only difference between the two problems is that, in gossip, 
all calls are non-concurrent, while in \ML\ we allow concurrent temporal edges, \ie two or more edges can appear at the same time slot $t$. 
Therefore, in order to transfer the known results from gossip to \ML, it suffices to prove that in \ML\ we can equivalently consider solutions with non-concurrent edges (see~\cref{lemma:gossip-tgGeneral}).

From the set of agents $A$ and a sequence of calls $\mathcal{C}=c(1), c(2), \dots , c(m)$
we build a temporal graph $\mathcal{G}_\mathcal{C} = (G,\lambda)$ by the following procedure.
For every agent $x \in A$ we create a vertex $v_x \in V(G)$.
Every phone call $c(i)$ between two agents $x,y$ gives rise to a time edge $(v_x v_y, i)$ of $\mathcal{G}_\mathcal{C}$.
Therefore the labeling $\lambda$ is defined by the sequence of phone calls. 
Since no two calls are concurrent, we can order them linearly: for every $1\leq i\leq m$, the phone call $c(i)$ gives the label $i$ to the edge between the two agents involved.

\begin{observation}
	\label{obs:gossiptg1}
	If the sequence $c(1), c(2), \dots, c(m)$ of $m$ phone calls results in all agents knowing all secrets, then the above construction produces a temporally connected temporal graph $\mathcal{G}_\mathcal{C}=(G,\lambda)$ with $|\lambda|=m$.
\end{observation}


Now note that the temporal graph $\mathcal{G}_\mathcal{C}$ produced by the above procedure has the special property that, for every time-label $t = 1,2,\ldots,m$, there exists exactly one edge labeled with $t$. 
In the next lemma we prove the reverse statement of Observation~\ref{obs:gossiptg1}.

\begin{lemma}
	\label{lemma:gossip-tgGeneral}
	Let $(G,\lambda)$ be an arbitrary temporally connected temporal graph with $|\lambda|=m$ time-labels in total. 
	Then there exists a sequence $c(1), c(2), \dots, c(m)$ of $m$ phone calls that results in all agents knowing all secrets.
\end{lemma}

\begin{proof}
	Let $(G,\lambda)$ be an arbitrary temporally connected temporal graph. W.l.o.g.~we may assume that, for every $t=1,2,\ldots,\alpha(G,\lambda)$, there exists at least one edge $e$ such that $t\in \lambda(e)$. Indeed, if such an edge does not exist in $(G,\lambda)$, we can replace in $(G,\lambda)$ every label $t'>t$ by $t'-1$, thus obtaining another temporally connected graph with a smaller age. 
	
	Now we proceed as follows. Let $t\in \{1,2,\ldots,\alpha(G,\lambda)\}$ be an arbitrary time step within the lifetime of $(G,\lambda)$, and let $\{e_{i_k}\}_{k=1}^{t}$ be the edges of $G$ such that $t\in \lambda(e_{i_k})$. Let $\varepsilon = \frac{1}{2t}$. For every $k=1,\ldots,t$, we replace the label $t$ of the edge $e_{i_k}$ by the label $t+k\varepsilon$. Finally, we normalize the new time-labels of the edges of $G$ such that they become the distinct consecutive natural numbers from 1 to $m$ (since $|\lambda|=m$ by the assumption of the lemma). Denote the resulting temporal graph by $(G,\lambda')$. Note that every temporal path in $(G,\lambda)$ corresponds to a temporal path in $(G,\lambda')$ with the same sequence of edges, and vice versa.
	
	Finally we create the required sequence of phone calls as follows: for every $i=1,2,\ldots,m$, if $(G,\lambda')$ contains the edge $e$ with time-label $i$, we add a phone call $c(i)$ between the two endpoints of the edge $e$. 
	Since both $(G,\lambda)$ aqnd $(G,\lambda')$ are temporally connected, it follows that after the sequence $c(1), c(2), \dots, c(m)$ of calls results in every agent knowing every secret. This completes the proof.
\end{proof}

\medskip

Now denote with $f(n)$ the minimum number of calls needed to complete gossiping among a set $A$ of $n$ agents, where a specific set of pairs of vertices $B\subseteq A \times A$ are allowed to make a direct call between each other. Let $G_0=(A,B)$ be the (static) graph having the agents in $A$ as vertices and the pairs of $B$ among them as edges. 
Then it is known by~\cite{Bumby1981Problem} that, if $G_0$ contains a $C_4$ as a subgraph then $f(n)=2n-4$, while otherwise $f(n)=2n-3$. Therefore the next theorem follows by Observation~\ref{obs:gossiptg1} and~\cref{lemma:gossip-tgGeneral} and by the results of~\cite{Bumby1981Problem}.

\begin{theorem}
	\label{thm:optLabGossipC4}
	Let $G=(V,E)$ be a connected graph.
	Then the smallest $k\in \mathbb{N}$ for which $(G,k)$ is a \textsc{YES} instance of \ML\ is:
	\begin{equation*}
		k = 
		\begin{cases}
			2n - 4, & \text{if $G$ contains $C_4$ as a subgraph,} \\
			2n - 3, & \text{otherwise.}
		\end{cases}
	\end{equation*}
\end{theorem}

\subsection{A polynomial-time algorithm for directed acyclic graphs}\label{DAG-subsec}

As a second warm-up, we show that the minimization analogues of \ML\ and \MAL\ on directed acyclic graphs (DAGs) are solvable in polynomial time. 
More specifically, for the minimization analogue of \ML\ we provide an algorithm which, 
given a DAG $G = (V,A)$ with diameter $d_G$, computes a temporal labeling function $\lambda$ which assigns the smallest possible number of time-labels on the arcs of $G$ with the following property: 
for every two vertices $u,v\in V$, there exists a directed temporal path from $u$ to $v$ in $(G,\lambda)$ if and only if there exists a directed path from $u$ to $v$ in~$G$. 
Moreover, the age $\alpha(G,\lambda)$ of the resulting temporal graph is \emph{equal} to $d_G$. 
Therefore, this immediately implies a polynomial-time algorithm for the minimization analogue of \MAL\ on DAGs. 
For notation uniformity, we call these minimization problems \ML$_{directed}$ and \MAL$_{directed}$, respectively. 
First we define a \emph{canonical layering} of a DAG, which is useful for our algorithm.

\begin{definition}
	Let $G = (V,A)$ be a DAG with $n$ vertices, $m$ arcs, and diameter $d$. 
	A partition $L_0, L_1, L_2, \dots, L_d$ of $V$ into $d+1$ sets is a \emph{canonical layering} of $G$ if, for every $0\leq i\leq d$, the set 
	$L_i$ contains all the source vertices in the induced subgraph $G_i := G[\{L_i, L_{i+1}, \dots , L_{d}\}]$.
\end{definition}

An example of a canonical layering of a DAG $G$ is illustrated in \cref{fig:canonical-layering}.

\begin{figure}[!htb]
	\centering
	\includegraphics[width=0.5\linewidth]{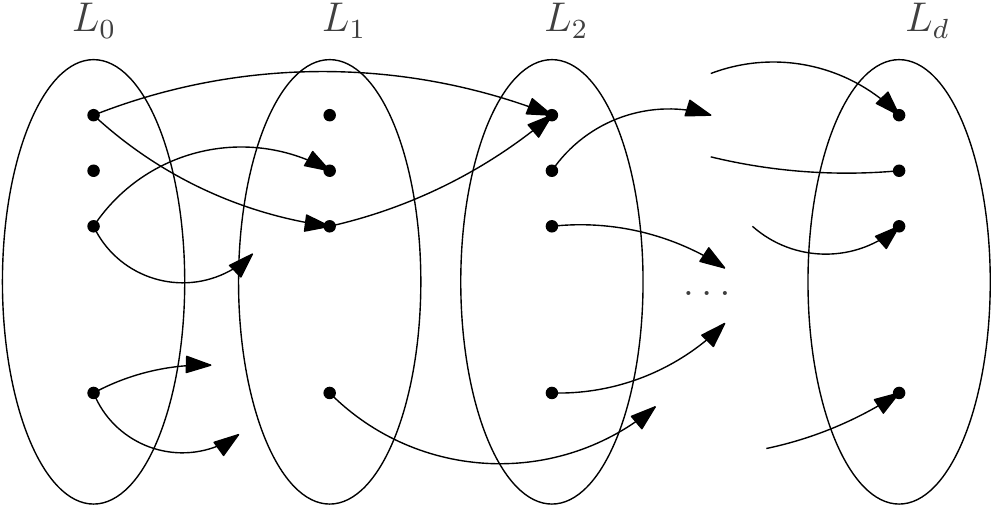}
	\caption{Example of a canonical layering.}
	\label{fig:canonical-layering}
\end{figure}

\begin{lemma}
	\label{canonical-alg-lem}
	Let $G=(V,E)$ be a DAG with $n$ vertices and $m$ arcs. We can produce the canonical layering of $G$ in linear $O(n+m)$ time.
\end{lemma}

\begin{proof}
	First we initialize an auxiliary vertex subset $S=\emptyset$ and a counter $s_v=0$ for every vertex $v$.
	We start by computing the vertices of $L_0$ in $O(n+m)$ time by just visiting all vertices and arcs of $G$; $L_0$ contains all vertices $u$ such that $N^-(u)=\emptyset$. 
	Now, for every $i\geq 0$ we proceed as follows. First we set $S=\emptyset$. Then, for every arc $(u,v)$, where $u\in L_i$, 
	we add $v$ to $S$ and we increase the counter $s_v$ by 1. 
	Then we set $L_{i+1} = \{v\in S : s_v = |N^-(v)|\}$. 
	Before we continue to the next iteration $i+1$, we reset the set $S$ to be $\emptyset$, 
	and we iterate until we reach all vertices of $G$, \ie until we add every vertex $u$ to one of the sets $L_0, L_1, \ldots, L_d$. 
	
	It is easy to check that the above procedure is correct, as at every iteration $i+1$ (where $i\geq 0$) we include to $L_i$ all vertices 
	$v$ which have zero in-degree in the graph induced by the vertices in $V\setminus \bigcup_{k=1}^{i}L_k$. 
	Furthermore, the running time is clearly $O(n+m)$ as we visit each vertex and arc a constant number of times.
\end{proof}

\medskip

The following observations will be useful when considering the canonical layering.
\begin{observation}
	\label{obs1}
	Each layer $L_i$ is an independent set in $G$.
\end{observation}

\begin{observation}
	\label{obs2}
	For every $i=0,1,\ldots,d-1$ and every $u \in L_i$, there exists an arc $(u,v)\in A$ such that $v \in L_{i+1}$.
\end{observation}

\begin{observation}
	\label{obs3}
	For every arc $(u,v)\in A$, where $u\in L_i$ and $v\in L_{i+1}$ for some $i\in \{0,1,\ldots,d-1\}$, 
	there is no directed path of length two or more from $u$ to $v$ in $G$.
\end{observation}

We use the canonical layering to prove the following result.

\begin{theorem}\label{thmDAGpoly}
	Let $G=(V,E)$ be a DAG with $n$ vertices and $m$ arcs. Then \ML$_{directed}(G)$ and \MAL$_{directed}(G)$ can be both computed in $O(n(n+m))$ time.
\end{theorem}

\begin{proof}
	For the purposes of simplicity of the proof, we denote by $\kappa(G)$ the optimum value of \ML$_{directed}$ with the DAG $G$ as its input.
	First we calculate the canonical layering $L_0, L_1, \dots, L_d$ of $G$ in $O(n+m)$ time by \cref{canonical-alg-lem}. 
	For simplicity of the presentation, denote by $G_v$ the induced subgraph of $G$ that contains $v$ and all vertices that are reachable by $v$ in $G$ with a directed path. Let $d_v$ be the diameter of $G_v$; note that $d_v$ is the length of the longest shortest directed path in $G$ that starts at $v$. 
	For every vertex $u\in V$, we define the set $L_0^u = \{u\}$ and we initialize the set $S_u=N^+(u)$. Then, similarly to the proof of \cref{canonical-alg-lem}, we iterate over all vertices $v\in S_u=N^+(u)$ and over all vertices $w\in N^+(v)$. Whenever we encounter a vertex $w\in N^+(v) \cap N^+(u)$, we remove $v$ from $S_u$. At the end of this procedure, the set $S_u$ contains exactly those vertices of $v\in N^+(u)$, for which there is no directed path of length two or more from $u$ to $v$ in $G$. The above procedure can be completed in $O(n(n+m))$ time, as for every vertex $u$, we iterate at most over all arcs in $G$ a constant number of times.
	
	Now we define the labeling $\lambda$ of $G$ as follows: Every arc $(u,v) \in A$, where $u \in L_{i}$, $v \in L_{j}$, and $v\in S_u$, gets the label $\lambda((u,v))=j$. 
	Note here that $1\leq \lambda((u,v))\leq d$ for every arc of $G$, and thus the age $\alpha(G,\lambda)$ of the resulting temporal graph is equal to the diameter $d$ of $G$. 
	We will prove that $|\lambda|=\kappa(G)$. To prove that $|\lambda|\leq \kappa(G)$, it suffices to show that 
	every label of $\lambda$ must participate in every temporal labeling of $G$ which preserves temporal reachability. 
	In fact, this is true as the only arcs of $G$, which have a label in $\lambda$, are those arcs $(u,v)$ such that there is no other directed path from $u$ to $v$. That is, in order to preserve temporal reachability, we need to assign at least one label to all these arcs. 
	
	Conversely, to prove that $|\lambda|\geq \kappa(G)$, it suffices to show that $\lambda$ preserves all temporal reachabilities. 
	For this, observe first that, every directed path $P=(a,\ldots,b)$ in $G$ can be transformed to a directed path $P'=(a,\ldots, b)$ such that, for every arc $(u,v)$ in $P'$, there is no other directed path from $u$ to $v$ in $G$ apart from the arc $(u,v)$ (\ie there is no ``shortcut'' from $u$ to $v$ in $G$). Therefore, since every arc in $P'$ is assigned a label in $\lambda$ and these labels are increasing along $P'$, it follows that $\lambda$ preserves all temporal reachabilities, and thus $|\lambda|\geq \kappa(G)$. Summarizing, $|\lambda|= \kappa(G)$ and the labeling $\lambda$ can be computed in $O(n(n+m))$ time.
	
	Finally, since $\alpha(G,\lambda)=d$, the obtained optimum labeling for \ML\ is also an optimum labeling for \MAL\ (provided that the upper bound $a$ in the input of \MAL\ is at least $d$).
\end{proof}

\section{\MAL\ is NP-complete}\label{MAL-NP-complete-sec}

In this section we prove that it is \NP-hard to determine the number of labels in an optimal labeling of a static, undirected graph $G$, where the age, \ie the maximum label used, is not larger than the diameter of the input graph.

To prove this we provide a reduction from the NP-hard problem \MAXSAT\ (or \MAXSATshort\ for short).
This is a special case of the classical Boolean satisfiability problem,
where the input formula $\phi$ consists of the conjunction of \emph{monotone} \textsc{XOR} clauses of the form $(x_i \oplus x_j)$, \ie variables $x_i, x_j$ are non-negated.
If each variable appears in exactly $r$ clauses, then $\phi$ is called a \emph{monotone} \textsc{Max XOR($r$)} formula.
A clause $(x_i \oplus x_j)$ is \emph{\textsc{XOR}-satisfied} (or simply \emph{satisfied}) if and only if $x_i \neq x_j$.
In \textsc{Monotone Max XOR($r$)} we are trying to find a truth assignment $\tau$ of $\phi$ which satisfies the maximum number of clauses.
As it can be easily checked, \MAXSATshort\ encodes the problem \textsc{Max-Cut} on cubic graphs, which is known to be NP-hard \cite{Alimonti1997Hardness}. Therefore we conclude the following.

\begin{theorem}[\hspace{-0.0001cm}\cite{Alimonti1997Hardness}]\label{maxsat-hard-thm}
	\MAXSATshort\ is NP-hard.
\end{theorem}

Now we explain our reduction from \MAXSATshort\ to the problem \textsc{Minimum Aged Labeling (MAL)}, where the input static graph $G$ is undirected and the desired age of the output temporal graph is the diameter $d$ of $G$ .
Let $\phi$ be a monotone \textsc{Max XOR($3$)} formula with $n$ variables $x_1, x_2, \dots , x_n$ and $m$ clauses $C_1, C_2, \dots , C_m$. 
Note that $m = \frac{3}{2} n$, since each variable appears in exactly $3$ clauses. 
From $\phi$ we construct a static undirected graph $G_\phi$ with diameter $d=10$, 
and prove that there exists a truth assignment $\tau$	which satisfies at least $k$ clauses in $\phi$, if and only if
there exists a labeling $\lambda_\phi$ of $G_\phi$, with $|\lambda_\phi|\leq \frac{13}{2}n^2+ \frac{99}{2}n-8k$ labels and with age $\alpha(G,\lambda)\leq 10$. 

\paragraph*{High-level construction}
For each variable $x_i$, $1\leq i \leq n$, we construct a variable gadget $X_i$ that consists of a ``starting'' vertex $s_i$ and three ``ending'' vertices $t_i^\ell$ (for $\ell \in \{1,2,3\}$); these ending vertices correspond to the appearances of $x_i$ in three clauses of $\phi$.
In an optimum labeling $\lambda(\phi)$, in each variable gadget there are exactly two labelings that temporally connect starting and ending vertices, which correspond to the \textsc{True} or \textsc{False} truth assignment of the variable in the input formula $\phi$.
For every clause $(x_i \oplus x_j)$ we identifying corresponding ending vertices of $X_i$ and $X_j$ (as well as some other auxiliary vertices and edges).
Whenever $(x_i \oplus x_j)$ is satisfied by a truth assignment of $\phi$, the labels of the common edges of $X_i$ and $X_j$ in an optimum labeling coincide (thus using few labels); otherwise we need additional labels for the common edges of $X_i$ and $X_j$. 

\paragraph*{\boldmath Detailed construction of $G_\phi$}
For each variable $x_i$ from $\phi$ we create a variable gadget $X_i$, that consists of a \emph{base} $BX_i$ on $11$ vertices, 
$BX_i = \{s_i, a_i, b_i, c_i, d_i, e_i, \overline{a_i}, \overline{b_i}, \overline{c_i}, \overline{d_i}, \overline{e_i}\}$, 
and three \emph{forks} $F^1X_i, F^2X_i,F^3X_i$, 
each on $9$ vertices,
$F^\ell X_i = \{ t^\ell_i, f^\ell_i, g^\ell_i, h^\ell_i, m^\ell_i, \overline{f_i}^\ell, \overline{g_i}^\ell, \overline{h_i}^\ell, \overline{m_i}^\ell \}$, where $\ell \in \{1,2,3\}$.
Vertices in the base $BX_i$ are connected in the following way: 
there are two paths of length $5$: $s_i a_i b_i c_i d_i e_i$ and $s_i \overline{a_i} \overline{b_i} \overline{c_i} \overline{d_i} \overline{e_i}$, 
and $5$ extra edges of form $y_i \overline{y_i}$, where $y \in \{a, b,c,d,e\}$.
Vertices in each fork $F^\ell X_i$ (where $\ell \in \{1,2,3\}$) are connected in the following way:
there are two paths of length $4$: $t^\ell_i m^\ell_i h^\ell_i g^\ell_i f^\ell_i$ and
$t^\ell_i  \overline{m_i}^\ell \overline{h_i}^\ell \overline{g_i}^\ell \overline{f_i}^\ell$, 
and $4$ extra edges of form $y_i \overline{y_i}^\ell$, where $y \in \{m,h,g,f\}$.
The base $BX_i$ of the variable gadget $X_i$ is connected to each of the three forks $F^\ell X_i$ via two edges $e_i f_i^\ell$ and $\overline{e_i} \overline{f_i}^\ell$, where $\ell \in \{1, 2, 3\}$.
For an illustration see~\cref{fig:npdstatic-variablegadget}.

For an easier analysis we fix the following notation.
The vertex $s_i \in B X_i$ is called a \emph{start vertex} of $X_i$, vertices $t_i^\ell$ ($\ell \in \{1,2,3\}$) are called \emph{ending vertices} of $X_i$,
a path connecting $s_i, t_i^\ell$ that passes through vertices $a_i b_i c_i d_i e_i f_i^\ell g_i^\ell h_i^\ell m_i^\ell$ (resp.~$\overline{a_i} \overline{b_i} \dots \overline{m_i}^\ell$) is called the \emph{left} (resp.~\emph{right}) $s_i, t_i^\ell$-path.
The left (resp.~right) $s_i, t_i^\ell$-path is a disjoint union of the left (resp.~right) path on vertices of the base $BX_i$ of $X_i$, 
an edge of form $e_i f_i^\ell$ (resp.~$\overline{e_i} \overline{f_i}^\ell$) called the left (resp.~right) \emph{bridge edge} and
the left (resp.~right) path on vertices of the $\ell$-th fork $F^\ell X_i$ of $X_i$.
The edges $y_i \overline{y_i}$, where $y \in \{a, b, c, d, e, f^\ell, g^\ell, h^\ell, m^\ell\}$, $\ell \in \{1,2,3\}$, are called \emph{connecting edges}.

\begin{figure}[htb]
	\centering
	\includegraphics[width=0.6\linewidth]{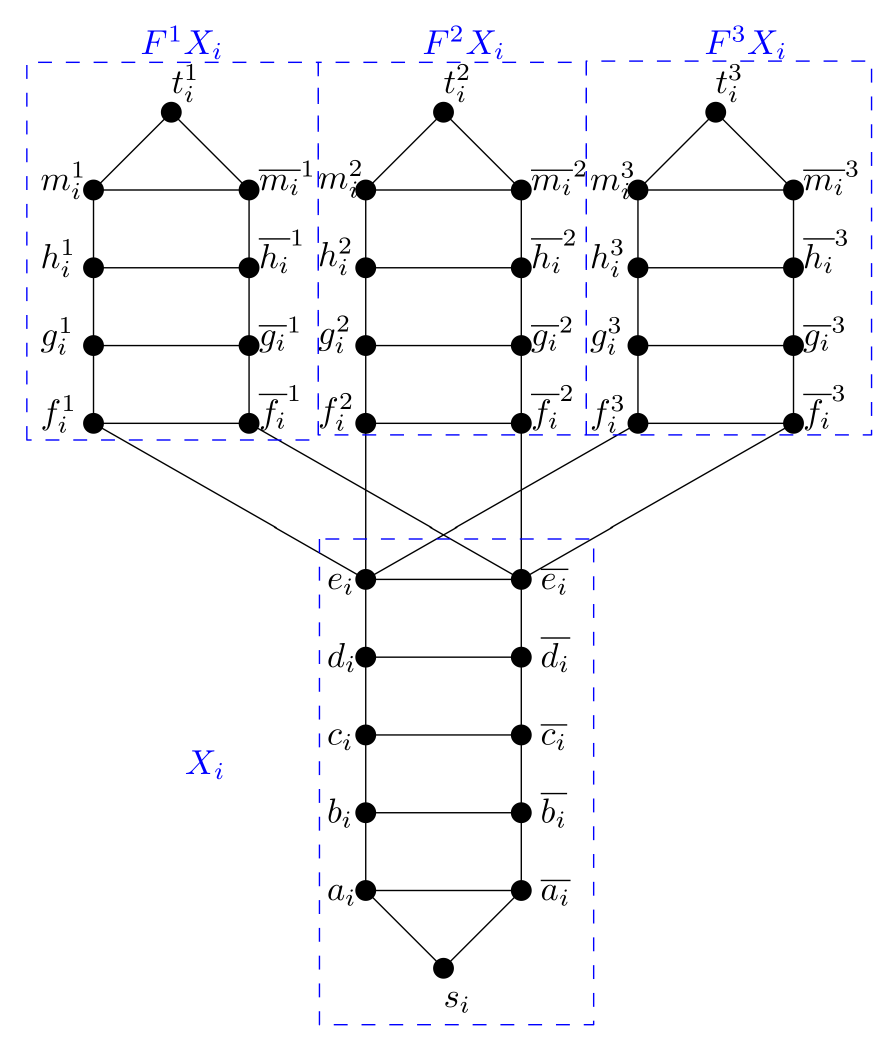}
	\caption{An example of a variable gadget $X_i$ in $G_\phi$, corresponding to the variable $x_i $ from $\phi$.}
	\label{fig:npdstatic-variablegadget}
\end{figure}

\paragraph*{Connecting variable gadgets}
There are two ways in which we connect two variable gadgets, depending whether they appear in the same clause in $\phi$ or not.
\begin{enumerate}
	\item Two variables $x_i, x_j$ do not appear in any clause together.
	In this case
	we add the following edges between the variable gadgets $X_i$ and $X_j$:
	\begin{itemize}
		\item from $e_i$ (resp.~$\overline{e_i}$) to $f_j^{\ell'}$ and $\overline{f_j}^{\ell'}$, where $\ell' \in \{1,2,3\}$,
		\item from $e_j$ (resp.~$\overline{e_j}$) to $f_i^\ell$ and $\overline{f_i}^\ell$, where $\ell \in \{1,2,3\}$,
		\item from $d_i$ (resp.~$\overline{d_i}$) to $d_j$ and $\overline{d_j}$.
	\end{itemize}
	We call these edges the \emph{variable edges}.
	For an illustration see~\cref{fig:npdstatic-connectingvariablegadget}.
	
	\begin{figure}[t]
		\centering
		\includegraphics[width=0.85\linewidth]{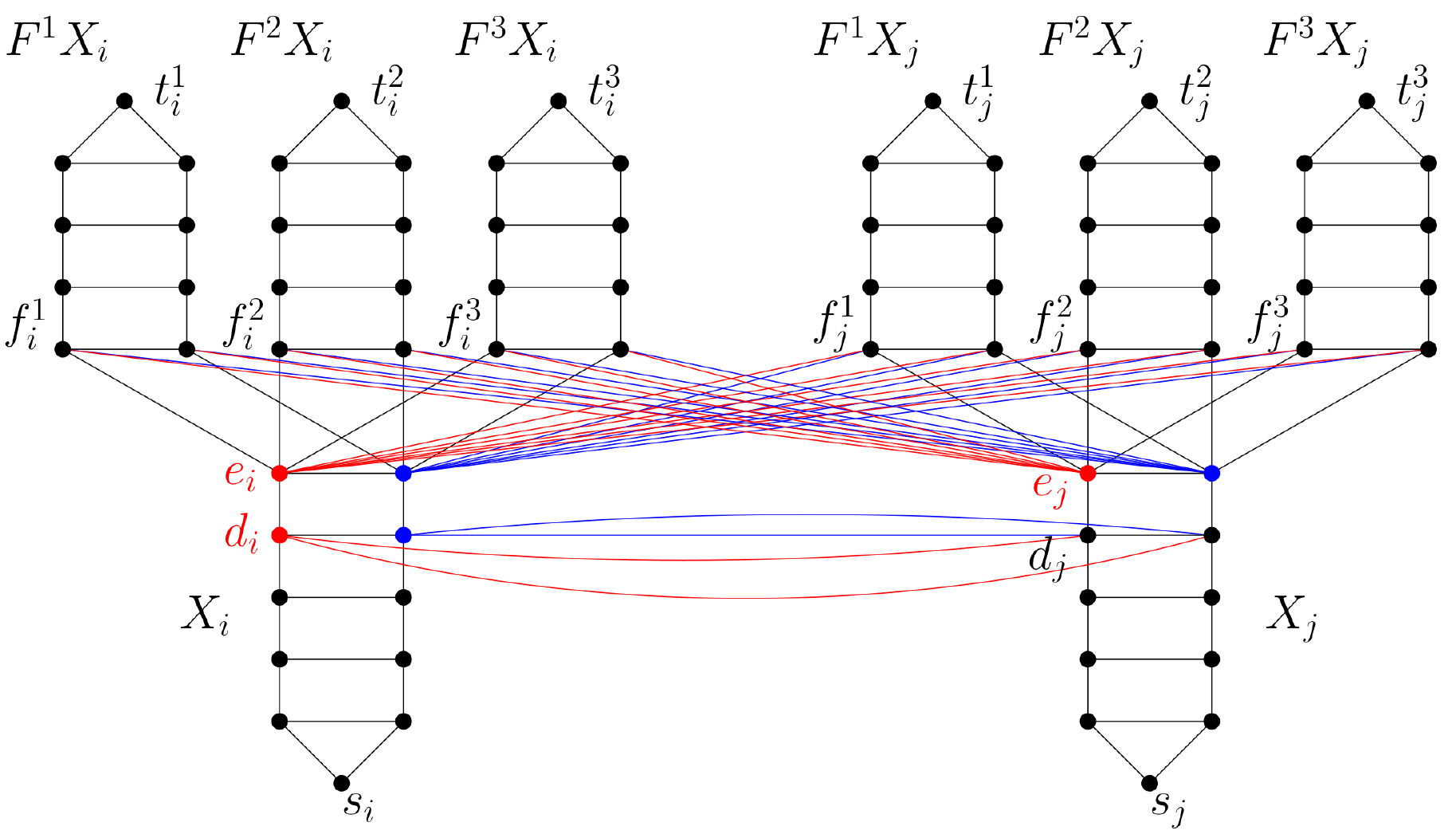}
		\caption{An example of two non-intersecting variable gadgets and variable edges among them.}
		\label{fig:npdstatic-connectingvariablegadget}
	\end{figure}
	
	\item Let $C = (x_i \oplus x_j)$ be a clause of $\phi$, that contains the $r$-th appearance of the variable $x_i$ and $r'$-th appearance of the variable $x_j$.
	In this case we identify the $r$-th fork $F^r X_i$ of $X_i$ with the $r'$-th fork $F^{r'} X_j$ of $X_j$ in the following way:
	\begin{itemize}
		\item $t_i^r = t_j^{r'}$,
		\item $\{f_i^r, g_i^r, h_i^r, m_i^r\} = \{ \overline{f_j}^{r'}, \overline{g_j}^{r'}, \overline{h_j}^{r'}, \overline{m_j}^{r'} \}$ respectively, and
		\item $\{\overline{f_i}^{r}, \overline{g_i}^{r}, \overline{h_i}^{r}, \overline{m_i}^{r} \} =
		\{f_j^{r'}, g_j^{r'}, h_j^{r'}, m_j^{r'}\}$ respectively.
	\end{itemize}
	Besides that we add the following edges between the variable gadgets $X_i$ and $X_j$:
	\begin{itemize}
		\item from $e_i$ (resp.~$\overline{e_i}$) to $f_j^{\ell'}$ and $\overline{f_j}^{\ell'}$, where $\ell' \in \{1,2,3\} \setminus \{r'\}$,
		\item from $e_j$ (resp.~$\overline{e_j}$) to $f_i^\ell$ and $\overline{f_i}^\ell$, where $\ell \in \{1,2,3\} \setminus \{r\}$,
		\item from $d_i$ (resp.~$\overline{d_i}$) to $d_j$ and $\overline{d_j}$.
	\end{itemize}
	For an illustration see~\cref{fig:npdstatic-clausegadget}.

	\begin{figure}[!tbh]
		\centering
		\includegraphics[width=0.65\linewidth]{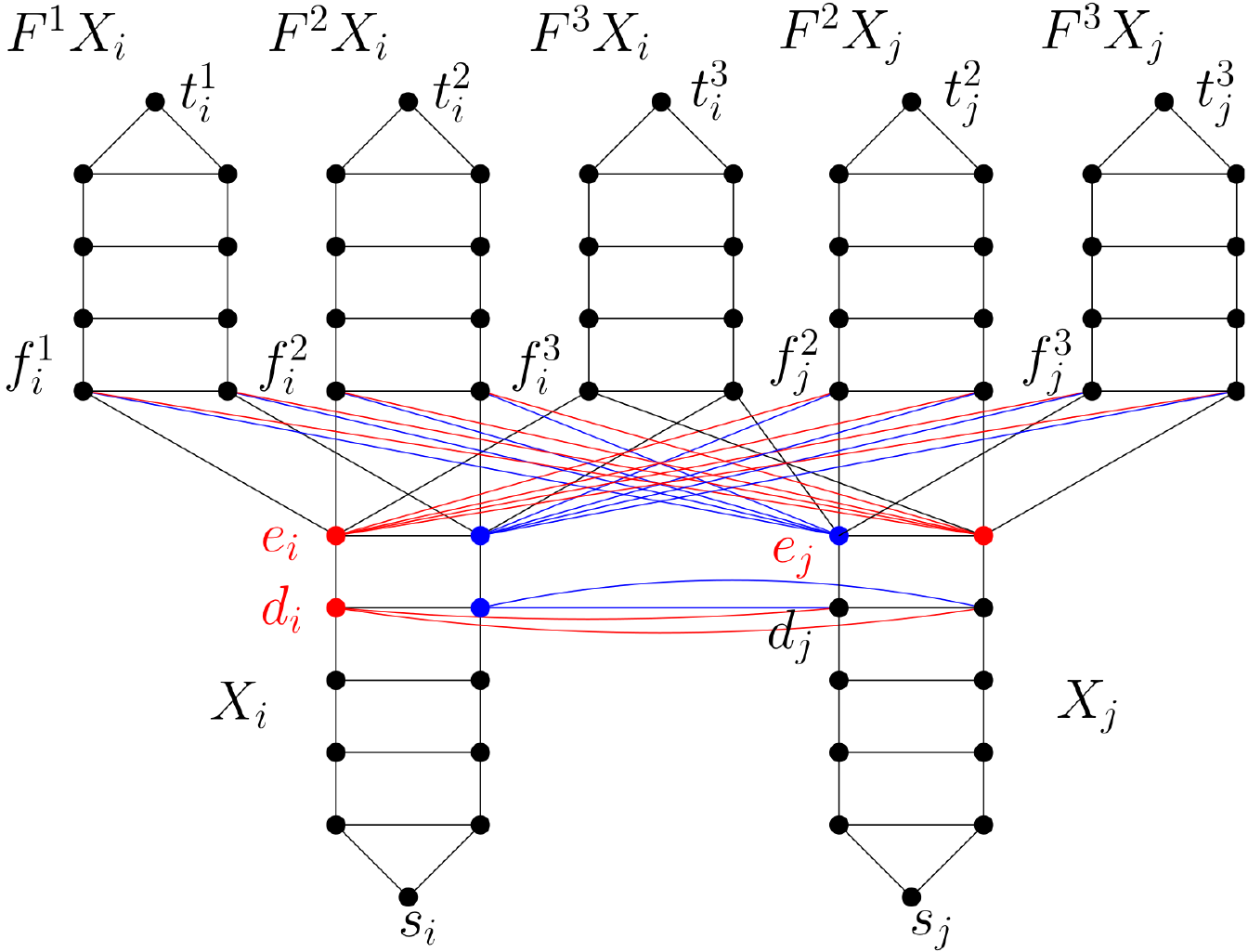}
		\caption{An example of two intersecting variable gadgets $X_i, X_j$ corresponding to variables $x_i,x_j$, that appear together in some clause in $\phi$, where it is the third appearance of $x_i$ and the first appearance of $x_j$.}
		\label{fig:npdstatic-clausegadget}
	\end{figure}
	
\end{enumerate}
This finishes the construction of $G_\phi$. 
Before continuing with the reduction, we prove the following structural property of $G_\phi$.

\begin{lemma}\label{diameter-G-phi-lem}
	The diameter $d_\phi$ of $G_\phi$ is $10$.
\end{lemma}

\begin{proof}
	We prove this in two steps. 
	First we show that the diameter of any variable gadget is $10$ and then show that the diameter does not increase, 
	when the variable edges are introduced, \ie vertices in any two variable gadgets are at most $10$ apart.
	
	Let us start with fixing a variable gadget $X_i$. 
	A path from the starting vertex $s_i$ to any ending vertex  $t_i^\ell$ ($\ell \in \{1,2,3\}$) has to go through at least one of the vertices from $\{a_i, \overline{a_i}\}$, then through at least one of the vertices from $\{b_i, \overline{b_i}\}$,
	then through $\{c_i, \overline{c_i}\}$, $\{d_i, \overline{d_i}\}$,
	$\{e_i, \overline{e_i}\}$, $\{f^\ell_i, \overline{f_i}^\ell\}$, $\{g^\ell_i, \overline{g_i}^\ell\}$, $\{h^\ell_i, \overline{h_i}^\ell\}$ and finally through $\{m^\ell_i, \overline{m_i}^ \ell\}$, before reaching the ending vertex.
	The shortest $s_i,t_i^ \ell$ path will go through exactly one vertex from each of the above sets. 
	Therefore it is of length $10$.
	Because of the construction of $X_i$, there are exactly two $s_i, t_i^\ell$ paths of length $10$, 
	which are edge and vertex disjoint, as they share only the starting and ending vertices.
	One of this paths uses the vertices $a_i, b_i, c_i, d_i, e_i, f_i^\ell, g_i^\ell, h_i^\ell, m_i^\ell$ (\ie the left path) 
	and the other uses vertices $\overline{a_i}, \overline{b_i}, \dots , \overline{m_i}^\ell$ (\ie the right path).
	A path between any two ending vertices $t_i^\ell, t_i^{\ell'}$ (where $\ell, \ell' \in \{1,2,3\}$ and $\ell \neq \ell'$),
	has to go through the following sets of vertices,
	$\{m_i^\ell, \overline{m_i}^\ell\}$,$\{m_i^{\ell'}, \overline{m_i}^{\ell'}\}$, 
	$\{h_i^\ell, \overline{h_i}^\ell\}$,$\{h_i^{\ell'}, \overline{h_i}^{\ell'}\}$,
	$\{g_i^\ell, \overline{g_i}^\ell\}$,$\{g_i^{\ell'}, \overline{g_i}^{\ell'}\}$,
	$\{f_i^\ell, \overline{f_i}^\ell\}$,$\{f_i^{\ell'}, \overline{f_i}^{\ell'}\}$,
	$\{e_i, \overline{e_i}\}$.
	Similarly as before, the shortest path uses exactly one vertex from each set and is of size $10$.
	Even more, 
	there are exactly two $t_i^\ell, t_i^{\ell'}$ paths of length $10$.
	They are edge and vertex disjoint, as they share only the starting and ending vertices.
	One of this paths uses the vertices without the line in the label (\ie the left path) 
	and the other uses vertices with the line in the label (\ie the right path).
	It is not hard to see that the distance between any other vertex  in $X_i$ and starting or ending vertices is at most $9$, 
	as that vertex lies on one of the $s_i, t_i^\ell$ or $t_i^\ell,t_i^{\ell'}$-paths, but it is not an endpoint of it.
	By the similar reasoning there exists a path between any two vertices in $X_i$ (different than $s_i, t_i^\ell$), of distance at most $9$.
	Therefore the diameter of $X_i$ is $10$.
	
	Now let us fix two variable gadgets $X_i, X_j$, that share no fork (\ie $x_i$ and $x_j$ appear in no clause of $\phi$).
	The shortest path from the starting vertex $s_i$ of $X_i$
	to the starting vertex $s_j$ of $X_j$
	has to reach vertex $d_i$ (resp.~$\overline{d_i}$), which is done in $4$ steps,
	from where it connects to either $d_j$ or $\overline{d_j}$, using a variable edge, and continues toward $s_j$, with $4$ edges.
	Therefore, $d(s_i,s_j) = 9$.
	The shortest path connecting vertex $s_i$ with $t_j^{\ell'}$,
	uses one of the vertices $e_i$ or $\overline{e_i}$, that are on the distance $5$ from $s_i$,
	then using one variable edge reaches $f_j^{\ell'}$ or $\overline{f_j}^{\ell'}$, which is on the distance $4$ from the ending vertex $t_j^{\ell'}$.
	Therefore, $d(s_i,t_j^{\ell'}) = 10$, for all $\ell' \in \{1,2,3\}$.
	Lastly, the shortest path between an ending vertex $t_i^\ell$ of $X_i$ and an ending vertex $t_j^{\ell'}$
	uses $4$ edges in the fork $F^\ell X_i$ to reach the vertex $f_i^\ell$ or $\overline{f_i^\ell}$,
	from where it uses a variable edge that connects it to the vertex $e_j$ or $\overline{e_j}$, 
	that is on the distance $5$ from the $t_j^{\ell'}$.
	Therefore, $d(t_i^\ell,t_j^{\ell'}) = 10$, for all $\ell, \ell' \in \{1,2,3\}$.
	It is not hard to see that if two variable gadgets $X_i, X_j$ share a fork the shortest path among any two vertices does not increase.
	
	We proved that the distance among any two vertices in $G_\phi$ is at most $10$ and thus its diameter is $10$.
\end{proof}

\begin{theorem}\label{thm:APXhardnessFirstDirection}
	If OPT$_{\MAXSATshort}(\phi) \geq k$ then OPT$_{\MAL}(G_\phi,d_\phi) \leq \frac{13}{2}n^2+ \frac{99}{2}n-8k$, 
	where $n$ is the number of variables in the formula $\phi$.
\end{theorem}

\begin{proof}
	Let $\tau$ be an optimum truth assignment of $\phi$, \ie a truth assignment that satisfies at least $k$ clauses of $\phi$. 
	We will prove that there exists a temporal labeling $\lambda_\phi$ of $G_\phi$ which uses $|\lambda_\phi| \leq \frac{13}{2}n^2+ \frac{99}{2}n-8k$ labels, such that $(G,\lambda)$ is temporally connected and $\alpha(G,\lambda)=d_\phi=10$. 
	Recall that, since $\phi$ is an instance of \MAXSATshort\ with $n$ variables, it has $m=\frac{3}{2}n$ clauses.
	We build the labeling $\lambda_\phi$ using the following rules. For an illustration see \cref{fig:NP-labeled}.
	\begin{enumerate}
		\item
		If a variable $x_i$ from $\phi$ is set to be \textsc{True} by the truth assignment $\tau$, 
		we label the edges in $X_i$ in the following way:
		\begin{itemize}
			\item all three left $(s_i, t_i^\ell)$-paths, for all $\ell \in \{1,2,3\}$,
			get the labels $1, 2, 3, \dots , 10$, one on each edge,
			\item similarly, all left $(t_i^\ell , s_i)$-paths,
			get the labels $1, 2, 3, \dots , 10$, one on each edge,
			\item all connecting edges (\ie edges of form $y_i \overline{y_i}$, where $y \in \{a, b, c, d, e , f^\ell, g^\ell, h^\ell, m^\ell\}$) get the labels $1$ and $10$.
		\end{itemize}
		If a variable $x_i$ from $\phi$ is set to be \textsc{False} by the truth assignment $\tau$, 
		we label the edges in $X_i$ in the following way:
		\begin{itemize}
			\item all three right $(s_i, t_i^\ell)$-paths, for all $\ell \in \{1,2,3\}$,
			get the labels $1, 2, 3, \dots , 10$, one on each edge,
			\item similarly, all right $(t_i^\ell , s_i)$-paths,
			get the labels $1, 2, 3, \dots , 10$, one on each edge,
			\item all connecting edges 
			get the labels $1$ and $10$.
		\end{itemize}
		Labeling $\lambda_\phi$ uses $10$ labels on the left (resp.~right) path of the base $BX_i$, 
		$10$ labels on the left (resp.~right) path of each fork $F^\ell X_i$, where $\ell \in \{1,2,3\}$ and 
		$10 + 3 \cdot 8$ labels on the connecting edges.
		All in total $\lambda_\phi$ uses $74$ labels on the variable gadget $X_i$.
		
		We still need to prove that there exists a temporal path among any two vertices in $X_i$.
		There is a (unique) temporal path from the starting $s_i$ vertex to all three ending vertices $t_i^\ell$, where $\ell \in \{1,2,3\}$, 
		using left (in case of $x_i$ being \textsc{True}) or right (in case of $x_i$ being \textsc{False}) paths of the base
		$BX_i$ and forks $F^\ell X_i$.
		Similarly it holds for all temporal $(t_i^\ell,s_i)$-paths.
		The temporal path connecting two ending vertices $t_i^{\ell_1}, t_i^{\ell_2}$, 
		uses first the left (in case of $x_i$ being \textsc{True}) or right (in case of $x_i$ being \textsc{False}) path
		of the fork $F^{\ell_1}X_i$ to reach $e_i$ (in case of $x_i$ being \textsc{True}) or $\overline{e_i}$ (in case of $x_i$ being \textsc{False}), using the labels $1$ to $5$,
		and then continues on the left (in case of $x_i$ being \textsc{True}) or right (in case of $x_i$ being \textsc{False})
		path of the $F^{\ell_2}X_i$ from $e_i$ or $\overline{e_i}$ to $t_i^{\ell_2}$, using labels $6$ to $10$.
		Any vertex not on the left (in case of $x_i$ being \textsc{True}) or right (in case of $x_i$ being \textsc{False}) path, 
		can reach the starting vertex or any of the ending vertices, using a connecting edge 
		at time $1$.
		Similarly it hold for the paths in the opposite direction, where the connecting edges
		have the label $10$.
		A temporal path among two vertices not on the left (in case of $x_i$ being \textsc{True}) or right (in case of $x_i$ being \textsc{False}) path uses first 
		a connecting edge
		at time $1$, then a portion of the left (in case of $x_i$ being \textsc{True}) or right (in case of $x_i$ being \textsc{False}) path and again the appropriate connecting edge
		at time $10$.
		This proves that 
		$\lambda_\phi$ on $X_i$
		admits a temporal path among any two vertices in $X_i$.
		
		\item
		If two variable gadgets $X_i$ and $X_j$ do not share a fork, 
		\ie variables $x_i$ and $x_j$ are not in the same clause in $\phi$,
		and are both set to \textsc{True} by $\tau$, then we label the following variable gadgets: 
		\begin{itemize}
			\item the edge $d_i d_j$, connecting the left path of $BX_i$ with the left path of $BX_j$, gets the label $5$,
			\item three edges of the form $e_i f_j ^{\ell'}$ ($\ell ' \in \{1,2,3\}$), that connect the left path of $BX_i$ to left paths of $F^{\ell'}X_j$,
			with the labels $4$ and $6$,
			\item three edges of the form $e_j f_i ^{\ell}$ ($\ell  \in \{1,2,3\}$), that connect the left path of $BX_j$ to left paths of $F^{\ell }X_i$,
			with the labels $4$ and $6$.
		\end{itemize}
		The labeling $\lambda_\phi$ uses $74$ labels for each variable gadget and $13$ labels on $7$ variable edges that connect both variable gadgets.
		Note, the three other combinations ($x_i,x_j$ are both \textsc{False}, one of $x_i, x_j$ is \textsc{True} and the other \textsc{False}) 
		give rise to the labeling $\lambda_\phi$ that uses the same number of labels on both variable gadgets and variable edges,
		where the labeled variable edges are chosen appropriately.
		
		Since labeling variable edges does not change the labeling on each variable gadget, we know that there is still a temporal path among any two vertices from the same variable gadget.
		We need to prove now that there is a temporal path among any two vertices from $X_i$ and $X_j$.
		The edge $d_i d_j$, with the label $5$,
		connects all the vertices from the $BX_i \setminus \{e_i, \overline{e_i}\}$ 
		to the vertices from the 
		$BX_j \setminus \{e_j, \overline{e_j}\}$ 
		and vice versa.
		To go from the starting vertex $s_i$ of $X_i$ to the ending vertex $t_j^\ell$ of $X_j$ we use the following route.
		From $s_i$ to $e_i$ we use the left labeled path on $X_i$ with labels from $1$ to $5$, 
		then the edge $e_i f_j^{\ell'}$ at time $6$ to reach the corresponding fork $F^{ell'} X_j$ of $X_j$ 
		and from $f_j^{\ell'}$ to the ending vertex $t_j^{\ell'}$
		we use the left labeled path of $X_j$ with labels $7$ to $10$.
		This temporal path connects all vertices in the base $BX_i$ to all vertices in the forks $FX_j^{\ell'}$, where $\ell' \in \{1,2,3\}$.
		Similar we obtain temporal paths from vertices in the base $BX_j$ to vertices in the forks $FX_i^\ell$, where $\ell \in \{1,2,3\}$.
		To go from any vertex in the fork $F^\ell X_i$ to any vertex of the $X_j$ we use the following route.
		First, we reach the vertex $f_i^{\ell}$, by the time $4$,
		using the left labeled path of $X_i$.
		Then we use the edge
		$f_i^\ell e_j$ at time $5$.
		Now, by the construction of $\lambda_\phi$ of $X_j$, 
		each vertex in $X_j$ can be reached from $e_j$ from time $5$ to $10$.
		Therefore all vertices from $F^\ell X_i$ can reach any vertex in  $X_j$. 
		This is true for all $\ell \in \{1,2,3\}$.
		Similarly it holds for temporal paths from any vertex in the fork $F^{\ell'} X_j$ ($\ell' \in \{1,2,3\}$) to vertices of the $X_i$.
		The only thing left to show is that the vertices $\{e_i, \overline{e_i}$ can reach all other vertices in $BX_j$.
		This is true as there is a temporal path using the edge $e_i f_j^{\ell'}$ at time $5$
		and then, from $f_j^{\ell'}$ to any vertex in the base $BX_j$, the left labeled path of $BX_j$, that is labeled by $\lambda_\phi$. 
		This is true for all $\ell' \in \{1,2,3\}$.
		Similarly it holds for the temporal paths from $\{e_j, \overline{e_j}\}$ to the vertices in $BX_i$.
		Therefore $\lambda_\phi$ admits a temporal path among any two vertices of variable gadgets, that do not share the fork.
		\item
		If two variable gadgets $X_i$ and $X_j$ share a fork, 
		\ie variables $x_i$ and $x_j$ are in the same clause,
		are both set to \textsc{True}
		and
		$F^rX_i = F^{r'}X_j$,
		then we label the following variable edges:
		\begin{itemize}
			\item the edge $d_i d_j$, connecting the left path of $BX_i$ and $BX_j$, gets the label $5$,
			\item two edges of the form $e_i f_j ^{\ell'}$ ($\ell ' \in \{1,2,3\} \setminus \{r'\}$), that connect the left path of $BX_i$ to left paths of $F^{\ell'}X_j$,
			with the labels $4$ and $6$,
			\item two edges of the form $e_j f_i ^{\ell}$ ($\ell  \in \{1,2,3\} \setminus \{r\}$), that connect the left path of $BX_j$ to left paths of $F^{\ell }X_i$,
			with the labels $4$ and $6$.
		\end{itemize}
		The labeling $\lambda_\phi$ uses $9$ labels on $5$ variable edges that connect both variable gadgets.
		Note, the three other combinations ($x_i, x_j$ are both \textsc{False}, one of $x_i, x_j$ is \textsc{True} and the other \textsc{False}) 
		give rise to the labeling $\lambda_\phi$ that uses the same number of labels on variable edges,
		where the labeled edges are chosen accordingly to the truth values of $x_i$ and $x_j$.
		The only difference is in the labeling of the shared fork $F^rX_i = F^{r'}X_j$.
		There are two possibilities, one when the truth value of $x_i$ and $x_j$ is the same and one when it is different, \ie $x_i = x_j$ or $x_i \neq x_j$.
		\begin{enumerate}[label=\alph*)]
			\item Let us start with the case when $x_i \neq x_j$. 
			Without loss of generality (w.l.o.g.) we can assume that $x_i$ is \textsc{True} and $x_j$ \textsc{False}.
			In the labeling $\lambda_\phi$ we label all left paths in the variable gadget $X_i$ and all right paths in $X_j$.
			By the construction of the graph $G_\phi$ (and the rules of how to identify vertices of the two forks), the left labeling of $F^rX_i$ coincides with the right labeling of $F^{r'}X_j$.
			Therefore $\lambda_\phi$ uses $2 \cdot 74 - 16 = 132$ labels on both variable gadgets.		
			\item Let us now observe the case when $x_i = x_j$.
			W.l.o.g.~we can assume that both variables are \textsc{True}.
			In the labeling $\lambda_\phi$ we label all left paths of both variable gadgets.
			By the construction of the graph $G_\phi$ (and the rules of how to identify vertices of the two forks), the fork $F^rX_i = F^{r'}X_j$
			gets labeled from both sides, \ie all edges in the fork get $2$ labels.
			Therefore $\lambda_\phi$ uses $2 \cdot 74 - 8 = 140$ labels on both variable gadgets.
		\end{enumerate}
		Identifying two forks $F^rX_i = F^{r'}X_j$
		and labeling them using the union of both labelings on each fork,
		preserves temporal paths among all the vertices from $X_i$ and $X_j$.
		This is true as the labeling in each variable is not changed by the labeling in the other variable.
		Among forks that are not in the intersection there are still the variable edges left, 
		which assure that vertices from different variable gadgets can reach them or can be reached by them.
		Therefore the labeling $\lambda_\phi$ admits a temporal path among any two vertices from the variable gadgets $X_i, X_j$, that have a fork in the intersection.
	\end{enumerate}
	
	\begin{figure}[ht!]
		\begin{subfigure}{\textwidth}
			\centering
			\includegraphics[width=\linewidth]{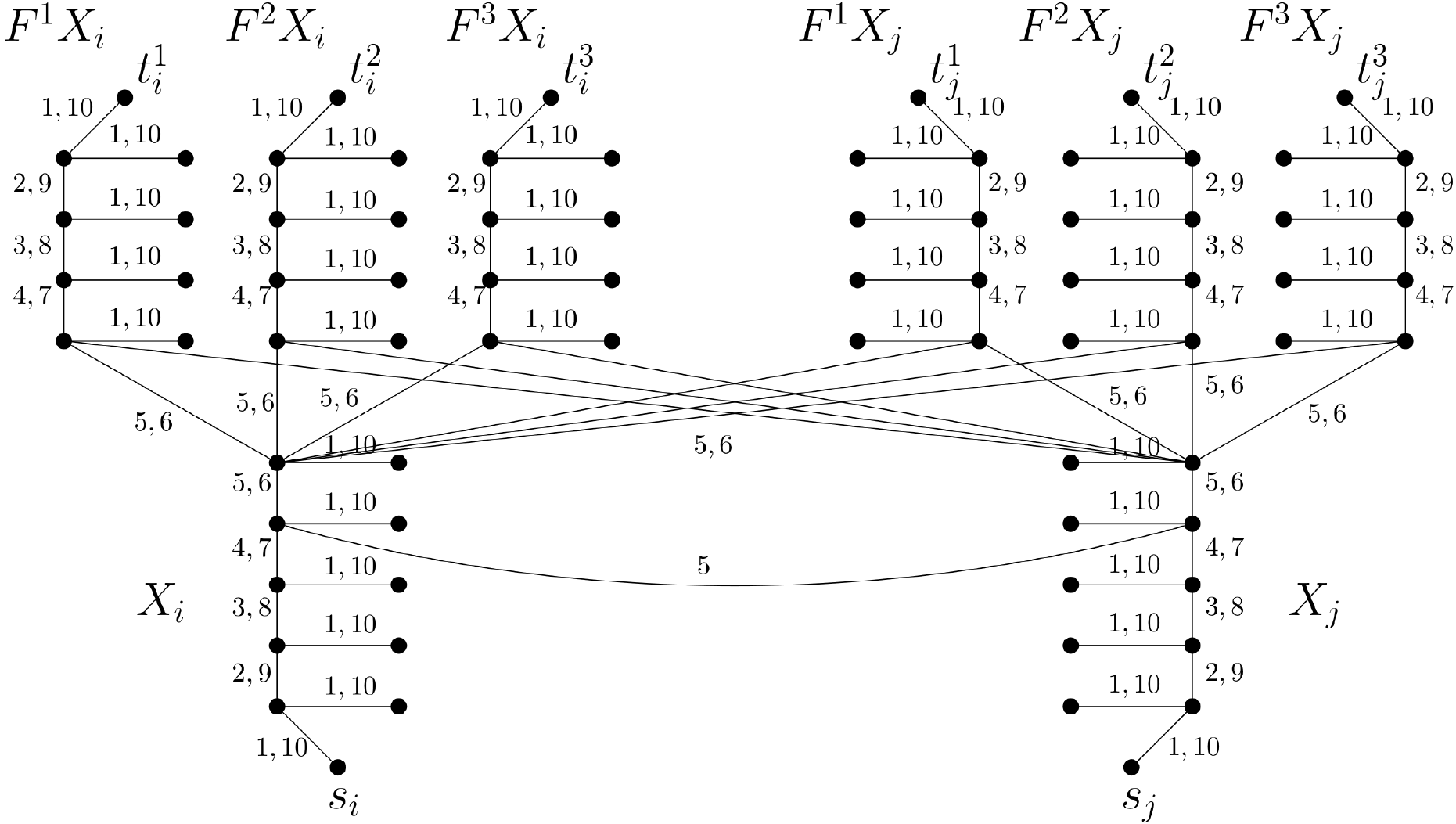}
			\caption{$x_i$ and $x_j$ do not appear together in any clause.\vspace{0,1cm}}
			\label{fig:NP-labeledNonIntersecting}
		\end{subfigure}%
		
		\begin{subfigure}{\textwidth}
			\centering
			\includegraphics[width=0.8\linewidth]{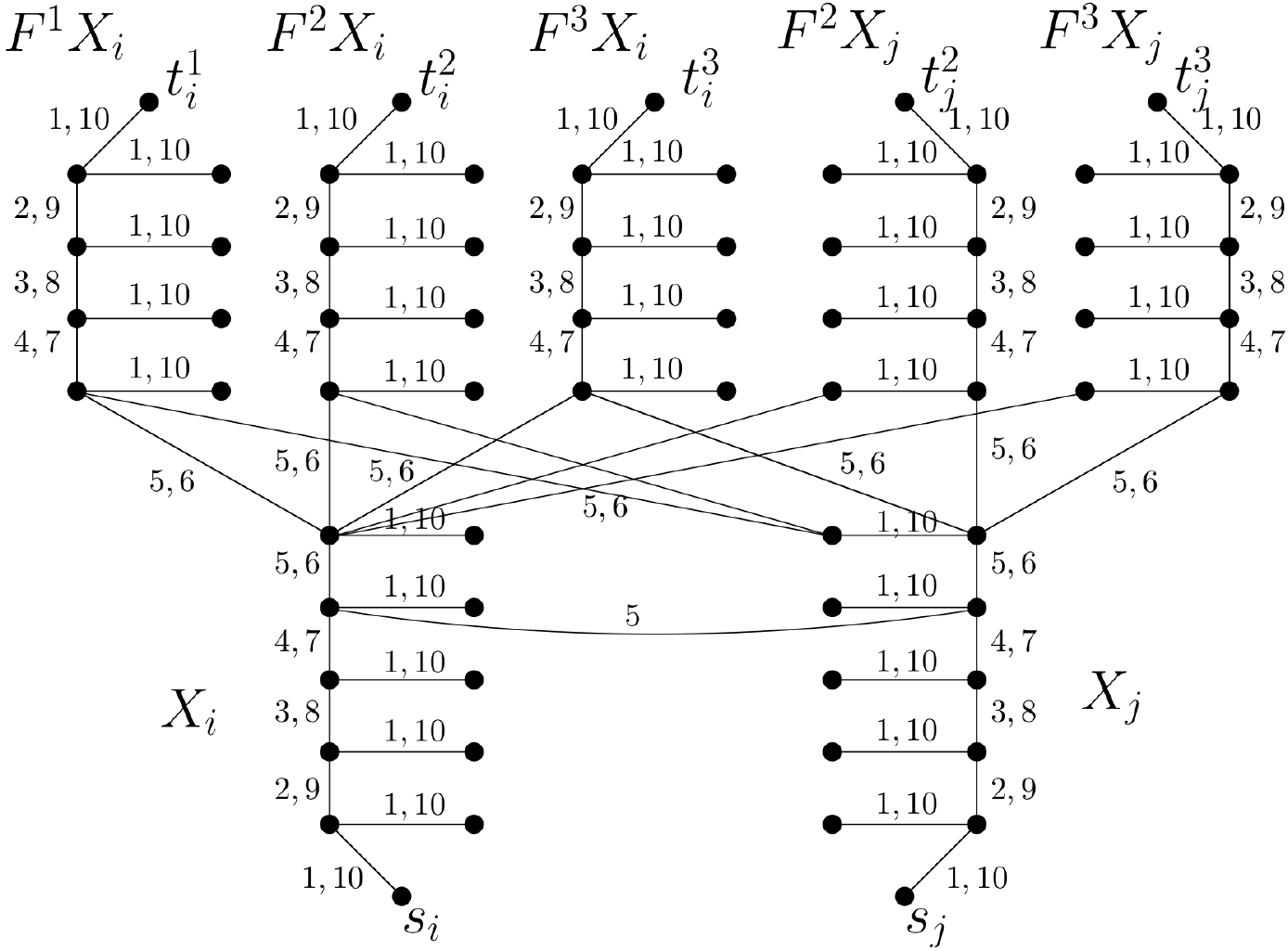}
			\caption{$x_i$ and $x_j$ appear together in a clause, where $x_i$ appears with its third and $x_j$ with its first appearance.}
			\label{fig:NP-labeledIntersecting}
		\end{subfigure}
		\caption{Example of the labeling $\lambda$ on variable gadgets $X_i, X_j$ and variable edges between them, where $x_i$ is \textsc{True} and $x_j$ \textsc{False} in $\phi$. Note, edges that are not labeled are omitted, $F^3X_i = F^1X_j$ and $t_i^3 = t_j^1$. \label{fig:NP-labeled}}
	\end{figure}
	
	Summarizing all of the above we get that the labeling $\lambda_\phi$ uses
	$74$ labels on each variable gadget and  
	$13$ labels on variable edges among any two variables,
	from which we have to subtract the following:
	\begin{itemize}
		\item $4$ labels for each pairs of variable edges between two variables that appear in the same clause,
		\item $16$ labels for the shared fork between two variables, that appear in a satisfied clause,
		\item $8$ labels for the shared fork between two variables, that appear in a non-satisfied clause.
	\end{itemize}
	Altogether sums up to the
	$74n +13 \frac{n(n-1)}{2} - 4m -16k - 8(m-k)$ labels.
	Therefore, if $\tau$ satisfies at least $k$ clauses of $\phi$,
	the labeling $\lambda_\phi$ consists of at most
	$\frac{13}{2}n^2+ \frac{99}{2}n-8k$ labels.
\end{proof}

\medskip

Before proving the statement in the other direction, we have to show some structural properties.
Let us fix the following notation.
If a labeling $\lambda_\phi$ labels all left (resp.~right) paths of the variable gadget $X_i$ (\ie both bottom-up from $s_i$ to $t_i^1, t_i^2, t_i^3$ and top-down from $t_i^1, t_i^2, t_i^3$ to $s_i$ with labels $1,2\ldots,10$ in this order), 
then we say that the variable gadget $X_i$ is \emph{left-aligned} (resp.~\emph{right-aligned}) in the labeling $\lambda_\phi$.
Note, if at least one edge on any of these left (resp.~right) paths of $X_i$ is not labeled with the appropriate label between $1$ and $10$, then the variable gadget is \emph{not} left-aligned (resp.~\emph{not} right-aligned).
Every temporal path from $s_i$ to $t_i^\ell$ (resp.~from $t_i^\ell$ to $s_i$) of length 10 in $X_i$ is called an \emph{upward path} (resp.~a \emph{downward path}) in $X_i$. Any part of an upward (resp.~downward) path is called a \emph{partial} upward (resp.~downward) path. 
Note that, for any $\ell,\ell'\in\{1,2,3\}$, $\ell\neq \ell'$, a temporal path from $t_i^\ell$ to $t_i^{\ell'}$ of length 10 is the union of 
a partial downward path on the fork $F_i^\ell$ and a partial upward path on $F_i^{\ell'}$. Moreover, note that these two partial downward/upward paths must be either both parts of a left temporal path or both parts of a right temporal path between $s_i$ and $t_i^\ell,t_i^{\ell'}$. 
The following technical lemma will allow us to prove the correctness of our reduction.

\begin{lemma} \label{lem:APXproofSecondDirectionAlignedVariables}	
	Let $\lambda_\phi$ be a minimum labeling of $G_\phi$. Then $\lambda_\phi$ can be modified in polynomial time to a minimum labeling of $G_\phi$ in which each variable gadget $X_i$ is either left-aligned or right-aligned.
\end{lemma}

\begin{proof}
	Let $\lambda_\phi$ be a minimum labeling of $G_\phi$ that admits at least one variable gadget $X_i$ that is neither left-aligned nor right-aligned. 
	
	First we will prove that there exists a fork $F^\ell X_i$ of $X_i$ that admits at least three partial upward or downward paths,
	\ie it either has two partial upward paths (one on each side of the fork) and at least one partial downward path, 
	or two partial downward paths (one on each side of the fork) and at least one partial upward path.
	For the sake of contradiction, suppose that each of the forks $F^1 X_i, F^2 X_i, F^3 X_i$ contains at most two partial upward or downward paths. Then, since $\lambda_\phi$ must have in $X_i$ at least one upward and at least one downward path between $s_i$ and $t_i^\ell$, $\ell\in \{1,2,3\}$, it follows that each fork $F^\ell X_i$ has \emph{exactly} one partial upward and \emph{exactly} one partial downward path. 
	
	Assume that each of the forks $F^1 X_i, F^2 X_i, F^3 X_i$ has both its partial upward and downward paths on the same side of $X_i$ (\ie either both on the left or both on the right side of $X_i$). If all of them have their partial upward and downward paths on the left (resp.~right) side of $X_i$, then $X_i$ is left-aligned (resp.~right-aligned), which is a contradiction. Therefore, at least one fork (say $F^1 X_i$) has its partial upward and downward paths on the left side of $X_i$ and at least one other fork (say $F^2 X_i$) has its partial upward and downward paths on the right side of $X_i$. But then there is no temporal path from $t_i^1$ to $t_i^2$ of length 10 in $\lambda_\phi$, which is a contradiction. Therefore there exists at least one fork $F^\ell X_i$ (say, $F^1 X_i$ w.l.o.g.), in which (w.l.o.g.) the partial upward path is on the right side and the partial downward path is on the left side of $X_i$.
	
	Since the partial downward path of $F^1 X_i$ is on the left side of $X_i$, it follows that the partial upward path of each of $F^2 X_i$ and $F^2 X_i$ is on the left side of $X_i$. Indeed, otherwise there is no temporal path of length 10 from $t_i^1$ to $t_i^2$ or $t_i^3$ in $\lambda_\phi$, a contradiction. Similarly, since the partial upward path of $F^1 X_i$ is on the right side of $X_i$, it follows that the partial downward path of each of $F^2 X_i$ and $F^2 X_i$ is on the right side of $X_i$. But then, there is no temporal path of length 10 from $t_i^2$ to $t_i^3$, or from $t_i^3$ to $t_i^2$ in $\lambda_\phi$, which is also a contradiction. 
	Therefore at least one fork $F^\ell X_i$ (say $F^3 X_i$) of $X_i$ admits at least three partial upward or downward paths.
	
	W.l.o.g.~we can assume that the fork $F^3 X_i$ has two partial downward paths and at least one partial upward path which is on the left side of $X_i$. We distinguish now the following cases.
	
	\noindent\textbf{Case A.} The fork $F^3 X_i$ has no partial upward path on the right side of~$X_i$. Then the base $BX_i$ has a partial upward path on the left side of $X_i$. Furthermore, each of the forks $F^1 X_i,F^2 X_i$ has a partial downward path on the left side of~$X_i$. 
	
	\medskip
	
	\noindent\textbf{Case A-1.} The base $BX_i$ of $X_i$ has no partial downward path on the left side of~$X_i$; that is, there is no temporal path from vertex $e_i$ to vertex $s_i$ with labels ``6,7,8,9,10''. Then the base $BX_i$ of $X_i$ has a partial downward path on the right side of~$X_i$, as otherwise there would be no temporal path of length 10 from any of $t_i^1,t_i^2,t_i^3$ to $s_i$. For the same reason, each of the forks $F^1 X_i,F^2 X_i$ has a partial downward path on the right side of $X_i$.
	
	\noindent\textbf{Case A-1-i.} None of the forks $F^1 X_i,F^2 X_i$ has a partial upward path on the left side of~$X_i$. 
	Then each of the forks $F^1 X_i,F^2 X_i$ has a partial upward path on the right side of~$X_i$, as otherwise there would be no temporal path of length 10 from $s_i$ to $t_i^1$ or $t_i^2$. For the same reason, the base $BX_i$ has a partial upward path on the right side of $X_i$. Therefore we can remove the label ``5'' from the left bridge edge $e_i f_i^3$ of the fork $F^3 X_i$, thus obtaining a labeling with fewer labels than $\lambda_\phi$, a contradiction.
	
	\noindent\textbf{Case A-1-ii.} Exactly one of the forks $F^1 X_i,F^2 X_i$ (say $F^1 X_i$) has a partial upward path on the left side of~$X_i$. Then the fork $F^2 X_i$ has a partial upward path on the right side of $X_i$. 
	Furthermore the base $BX_i$ has a partial upward path on the right side of $X_i$, since otherwise there would be no temporal path of length 10 from $s_i$ to $t_i^2$. In this case we can modify the solution as follows: remove the labels ``1,2,3,4,5'' from the partial right-upward path of $BX_i$ and add the labels ``6,7,8,9,10'' to the partial left-upward path of the fork $F^2 X_i$. Finally we can remove the label ``5'' from the right bridge edge $\overline{e_i} \overline{f_i}^3$ of the fork $F^3 X_i$, thus obtaining a labeling with fewer labels than $\lambda_\phi$, a contradiction.
	
	\noindent\textbf{Case A-1-iii.} Each of the forks $F^1 X_i,F^2 X_i$ has a partial upward path on the left side of~$X_i$. 
	In this case we can modify the solution as follows: remove the labels ``10,9,8,7,6'' from the partial right-downward path of $BX_i$ and add the same labels ``10,9,8,7,6'' to the partial left-downward path of the base $BX_i$. Finally we can remove the label ``5'' from the right bridge edge $\overline{e_i} \overline{f_i}^3$ of the fork $F^3 X_i$, thus obtaining a labeling with fewer labels than $\lambda_\phi$, a contradiction.
	
	\medskip
	
	\noindent\textbf{Case A-2.} The base $BX_i$ of $X_i$ has a partial downward path on the left side of~$X_i$; that is, there is a temporal path from vertex $e_i$ to vertex $s_i$ with labels ``6,7,8,9,10''. 
	
	\noindent\textbf{Case A-2-i.} None of the forks $F^1 X_i,F^2 X_i$ has a partial upward path on the left side of~$X_i$. 
	Then the base $BX_i$ and each of the forks $F^1 X_i,F^2 X_i$ have a partial upward path on the right side of~$X_i$, as otherwise there would be no temporal paths of length 10 from $s_i$ to $t_i^1,t_i^2$. 
	Moreover, as none of $F^1 X_i,F^2 X_i$ has a partial left-upward path, it follows that each of $F^1 X_i,F^2 X_i$ has a partial downward path on the right side of~$X_i$. Indeed, otherwise there would be no temporal paths of length 10 between $t_i^1$ and $t_i^2$. 
	In this case we can modify the solution as follows: remove the labels ``1,2,3,4,5'' from the partial left-upward path of $BX_i$ and add the labels ``6,7,8,9,10'' to the partial right-upward path of the fork $F^3 X_i$. Finally we can remove the label ``6'' from the left bridge edge $e_i f_i^3$ of the fork $F^3 X_i$, thus obtaining a labeling with fewer labels than $\lambda_\phi$, a contradiction.

	\noindent\textbf{Case A-2-ii.} Exactly one of the forks $F^1 X_i,F^2 X_i$ (say $F^1 X_i$) has a partial upward path on the right side of~$X_i$. Then the fork $F^2 X_i$ has a partial upward path on the left side of $X_i$. 
	Furthermore the base $BX_i$ must have a partial right-upward path, as otherwise there would be no temporal path from $s_i$ to $t_i^2$. In this case we can modify the solution as follows: remove the labels ``1,2,3,4,5'' from the partial right-upward path of $BX_i$ and add the labels ``6,7,8,9,10'' to the partial left-upward path of the fork $F^2 X_i$. Finally we can remove the label ``5'' from the right bridge edge $\overline{e_i} \overline{f_i}^3$ of the fork $F^3 X_i$, thus obtaining a labeling with fewer labels than $\lambda_\phi$, a contradiction.

	\noindent\textbf{Case A-2-iii.} Each of the forks $F^1 X_i,F^2 X_i$ has a partial upward path on the right side of~$X_i$. Then we we can simply remove the label ``5'' from the right bridge edge $\overline{e_i} \overline{f_i}^3$ of the fork $F^3 X_i$, thus obtaining a labeling with fewer labels than $\lambda_\phi$, a contradiction.

	\medskip
	
	\noindent\textbf{Case B.} The fork $F^3 X_i$ has also a partial upward path on the right side of~$X_i$. That is, $F^3 X_i$ has 
	partial upward-left, upward-right, downward-left, and downward-right paths. 
	
	\medskip
	
	\noindent\textbf{Case B-1.} The base $BX_i$ of $X_i$ has no partial downward path on the left side of~$X_i$. Then the base $BX_i$ of $X_i$ has a partial downward path on the right side of~$X_i$, as otherwise there would be no temporal path of length 10 from any of $t_i^1,t_i^2,t_i^3$ to $s_i$. For the same reason, each of the forks $F^1 X_i,F^2 X_i$ has a partial downward path on the right side of $X_i$.
	
	Note that Case B-1 is symmetric to the case where the base $BX_i$ of $X_i$ has no partial right-downward (resp.~left-upward, right upward) path.
	
	\noindent\textbf{Case B-1-i.} None of the forks $F^1 X_i,F^2 X_i$ has a partial upward path on the left side of~$X_i$. 
	This case is the same as Case A-1-i.
	
	\noindent\textbf{Case B-1-ii.} Exactly one of the forks $F^1 X_i,F^2 X_i$ (say $F^1 X_i$) has a partial upward path on the left side of~$X_i$. Then both the base $BX_i$ and the fork $F^2 X_i$ has a partial right-upward path, as otherwise there would be no temporal path of length 10 from $s_i$ to $t_i^2$. In this case, we can always remove the label ``6'' from the left bridge edge $e_i f_i^3$ of the fork $F^3 X_i$ (without compromising the temporal connectivity), thus obtaining a labeling with fewer labels than $\lambda_\phi$, a contradiction. 
	
	\noindent\textbf{Case B-1-iii.} Each of the forks $F^1 X_i,F^2 X_i$ has a partial upward path on the left side of~$X_i$. That is, each of $F^1 X_i,F^2 X_i$ has a partial left-upward and a partial right-downward path.
	The following subcases can occur:
	
	\noindent\textbf{Case B-1-iii(a).} None of the forks $F^1 X_i,F^2 X_i$ has a partial right-upward path. 
	Then each of the forks $F^1 X_i,F^2 X_i$ has a partial left-downward path, since otherwise there would not exist temporal paths of length 10 between $t_i^1$ and $t_i^2$. 
	Furthermore, the base $BX_i$ has a partial left-upward path, since otherwise there would not exist a temporal path of length 10 from $s_i$ to $t_i^1$ and $t_i^2$. 
	In this case, we can remove the label ``6'' from the right bridge edge $\overline{e_i} \overline{f_i}^3$ of the fork $F^3 X_i$, thus obtaining a labeling with fewer labels than $\lambda_\phi$, a contradiction.
	
	\noindent\textbf{Case B-1-iii(b).} Exactly one of the forks $F^1 X_i,F^2 X_i$ (say $F^1 X_i$) has a partial right-upward path. 
	Then the base $BX_i$ has a partial left-upward path, since otherwise there would not exist a temporal path of length 10 from $s_i$ to $t_i^2$. 
	Similarly, the fork $F^1 X_i$ has a partial left-downward path, since otherwise there would not exist a temporal path of length 10 from $t_i^1$ to $t_i^2$. 
	In this case we can modify the solution as follows: 
	First, remove the labels ``10,9,8,7,6'' from the partial right-downward path of $BX_i$ and add the labels ``10,9,8,7,6'' to the partial left-downward path of $BX_i$. 
	Second, remove the labels ``5,6'' from each of t two right bridge edges $\overline{e_i} \overline{f_i}^1$ and $\overline{e_i} \overline{f_i}^3$ of the forks $F^1 X_i$ and $F^3 X_i$, respectively. 
	Third, remove the label ``5'' from the right bridge edge $\overline{e_i} \overline{f_i}^1$ of the fork $F^2 X_i$. 
	Finally, add the five labels ``5,4,3,2,1'' to the partial left-downward path of the fork $F^2 X_i$. 
	The resulting labeling $\lambda_\phi^*$ still preserves the temporal reachabilities and has the same number of labels as $\lambda_\phi$, while the variable gadget $X_i$ is aligned.
	
	\noindent\textbf{Case B-1-iii(c).} Each of the forks $F^1 X_i,F^2 X_i$ has a partial right-upward path. 
	In this case, we can always remove the label ``5'' from the left bridge edge $e_i f_i^3$ of the fork $F^3 X_i$, thus obtaining a labeling with fewer labels than $\lambda_\phi$, a contradiction. 
	
	\medskip
	
	\noindent\textbf{Case B-2.} The base $BX_i$ of $X_i$ has partial left-downward, right-downward, left-upward, and right-upward paths. Then, due to symmetry, we may assume w.l.o.g.~that the fork  $F^1 X_i$ has a left-upward path. 
	Suppose that $F^1 X_i$ has also a left-downward path. 
	In this case we can modify the solution as follows: remove the labels ``1,2,3,4,5'' and ``10,9,8,7,6'' from the partial right-upward and right-downward paths of $BX_i$ and add the labels ``6,7,8,9,10'' and ``5,4,3,2,1'' to the partial left-upward and left-downward paths of the fork $F^2 X_i$. Finally we can remove the label ``6'' from the right bridge edge $\overline{e_i} \overline{f_i}^3$ of the fork $F^3 X_i$, thus obtaining a labeling with fewer labels than $\lambda_\phi$, a contradiction.
	
	Finally suppose that $F^1 X_i$ has no partial left-downward path. Then $F^1 X_i$ has a partial right-down path, since otherwise there would not exist any temporal path of length 10 from $t_i^1$ to $s_i$. 
	Similarly, the fork $F^2 X_i$ has a partial right-upward path, since otherwise there would not exist any temporal path of length 10 from $t_i^1$ to $t_i^2$. 
	In this case we can modify the solution as follows: First remove the labels ``1,2,3,4,5'' and ``10,9,8,7,6'' from the partial left-upward and left-downward paths of $BX_i$. 
	Second add the labels ``6,7,8,9,10'' to the partial right-upward path of the fork $F^1 X_i$ and add the labels ``5,4,3,2,1'' to the partial right-downward path of the fork $F^2 X_i$. 
	Finally remove the label ``6'' from the left bridge edge $e_i f_i^3$ of the fork $F^3 X_i$, thus obtaining a labeling with fewer labels than $\lambda_\phi$, a contradiction.
	
	\medskip
	
	Summarizing, starting from an optimum $\lambda_\phi$ of $G_\phi$, in which at least one variable gadget is neither left-aligned nor right-aligned, we can modify $\lambda_\phi$ to another labeling $\lambda_\phi^*$, such that $\lambda_\phi^*$ has one more variable-gadget that is aligned and $|\lambda_\phi|=|\lambda_\phi^*|$. Note that this modification can only happen in Case B-1-iii(b); in all other cases our case analysis arrived at a contradiction. 
	Note here that, by making the above modifications of $\lambda_\phi$, we need to also appropriately modify the ``connecting edges'' (within the variable gadgets) and the ``variable edges'' (between different variable gadgets), without changing the total number of labels in each of these edges. Finally, it is straightforward that all modifications of $\lambda_\phi$ can be done in polynomial time. This concludes the proof.
\end{proof}

\begin{theorem}\label{thm:APXhardnessSecondDirection}
	If OPT$_{\MAL}(G_\phi,d_\phi) \leq \frac{13}{2}n^2+ \frac{99}{2}n-8k$ then OPT$_{\MAXSATshort}(\phi) \geq k$, 
	where $n$ is the number of variables in the formula $\phi$.
\end{theorem}

\begin{proof}
	Recall by \cref{diameter-G-phi-lem} that $d_\phi=10$.
	Let $\lambda_\phi$ be an optimum solution to \textsc{MAL}$(G_\phi,10)$, which uses OPT$_{\MAL}(G_\phi,d_\phi)\leq \frac{13}{2}n^2+ \frac{99}{2}n-8k$ labels by the assumption of the theorem. We will prove that there exists a truth assignment $\tau$ that satisfies at least $k$ clauses of $\phi$. 
	Recall that, since $\phi$ is an instance of \MAXSATshort\ with $n$ variables, it has $m=\frac{3}{2}n$ clauses.

	Let $X_i$ and $X_j$ be two variable gadgets in $G_\phi$.
	First we observe that
	the temporal path from a starting vertex $s_i$ of $X_i$, to any of the ending vertices $t_i^\ell$, where $\ell \in \{1,2,3\}$,
	must only go through the vertices and edges of the variable gadget $X_i$.
	This is true since in any other case the temporal path would use at least one variable edge 
	and in this case the distance of the path would increase by at least one.
	Therefore, the path would be of length at least $11$, but since the diameter of the graph is $10$, the largest label that is allowed to be used is $10$
	and thus the longest temporal path can use at most $10$ edges.
	Similarly it holds for temporal paths from the ending vertices $t_i^\ell$ ($\ell \in \{1,2,3\}$) to the starting vertex $s_i$ and the temporal paths among the ending vertices.
	Even more, these temporal paths must be either all on the left or all on the right side of $X_i$, 
	\ie they have to use vertices and edges that are all on the left or the right side of the base $BX_i$ and each fork $F^\ell X_i$.
	This holds as paths of any other form (\ie containing vertices and edges of both sides) are of length at least $11$.
	Consequently, to label a $(s_i,t_i^1)$-path in both directions any labeling must use at least $2 \cdot 10$ labels.
	Now, to label $(s_i,t_i^2)$ and $(s_i,t_i^3)$-paths, the labels on the base $BX_i$ can be reused, which produces additional $10$ labels on each fork $F^2X_i$ and $F^3X_i$.
	In the case when all these labels were used on the same path of the variable gadget 
	\ie all labels were placed on the left or on the right side of $BX_i$ and $F^iX_i$,
	there are also temporal paths connecting all three ending vertices, without having to introduce any extra labels.
	The only missing part is to assure that also all the vertices from the opposite side 
	(\ie if the labeling used left paths, then the opposite vertices are on the right side, or vice versa)
	are able to reach and be reached by any other vertex.
	Therefore, we need at least $2$ more labels (one for incoming and one for outgoing temporal paths)
	on the edges connecting them with the path (vertices) on the other side.
	Altogether, to ensure the existence of a temporal path between any two vertices from $X_i$,
	a labeling must use at least $74$ labels on a variable gadget $X_i$.
	
	Now, let $X_i$ and $X_j$ be such variable gadgets that do not share the fork.
	As observed above, all vertices from each variable gadget can only be reached among each other, without using the variable edges.
	Therefore, the variable edges must be labeled in such a way, that they ensure a temporal path among vertices from different variable gadgets.
	W.l.o.g.~we can assume that $X_i$ is left-aligned and $X_j$ is right-aligned by $\lambda_\phi$ 
	(all the other cases of aligned and non-aligned labelings of $X_i$ and $X_j$ by $\lambda_\phi$, are symmetric).
	Since the starting vertex $s_i$ is on the distance $10$ from the ending vertices of $t_j^{\ell'}$ ($\ell' \in \{1,2,3\}$) of $X_j$, 
	there must be a temporal path using all labels, to connect them.
	This path must use the edge of the form $e_i \overline{f_j}^{\ell'}$, as any other path is longer than $10$. 
	Since the path must be traversed in both direction each edge $e_i \overline{f_j}^{\ell'}$ ($\ell' \in \{1,2,3\})$ must have at least $2$ labels.
	Similarly it holds for the $(s_j, t_i^\ell)$-paths ($\ell \in \{1,2,3\})$
	and the edges $\overline{e_j} f_i^\ell$ ($\ell \in \{1,2,3\})$.
	For a vertex $s_i$ to reach $s_j$ we must label the edge $d_i \overline{d_j}$, as any other $(s_i,s_j)$-path is longer than $10$.
	Therefore, we need at least one extra label for the edge $d_i \overline{d_j}$.
	Altogether, to ensure the existence of a temporal path among two vertices from two variable gadgets that do not share a fork,
	a labeling must use at least $13$ labels on the variable edges. 
	
	Lastly, let $X_i$ and $X_j$ be two variable gadgets that share a fork.
	W.l.o.g.~we can suppose that $X_i$ is left-aligned by the optimum labeling $\lambda_\phi$
	and that $F^rX_i = F^{r'}X_j$.
	By the construction of $G_\phi$, there exists a temporal path to and from all the vertices in the fork $F^rX_i = F^{r'}X_j$ to all vertices in $X_i$ and $X_j$,
	as there is a temporal path among all vertices from $X_i$
	and a temporal path among all vertices in $X_j$.
	As observed above, these paths do not use the variable edges,
	but the variable edges must be labeled in such a way, that they ensure a temporal path among vertices from different variable gadgets.
	Now if we suppose that the variable gadget $X_j$ is right-aligned by the labeling $\lambda_\phi$, then
	a temporal path between $s_i$ and $s_j$ must use the edge $d_i \overline{d_j}$ and therefore 
	at least one extra label is used for this edge.
	A temporal path between $s_i$ and $t_j^{\ell'}$, where $\ell' \in \{1,2,3\} \setminus \{r'\}$, must use the edge $e_i\overline{f_j}^{\ell'}$. 
	Since the edge of this form is traversed in both directions it must have at least two labels.
	Similarly it holds for the temporal paths between $t_i^\ell$ ($\ell \in \{1,2,3\} \setminus \{r\}$) and $s_j$.
	Altogether, to ensure the existence of a temporal path among any two vertices from two variable gadgets that share a fork,
	a minimum labeling must use at least $9$ labels on the variable edges. 
	Similarly we can see that also all other combinations of aligned and non-aligned labelings of $X_i$ and $X_j$ by $\lambda_\phi$,
	require at least $9$ labels on the variable edges.
	
	The only thing left to study, in the case of two variable gadgets that share a fork, is what happens in the intersecting fork.
	By \cref{lem:APXproofSecondDirectionAlignedVariables} we know that the variable gadgets $X_i$ and $X_j$ are aligned by the labeling $\lambda_\phi$. Suppose that $F^rX_i = F^{r'}X_j$.
	W.l.o.g.~we can assume that $X_i$ is left-aligned.
	We distinguish the following two cases.
	\begin{itemize}
		\item 
		The variable gadget $X_j$ is right-aligned.
		Then, by the construction of $G_\phi$, 
		the fork $F^1X_i = F^1X_j$ is labeled using the same labeling, 
		\ie the left labeling of the variable gadget $X_i$ coincides with the right labeling of the variable gadget $X_j$.
		This ``saves'' $16$ labels from the total number of labels used on variable gadgets $X_i$ and $X_j$.
		
		\item
		The variable gadget $X_j$ is left-aligned.
		In this case
		all edges in
		the fork $F^1X_i = F^1X_j$ admit two labels.
		This ``saves'' $8$ labels from the total number of labels used on variable gadgets $X_i$ and $X_j$,
		since both labelings coincide on the connecting edges.
	\end{itemize}
	
	From the labeling $\lambda_\phi$ of $G_\phi$ we construct a truth assignment $\tau$ of $\phi$ in the following way.
	If a variable gadget $X_i$ is left-aligned, we set $x_i$ to \textsc{True} and
	if it is right-aligned, we set $x_i$ to \textsc{False}.
	Using the results from above we deduce that the truth assignment $\tau$ satisfies at most $k$ clauses.
\end{proof}

\medskip

Since \MAL\ is clearly in NP, the next theorem follows directly by \cref{maxsat-hard-thm,thm:APXhardnessFirstDirection,thm:APXhardnessSecondDirection}.

\begin{theorem}\label{thm:NPDiameterStatic}
	\MAL\ is NP-complete on undirected graphs, even when the required maximum age is equal to the diameter of the input graph.
\end{theorem}

\section{The Steiner-Tree variations of the problem}\label{Steiner-sec}

In this section we investigate the computational complexity of the Steiner-Tree variations of the problem, namely \MSL\ and \MASL. 
First, we prove in Section~\ref{MSL-NP-complete-subsec} that the age-unrestricted problem \MSL\ remains NP-hard, using a reduction from \textsc{Vertex Cover}. In Section~\ref{MSL-FPT-subsec} we prove that this problem is in FPT, when parameterized by the number $|R|$ of terminals. 
Finally, using a parameterized reduction from \MCC, we prove in Section~\ref{MASL-W1-hard-subsec} that the age-restricted version \MASL\ is W[1]-hard with respect to $|R|$, even if the maximum allowed age is a constant.

\subsection{\MSL\ is NP-complete}\label{MSL-NP-complete-subsec}
\begin{theorem}\label{thm:MinCoreNP}
	\MSL\ is \NP-complete.
\end{theorem}

\begin{proof}
	\MSL\ is clearly contained in \NP.
	To prove that the \MSL\ is \NP-hard we provide a polynomial-time reduction from the \NP-complete \VC\ problem~\cite{Kar72}.
	
	\problemdef{\VC}
	{A static graph $G = (V,E)$, a positive integer $k$.}
	{Does there exist a subset of vertices $S \subseteq V$ such that $|S| = k$ and $\forall e \in E, e \cap S \neq \emptyset$.}	
	
	Let $(G,k)$ be an input of the \VC\ problem and denote $|V(G)| = n, |E(G)|=m$. We assume w.l.o.g.\ that $G$ does not admit a vertex cover of size $k-1$.
	We construct $(G^*,R^*,k^*)$, the input of \MSL\ using the following procedure.
	The vertex set $V(G^*)$ consists of the following vertices:
	\begin{itemize}
		\item two starting vertices $N = \{n_0, n_1\}$,
		\item a ``vertex-vertex'' corresponding to every vertex of G: $U_V = \{u_v | v \in V(G)\}$,
		\item an ``edge-vertex'' corresponding to every edge of G: $U_E = \{u_e | e \in E(G)\}$,
		\item $2 n + 12m \cdot k$ ``dummy'' vertices.
	\end{itemize}
	The edge set $E(G^*)$ consists of the following edges:
	\begin{itemize}
		\item an edge between starting vertices, \ie $n_0 n_1$,
		\item a path of length $3$ between a starting vertex $n_1$ and every vertex-vertex $u_v \in U_V$ using $2$ dummy vertices, and
		\item for every edge $e=vw \in E(G)$ we connect the corresponding edge-vertex $u_e$ with the vertex-vertices $u_v$ and $u_w$, each with a path of length $6k + 1$ using $6k$ dummy vertices.
	\end{itemize}
	We set $R^* = \{n_0\} \cup U_E$ and $k^* = 6k + 2m (6k + 1) + 1$. This finishes the construction. 
	It is not hard to see that this construction can be performed in polynomial time.
	For an illustration see \cref{fig:NP-MCLfromVC}.
	Note that any two paths in $G^*$ can intersect only in vertices from $N \cup U_V \cup U_E$ and not in any of the dummy vertices.
	At the end $G^*$ is a graph with $3n  + m (12k + 1) + 2$ vertices and $1 + 3n + 2m (6k + 1)$ edges.
	
	We claim that $(G,k)$ is a \textsc{YES} instance of the \VC\ if and only if $(G^*,R^*,k^*)$ is a \textsc{YES} instance of the \MSL.
	
	\begin{figure}[!htb]
		\centering
		\includegraphics[width=0.6\linewidth]{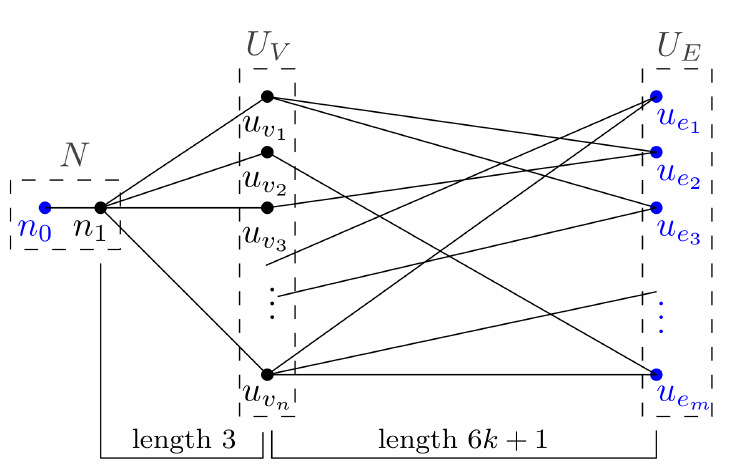}
		\caption{Example of a canonical layering of a directed acyclic graph (DAG).}
		\label{fig:NP-MCLfromVC}
	\end{figure}
	
	($\Rightarrow$):
	Assume $(G,k)$ is a \textsc{YES} instance of the \VC\ and let $S\subseteq V(G)$ be a vertex cover for $G$ of size $k$.
	We construct a labeling $\lambda$ for $G^*$ that uses $k^*$ labels and admits a temporal path between all vertices from $R^*$ as follows.	
	
	For the sake of easier explanation we use the following terminology. A temporal path starting at $n_0$ and finishing at some $u_e$ is called a \emph{returning path}.
	Contrarily, a temporal path from some $u_e$ to $n_0$ is called a \emph{forwarding path}.
	
	Let $U_S$ be the set of corresponding vertices to $S$ in $G^*$.
	From each edge vertex $u_e$ there exists a path of length $6k + 1$ to at least one vertex $u_v \in U_S$, since $S$ is a vertex cover in~$G$.
	We label exactly one of these paths, using labels $1, 2, \dots , 6k +1$. 
	Since $S$ is of size $k$, this part uses $k (6k+1)$ labels.
	Now we label a path from each $v \in U_S$ to $n_1$ using labels $6k + 2, 6k + 3, 6k + 4$.
	Each path uses $3$ labels, and since $S$ is of size $k$ we used $3k$ labels for all of them.
	At the end we label the edge $n_0 n_1$ with the label $\ell^* = 6k + 5$.
	Using this procedure we have created a forwarding path from each edge vertex $u_e$ to the start vertex $n_0$ and we used $3k  + m(6k + 1) + 1$ labels. 
	
	To create the returning paths, we label paths from $n_1$ to each vertex in $U_S$ with labels $\ell^* + 1,\ell^* + 2, \ell^* + 3$.
	Now again, we label exactly one path from vertices in $U_S$ to each edge-vertex $u_e$, using labels $\ell^* + 4, \ell^* + 5, \dots, \ell^* + 3 + 6k$.
	We used extra $3k + m(6k + 1)$ labels and created a returning path from $n_0$ to each vertex in $U_E$.
	
	All together, the constructed labeling uses $k^* = 6k + 2m (6k +1) + 1$ labels, the only thing left to show is that there exists a temporal path between any pair of edge-vertices
	$u_e, u_f \in U_E$.
	It is not hard to see that this holds, as we can construct a temporal path between two edge-vertices as a union of a (sub)path of a temporal path from the first edge-vertex to the starting vertex $n_0$ and a (sub)path of a temporal path from the starting vertex to the other edge-vertex.
	
	($\Leftarrow$):
	Assume that $(G^*,R^*,k^*)$ is a \textsc{YES} instance of the \MSL. We construct a vertex cover of size at most $k$ for $G$ as follows.	
	
	Let us first observe the following, a forwarding and returning path between the starting vertex $n_0$ and the same edge-vertex $u_e$, can intersect in at most one time edge.
	Even more, two temporal paths between the same pair of vertices, going in the opposite directions, intersect in at most one time edge.
	
	By the construction of $G^*$ each (temporal) path between $n_0$ and a vertex in $U_E$ passes through the set $U_V$. 
	Since there are $m$ vertices in $U_E$ and each path between a vertex $u_e \in U_E$ and some $u_v \in U_V$ is of length $6k + 1$, we need at least $m (6k + 1)$ labels to connect $U_E$ to $U_V$
	in ``one direction''.
	Using the observation from above, we get that there can be at most $1$ time edge in common between any two temporal paths among any pair of edge-vertices, 
	therefore we need at least $2 m (6k + 1) - 1$ labels for paths in both directions. We call these the forwarding path $F_{e}$ (from $u_e$ to some $u_v$) and the returning path $R_{e}$ (from some $u_{v'}$ to $u_e$) for $u_e$. It is straightforward to check that every $u_e$ can have at most one forwarding path and one returning path, since every additional path would require at least an additional $6k$ labels and then no connection between $n_0$ and $U_V$ would be possible.
	
	All labeled temporal paths between $N$ and $U_V$ can be split into two sets,
	one containing all temporal paths that are a part of (or can be extended to) some returning path, denote them $P^+_N$ and the others which are a part of (or can be extended to) some forwarding path, denote them $P^-_N$.
	It is not hard to see that each temporal path from $P^+_N$ or $P^-_N$ starts and ends in $N \cup U_V$, \ie no temporal path starts/ends in one of the dummy vertices.
	Therefore each temporal path in $P^+_N$ or $P^-_N$ uses $3$ labels.
	Again, using the above observation we get that temporal paths from $P^+_N$ and $P^-_N$ share at most one label. 
	Since this part uses at most $6k +1$ labels, there are at most $2k$ temporal paths in $P^+_N$ and $P^-_N$.
	Suppose that $|P^+_N| \le |P^-_N|$ (the case where $|P^+_N| > |P^-_N|$ is analogous). 
	Let $U_S \subseteq U_V$ be the set of vertices in $U_V$ such that $P^+_N \cap U_S \neq \emptyset$, 
	\ie $U_S$ consists of vertices that are endpoints of temporal paths in $P^+_N$.
	We claim that $S = \{v \mid u_v \in U_S \}$ is a vertex cover of $G$ and $|S| \leq k$.
	It is not hard to see that $|P^+_N| \leq k$ and therefore $|S| \leq k$.
	
	We first make the following observation. We define a partial order on the set $\mathcal{P}=\{F_e,R_e\mid e\in E\}$ of forwarding and returning paths as follows. For two paths $P,Q\in\mathcal{P}$, we say that $P<Q$ if all labels used in $P$ are strictly smaller than the smallest label used in $Q$.
	We can assume w.l.o.g.\ that the defined ordering is a total ordering on $\mathcal{P}$ since we can order incomparable path pairs arbitrarily by modifying the labels in a way that does not change the size and the connectivity properties of the labeling. Furthermore, we can observe that for any two $e,e'\in E$ with $e\neq e'$ we have that $F_e<R_{e'}$ since in order for $u_e$ to reach $u_{e'}$, the path $F_e$ needs to be used before the path $R_{e'}$. It follows that there is at most one edge $e\in E$ such that $R_e<F_e$, otherwise we would reach a contradiction to the above observation.
	
	Now assume for contradiction that $S$ is not a vertex cover of $G$. Then there is an edge $e=\{v,w\}\in E$ such that $\{v,w\}\cap S=\emptyset$.	
	To reach $u_e$ from $n_0$ there needs to be an edge $e'=\{v,w'\}$ (or symmetrically $\{w,w'\}$) such that we can reach $u_{w'}$ from $n_1$ via some path $P$, then continue to $u_{e'}$ using $R_{e'}$, then continue to $u_{v}$ using $F_{e'}$, finally reach $u_{e}$ using $R_e$. Notice that this requires $P<R_{e'}<F_{e'}<R_{e}$. This implies that the path from $n_0$ to $u_e$ cannot be longer since otherwise there would be two edges $e'$, $e''$ with $R_{e'}<F_{e'}$ and $R_{e''}<F_{e''}$, a contradiction. It also implies that edge $e$ is the only edge in $E$ with $e\cap S=\emptyset$.
	
	Now consider an edge $e''=\{w',v''\}\neq e'$ such that there is no direct path from $n_0$ to $u_{v''}$. If such an edge does not exist then $w'$ and all of its neighbors , different than $v$, are in $S$. Hence we can remove $w'$ from $S$ and add $v$ to $S$ to obtain a vertex cover for $G$ of size at most $k$. 
	Assume that edge $e''$ with the described properties exists and consider the temporal path from $u_{e'}$ to $u_{e''}$. This path must start with $F_{e'}$ thus reaching $u_v$. From there the path cannot continue to some $u_{e'''}$ since for all $e'''\neq e'$ we have that $F_{e'''}<R_{e'''}$ hence the path cannot continue from $u_{e'''}$. It follows that the path has to eventually reach $n_1$ continue to $u_{w'}$ from there. However, recall that $P<F_{e'}$ which means that we cannot use $P$ to reach $u_{w'}$ from $n_1$. Hence, there is a second temporal path $P'$ (using the same edges as $P$ with later labels) from $n_1$ to $u_{w'}$ with $F_{e'}<P'$.  This implies that $|S|<k$ and we can add $v$ to $S$ to obtain a vertex cover of size at most $k$ for $G$.
\end{proof}

\subsection{An FPT-algorithm for \MSL\ with respect to the number of terminals}\label{MSL-FPT-subsec}

In this section we provide an FPT-algorithm for \MSL, parameterized by the number $|R|$ of terminals. 
The algorithm is based on a crucial structural property of minimum solutions for \MSL: 
there always exists a minimum labeling $\lambda$ that labels the edges of a subtree of the input graph 
(where every leaf is a terminal vertex), and potentially one further edge that forms a $C_4$ with three edges of the subtree.

Intuitively speaking, we can use an FPT-algorithm for \textsc{Steiner Tree} parameterized by the number of terminals~\cite{Dreyfus1971Steiner} to reveal a subgraph of the \MSL\ instance that we can optimally label using~\cref{thm:optLabGossipC4}. 
Since the number of terminals in the created \textsc{Steiner Tree} instance is larger than the number of terminals in the \MSL\ instance by at most a constant, we obtain an FPT-algorithm for \MSL\ parameterized by the number of terminals.

\begin{lemma}\label{lem:MSLstructure}
	Let $G=(V,E)$ be a graph, $R\subseteq V$ a set of terminals, and $k$ be an integer such that $(G,R,k)$ is a \textsc{YES} instance of \MSL\ and $(G,R,k-1)$ is a \textsc{NO} instance of \MSL.
	\begin{itemize}
		\item If $k$ is odd, then there is a labeling $\lambda$ of size $k$ for $G$ such that the edges labeled by $\lambda$ form a tree, and every leaf of this tree is a vertex in $R$.
		\item If $k$ is even, then there is a labeling $\lambda$ of size $k$ for $G$ such that the edges labeled by $\lambda$ form a graph that is a tree with one additional edge that forms a $C_4$, and every leaf of the tree is a vertex in $R$.
	\end{itemize} 
\end{lemma}

The main idea for the proof of \cref{lem:MSLstructure} is as follows. Given a solution labeling $\lambda$, we fix one terminal $r^*$ and then (i)~we consider the minimum subtree in which $r^*$ can reach all other terminal vertices and (ii)~we consider the minimum subtree in which all other terminal vertices can reach $r^*$. 
Intuitively speaking, we want to label the smaller one of those subtrees using \cref{thm:optLabGossipC4} and potentially adding an extra edge to form a $C_4$; we then argue that the obtained labeling does not use more labels than $\lambda$. To do that, and to detect whether it is possible to add an edge to create a $C_4$, we make a number of modifications to the trees until we reach a point where we can show that our solution is correct.

\begin{proof}
	Assume there is a labeling $\lambda$ for $G$ that labels all edges in the subgraph $H$ of $G$. We describe a procedure to transform $H$ into a tree $T$ by removing edges from $H$ such that $T$ can be labeled with $k$ labels such that all vertices in $R$ are pairwise temporally connected.
	
	Consider a terminal vertex $r^*\in R$. Let $H^+_{r^*}$ be a minimum subgraph of $H$ and $\lambda^+_{r^*}$ a minimum sublabeling of $\lambda$ for $H^+_{r^*}$ such that $r^*$ can temporally reach all vertices in $R\setminus \{r^*\}$ in $(H^+_{r^*},\lambda^+_{r^*})$. Let us first observe that $H^+_{r^*}$ is a tree where all leafs are vertices from $R$ and $\lambda^+_{r^*}$ assigns exactly one label to every edge in $H^+_{r^*}$.
	
	First note that all vertices in $(H^+_{r^*},\lambda^+_{r^*})$ are temporally reachable from $r^*$. If a vertex is not reachable, we can remove it, a contradiction to the minimality of $H^+_{r^*}$. Now assume that $H^+_{r^*}$ is not a tree. Then there is a vertex $v\in V(H^+_{r^*})$ such that $v$ is temporally reachable from $r^*$ in $(H^+_{r^*},\lambda^+_{r^*})$ via two temporal paths $P,P'$ that visit different vertex sets, i.e.\ $V(P)\neq V(P')$. Assume w.l.o.g.\ that both $P$ and $P'$ are foremost among all temporal paths that visit the vertices in $V(P)$ and $V(P')$, respectively, in the same order. Let the arrival time of $P$ be at most the arrival time of $P'$. Then we can remove the last edge traversed by $P'$ with all its labels from $(H^+_{r^*},\lambda^+_{r^*})$ such that afterwards $r^*$ can still temporally reach all vertices in $R\setminus \{r^*\}$, a contradiction to the minimality of $H^+_{r^*}$. From now on, assume that $H^+_{r^*}$ is a tree. Assume that $H^+_{r^*}$ contains a leaf vertex $v$ that is not contained in $R$. Then we can remove $v$ from $(H^+_{r^*},\lambda^+_{r^*})$ such that afterwards $r^*$ can still temporally reach all vertices in $R\setminus \{r^*\}$, a contradiction to the minimality of $H^+_{r^*}$. Lastly, assume that there is an edge $e=uv$ in $H^+_{r^*}$ such that $\lambda^+_{r^*}$ assigns more than one label to $e$. Let $v$ be further away from $r^*$ than $u$ in $H^+_{r^*}$ and let $P$ be a foremost temporal path from $r^*$ to $v$ in $(H^+_{r^*},\lambda^+_{r^*})$ with arrival time $t$. Then we can remove all labels except for $t$ from $e$ and afterwards $r^*$ can still temporally reach all vertices in $R\setminus \{r^*\}$, a contradiction to the minimality of $\lambda^+_{r^*}$.
	
	Let $H^-_{r^*}$ be a minimum subgraph of $H$ and $\lambda^-_{r^*}$ a minimum sublabelling of $\lambda$ for $H^-_{r^*}$ such that each vertex in $R\setminus \{r^*\}$ can temporally reach~$r^*$ in $(H^-_{r^*},\lambda^-_{r^*})$. We can observe by analogous arguments as above that $H^-_{r^*}$ is a tree where all leafs are vertices from $R$ and $\lambda^-_{r^*}$ assigns exactly one label to every edge in $H^-_{r^*}$.
	
	We define the following sets of edges:
	\begin{itemize}
		\item The set of edges only appearing in $H^+_{r^*}$: $E^+_{r^*}=E(H^+_{r^*})\setminus E(H^-_{r^*})$.
		\item The set of edges only appearing in $H^-_{r^*}$: $E^-_{r^*}=E(H^-_{r^*})\setminus E(H^+_{r^*})$.
		\item The set of edges appearing in both $H^+_{r^*}$ and $H^-_{r^*}$: $E^{+-}_{r^*}=E(H^+_{r^*})\cap E(H^-_{r^*})$.
		\item The set of edges appearing in both $H^+_{r^*}$ and $H^-_{r^*}$ that receive the same label from $\lambda^+_{r^*}$ and $\lambda^-_{r^*}$: $E^*_{r^*} = \{e \in E^{+-}_{r^*}\mid \lambda^+_{r^*}(e) = \lambda^-_{r^*}(e)\}$.
	\end{itemize}
	
	We claim that there exists a labelling $\lambda'$ of size $k$ for $G$ such that there are two trees $H^+_{r^*}, H^-_{r^*}$ with the above described properties and $|E(H^+_{r^*})|+|E(H^-_{r^*})|-|E^*_{r^*}|=k-x$ for some $x\ge 0$ and 
	\begin{itemize}
		\item $|E^*_{r^*}|\le x+1$ if $k$ is odd, and
		\item if $k$ is even, then $|E^*_{r^*}|\le x+2$ and there exist two edges $e^+, e^-$ in $H$ that each of them, when added to $H^+_{r^*}, H^-_{r^*}$, respectively, creates a $C_4$ in $H^+_{r^*}, H^-_{r^*}$, respectively.
	\end{itemize}
	We first argue that the statement of the lemma follows from this claim. Afterwards we prove the claim.
	Assume that $|E^+_{r^*}|\le |E^-_{r^*}|$ (the case where $|E^+_{r^*}|> |E^-_{r^*}|$ is analogous). 
	
	Assume that $|E^*_{r^*}|\le x+1$.
	Then we clearly have 
	\[
	2|E(H^+_{r^*})|-1= 2|E^+_{r^*}| + 2|E^{+-}_{r^*}| -1\le |E(H^+_{r^*})|+ |E(H^-_{r^*})| -1= k-x+|E^*_{r^*}|-1\le k.
	\]
	It follows that we can temporally label $H^+_{r^*}$ with at most $k$ labels such that all vertices in $H^+_{r^*}$ can pairwise temporally reach each other, using the result that trees with $m$ edges can be temporally labeled with $2m-1$ labels (see \cref{thm:optLabGossipC4}).
	Since we assume $(G,R,k-1)$ is a \textsc{NO} instance of \MSL\ it follows that $k=2m-1$ and hence this can only happen if $k$ is odd.  
	
	Assume that $|E^*_{r^*}|\le x+2$ and there exist two edges $e^+, e^-$ in $H$ that each of them, when added to $H^+_{r^*}, H^-_{r^*}$, respectively, creates a $C_4$ in $H^+_{r^*}, H^-_{r^*}$, respectively..
	Then we clearly have 
	\[
	2|E(H^+_{r^*})\cup\{e^+\}|-4= 2|E^+_{r^*}| + 2|E^{+-}_{r^*}| -2\le |E(H^+_{r^*})|+ |E(H^-_{r^*})| -2= k-x+|E^*_{r^*}|-2\le k.
	\]
	It follows that we can temporally label $H^+_{r^*}$ together with edge $e^+$ with at most $k$ labels such that all vertices in $H^+_{r^*}$ with edge $e^+$ can pairwise temporally reach each other, using the result that graphs containing a $C_4$ with $n$ vertices can be temporally labeled with $2n-4$ labels (see \cref{thm:optLabGossipC4}). Since we assume $(G,R,k-1)$ is a \textsc{NO} instance of \MSL\ it follows that $k=2n-4$ and hence this can only happen if $k$ is even.

	Now we prove that there exists a labeling $\lambda'$ of size $k$ for $G$ such that there are two trees $H^+_{r^*}, H^-_{r^*}$ with the above described properties and $|E(H^+_{r^*})|+|E(H^-_{r^*})|-|E^*_{r^*}|=k-x$ for some $x\ge 0$ and $|E^*_{r^*}|\le x+1$.
	
	Let $H^+_{r^*}, H^-_{r^*}$ be two trees with the above described properties and $|E(H^+_{r^*})|+|E(H^-_{r^*})|-|E^*_{r^*}|=k-x$ for some $x\ge 0$. We will argue that by slightly modifying the labeling $\lambda$ (and with that $\lambda^+_{r^*}$ and $\lambda^-_{r^*}$, that way ultimately obtaining $\lambda'$) and $H^+_{r^*}, H^-_{r^*}$, we achieve that $|E(H^+_{r^*})|+|E(H^-_{r^*})|-|E^*_{r^*}|=k-x'$ for some $x'\ge 0$ and either $|E^*_{r^*}|\le x'+1$ or $|E^*_{r^*}|\le x'+2$.
	We will argue that in the former case we must have that $k$ is odd, and in the latter case we must have that $k$ is even.
	Note that if $|E^*_{r^*}| = 1$ we are done, hence assume from now on that $|E^*_{r^*}| \ge 2$.
	
	We consider several cases.
	For the sake of presentation of the next cases, define the \emph{head} of a temporal path as the last vertex visited by the path and the \emph{extended head} of a temporal path as the last two vertices visited by the path. Furthermore, define the \emph{tail} of a temporal path as the first vertex visited by the path and the \emph{extended tail} of a temporal path as the first two vertices visited by the path. 
	\medskip
	
	\noindent\textbf{Case A.} Assume there is a temporal path $P$ from $r^*$ to some $r\in R\setminus\{r^*\}$ in $H^+_{r^*}$ that traverses two edges in $E^*_{r^*}$.
	Let $e,e'\in E^*_{r^*}$ with $e\neq e'$ such that there is a temporal path $P$ from $r^*$ to some $r\in R\setminus\{r^*\}$ in $H^+_{r^*}$ that traverses w.l.o.g.\ first $e$ and then $e'$ and a \emph{maximum number} $\alpha$ of edges lies between them in $P$ and the distance $\beta$ between $r^*$ and $e$ is minimum. Note that this implies that $\lambda^+_{r^*}(e)<\lambda^+_{r^*}(e')$. 
	
	In the following we analyse several cases. In some of them we can deduce that the labeling $\lambda$ must use labels that are not present in $\lambda^+_{r^*}$ or $\lambda^-_{r^*}$ that are unique to that case. This implies that for each of these cases we can attribute one label outside of $\lambda^+_{r^*}$ and $\lambda^-_{r^*}$ to edge $e$ or $e'$.
	
	In some other cases we describe modifications that do not increase $|E(H^+_{r^*})\cup E(H^-_{r^*})|$ and either 
	\begin{itemize}
		\item strictly decrease $\beta$, or
		\item strictly decrease $\alpha$ and not increase $\beta$, or
		\item strictly decrease $|E^*_{r^*}|$ and not increase $\alpha$ and $\beta$,
	\end{itemize}
	while preserving that
	\begin{itemize}
		\item $H^+_{r^*}$ and $H^+_{r^*}$ are trees with leafs in $R$, and
		\item $\lambda^+_{r^*}$ and $\lambda^-_{r^*}$ assign at most one label per edge.
	\end{itemize} 
	Whenever a modification satisfies the above requirements it is clear that it can only be applied a finite number of times.
	Whenever we describe a case that requires modifications that do not satisfy the above requirements, we explicitly show that these modifications can only be applied a finite number of times as well.
	Overall this then shows that after a finite number of modifications, none of the described cases will apply.
	
	We partition the temporal path $P$ into the part $P_1$ from $r^*$ to $e$, the part consisting of $e$ itself, the part $P_2$ between $e$ and $e'$, the part consisting of $e'$ itself, and the part $P_3$ from $e'$ to $r$.
	Now in $H^-_{r^*}$ we can have two different scenarios.
	For illustrations of all variations of Case A see~\cref{fig:MSL-caseA--A-1-ii,fig:MSL-caseA-iii-vii,fig:MSL-caseA2}.
	
	\begin{figure}[!htb]
		\begin{subfigure}{0.48\textwidth}
			\centering
			\includegraphics[width=0.88\linewidth]{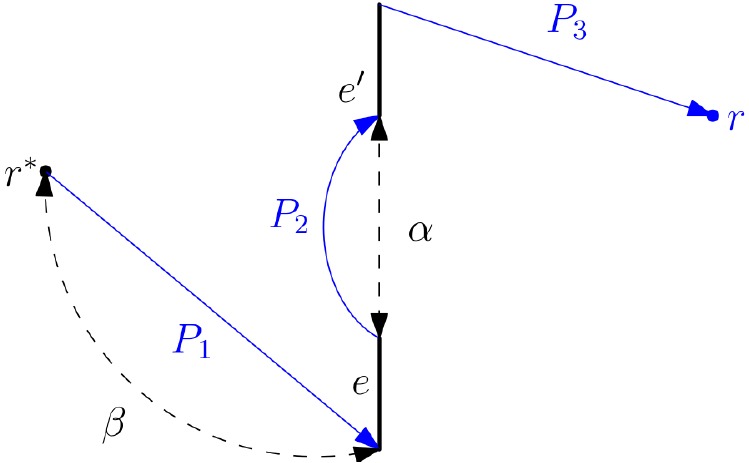}
			\caption{Case A: an example of a path $P$ from $r^*$ in $H^+_{r^*}$, that  traverses $e,e'\in E^*_{r^*}$.}
			\label{fig:MSL-CaseA}
		\end{subfigure}%
		\quad
		\begin{subfigure}{0.48\textwidth}
			\centering
			\includegraphics[width=0.88\linewidth]{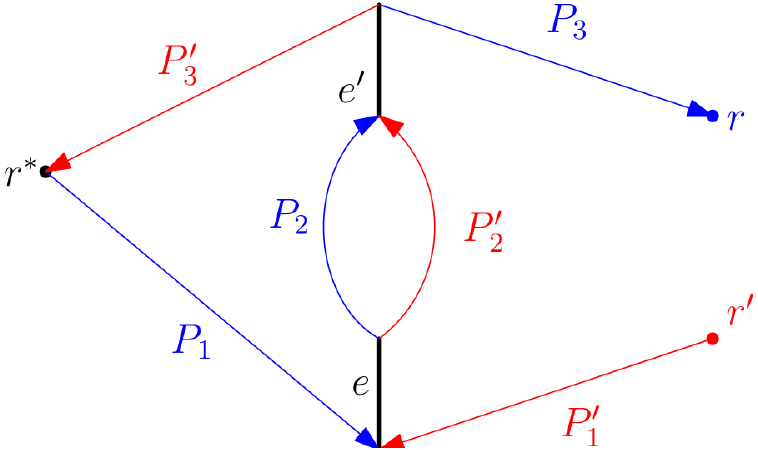}
			\caption{Case A-1: an example of $P$ in $H^+_{r^*}$ and $P'$ in $H^-_{r^*}$, that  share $e,e'\in E^*_{r^*}$.}
			\label{fig:MSL-CaseA-1}
		\end{subfigure}
		
		\begin{subfigure}{0.48\textwidth}
			\centering
			\includegraphics[width=0.88\linewidth]{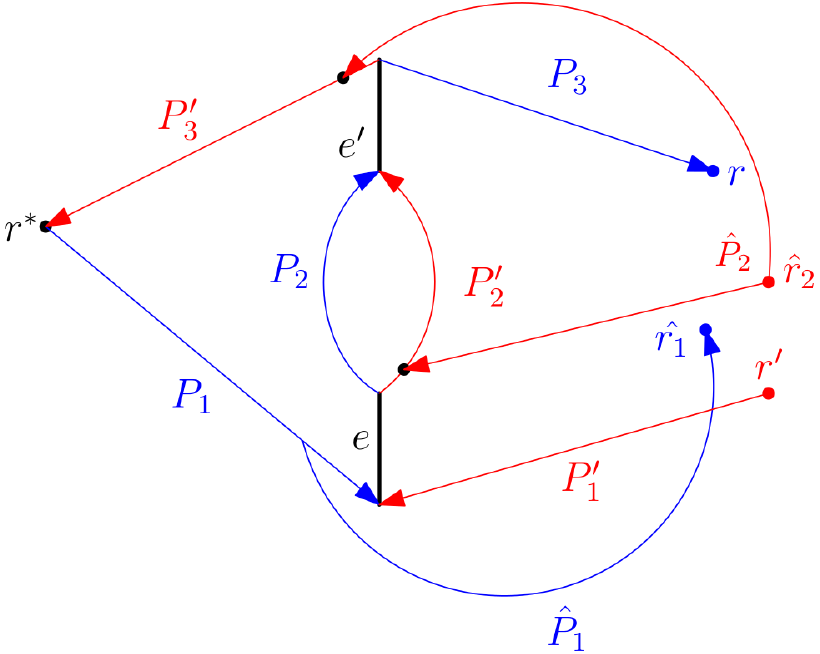}
			\caption{Case A-1-i: $P^*$ from $\hat{r}_2$ to $\hat{r}_1$ either uses no labels from $\lambda^+_{r^*}$ or no from $\lambda^-_{r^*}$.}
			\label{fig:MSL-CaseA-1-i}
		\end{subfigure}%
		\quad
		\begin{subfigure}{0.48\textwidth}
			\centering
			\includegraphics[width=0.88\linewidth]{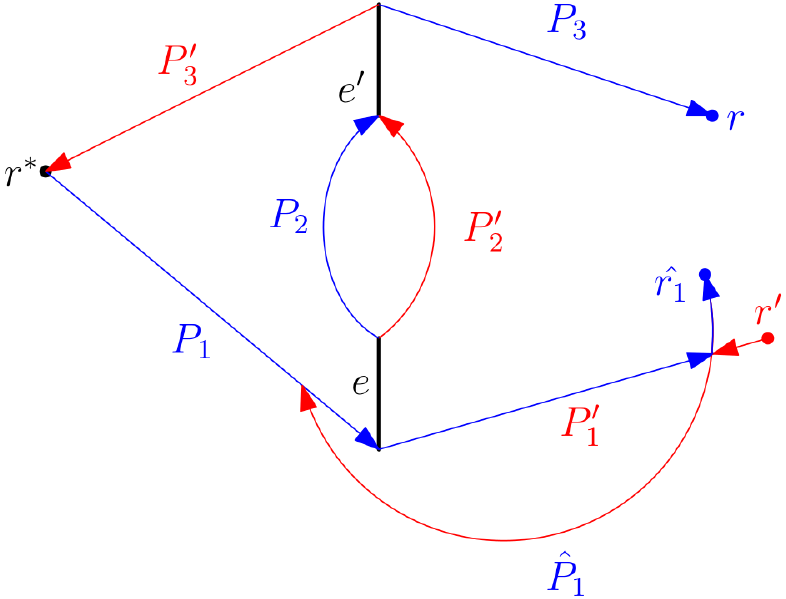}
			\caption{Modification of Case A-1-i.}
			\label{fig:MSL-CaseA-1-i_modification}
		\end{subfigure}%
		
		\begin{subfigure}{0.48\textwidth}
			\centering
			\includegraphics[width=0.88\linewidth]{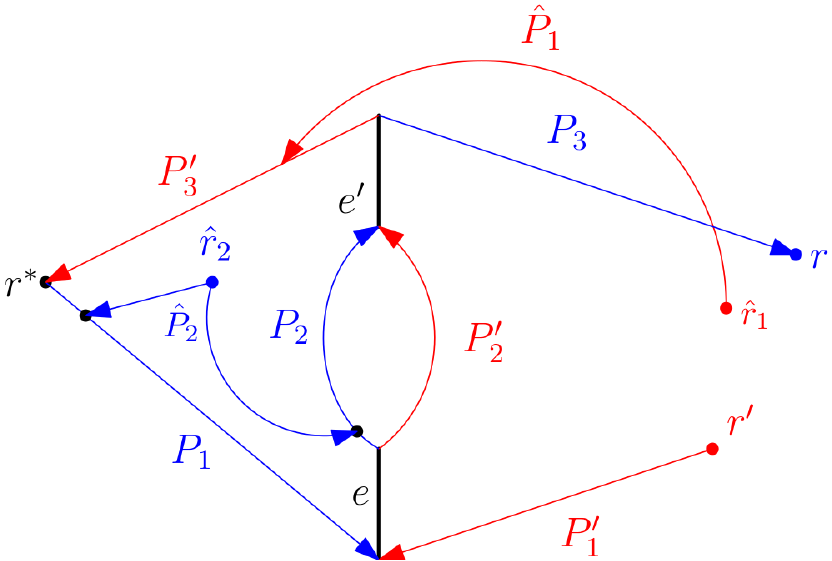}
			\caption{Case A-1-ii: $P^*$ from $\hat{r}_1$ to $\hat{r}_2$ either uses no labels from $\lambda^+_{r^*}$ or no from $\lambda^-_{r^*}$.}
			\label{fig:MSL-CaseA-1-ii}
		\end{subfigure}
		\quad
		\begin{subfigure}{0.48\textwidth}
			\centering
			\includegraphics[width=0.88\linewidth]{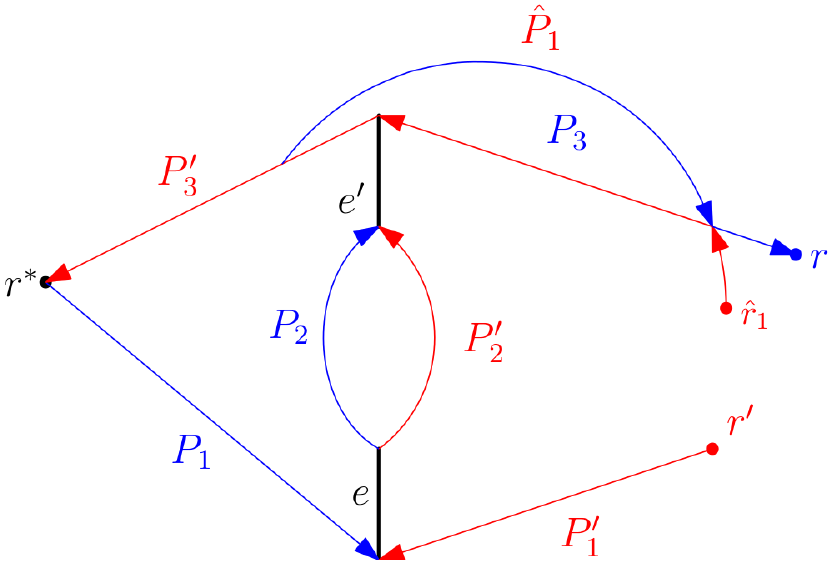}
			\caption{Modification of Case A-1-ii.}
			\label{fig:MSL-CaseA-1-ii_modification}
		\end{subfigure}
		\caption{Cases A-1 -- A-1-ii, where \emph{blue} color corresponds to the labeling $\lambda^+_{r^*}$ and \emph{red} to $\lambda^-_{r^*}$. \label{fig:MSL-caseA--A-1-ii}}
	\end{figure}
	
	\medskip
	
	\noindent\textbf{Case A-1.} There is a temporal path $P'$ from some $r'\in R\setminus\{r^*\}$ to $r^*$ in $H^-_{r^*}$ that traverses both $e$ and $e'$. Note that this implies that $e$ is traversed before $e'$.
	
	We partition the temporal path $P'$ into the part $P'_1$ from $r'$ to $e$, the part consisting of $e$ itself, the part $P'_2$ between $e$ and $e'$, the part consisting of $e'$ itself, and the part $P'_3$ from $e'$ to~$r^*$.

	The analysis of each one follows from the observation that the labels in $P_3'$ are larger than the ones in $P_1$.
	\medskip

	\noindent\textbf{Case A-1-i.}  Assume there is a path $\hat{P}_1$ in $H^+_{r^*}$ starting at a vertex that is visited by $P_1$ and ending at $\hat{r}_1\in R\setminus\{r^*\}$ such that $\hat{r}_1=r'$ or $\hat{P}_1$ and $P_1'$ intersect in a vertex.
	For our analysis, we treat these two cases the same since in both cases we can assume that $r'$ can reach $\hat{r}_1$, in the latter through the intersection point.
	If there is a path $\hat{P}_2$ in $H^-_{r^*}$ starting at some $\hat{r}_2\in R\setminus\{r^*,r'\}$ and ending at the extended tail of $P_2'$ or $P_3'$, then the temporal path $P^*$ in $(G,\lambda)$ from $\hat{r}_2$ to $\hat{r}_1$ either uses no labels from $\lambda^+_{r^*}$ or no from $\lambda^-_{r^*}$.
	
	\noindent\textbf{Case A-1-ii.}
	Assume there is a path $\hat{P}_1$ in $H^-_{r^*}$ starting at $\hat{r}_1\in R\setminus\{r^*\}$ and ending at a vertex that is visited by $P_3'$, such that  $\hat{r}_1=r$ or $\hat{P}_1$ and $P_3$ intersect in a vertex. 
	Again for our analysis, we treat these two cases the same since in both cases we can assume that $\hat{r}_1$ can reach $r$, in the latter through the intersection point.
	If there is a path $\hat{P}_2$ in $H^+_{r^*}$ starting at the extended tail of $P_1$ or $P_2$ and ending at some $\hat{r}_2\in R\setminus\{r^*,r\}$, then the temporal path $P^*$ in $(G,\lambda)$ from $\hat{r}_1$ to $\hat{r}_2$ either uses no labels from $\lambda^+_{r^*}$ or no from $\lambda^-_{r^*}$.
	
	Assume that one of the above two applies. 
	We assume that there is no path $\hat{P}_2$ in $H^-_{r^*}$ starting at some $\hat{r}_2\in R\setminus\{r^*,r'\}$ and ending at the extended tail of $P_2'$ or $P_3'$ in Case A-1-i and that there is no  path $\hat{P}_2$ in $H^+_{r^*}$ starting at the extended tail of $P_1$ or $P_2$ and ending at some $\hat{r}_2\in R\setminus\{r^*,r\}$, since in both cases we can directly deduce that we need labels outside of $\lambda^+_{r^*}$ and $\lambda^-_{r^*}$.
	Then we modify $\lambda$ in the following way without changing its connectivity properties. 
	First, we scale all labels in $\lambda$ by a factor of~$|V|$.
	
	The idea is first to essentially switch the roles of $P'_1$ and $\hat{P}_1$ in Case A-1-i and switch the roles of $P_3$ and $\hat{P}_1$ in Case A-1-ii. Assume Case A-1-i applies.
	\begin{itemize}
		\item We remove $\hat{P}_1$'s edges and labels from $H^+_{r^*}$ and  $\lambda^+_{r^*}$, respectively, add $\hat{P}_1$'s edges to $H^-_{r^*}$. Add the edges between the (original) tail of $\hat{P}_1$ to $e$ to $H^-_{r^*}$ and add the respective labels for those edges from $\lambda^+_{r^*}$ also to $\lambda^-_{r^*}$. Add new labels for the edges of $\hat{P}_1$ to $\lambda^-_{r^*}$ such that there is temporal paths from $r'$ to $r^*$ that does use edges from $P_1'$.
		\item We remove $P_1'$'s edges and labels from $H^-_{r^*}$ and  $\lambda^-_{r^*}$, respectively, add $P_1'$'s edges to $H^+_{r^*}$, and add new labels for the edges of $P_1'$ to $\lambda^+_{r^*}$ such that there is a temporal path from $r^*$ to $r'$.
	\end{itemize}
	Now assume Case A-1-ii applies. We make analogous modifications.
	\begin{itemize}
		\item We remove $\hat{P}_1$'s edges and labels from $H^-_{r^*}$ and  $\lambda^-_{r^*}$, respectively, add $\hat{P}_1$'s edges to $H^+_{r^*}$. Add the edges from the head of $\hat{P}_1$ to $e'$ to $H^+_{r^*}$ and add the respective labels for those edges from $\lambda^-_{r^*}$ also to $\lambda^+_{r^*}$. Add new labels for the edges of $\hat{P}_1$ to $\lambda^+_{r^*}$ such that there is temporal paths from $r^*$ to $r$ that does use edges from $P_3$.
		\item We remove $P_3$'s edges and labels from $H^+_{r^*}$ and  $\lambda^+_{r^*}$, respectively, add $P_3$'s edges to $H^-_{r^*}$, and add new labels for the edges of $P_3$ to $\lambda^-_{r^*}$ such that there are temporal paths from $r$ to $r^*$.
	\end{itemize}
	Note that after the modifications $H^+_{r^*}$ and $H^-_{r^*}$ are still trees, and $\lambda^+_{r^*}$ and $\lambda^-_{r^*}$ still assign at most one label per edge. Furthermore, we have that the modification do not increase the sum of edges in both trees $|E(H^+_{r^*})\cup E(H^-_{r^*})|$. Note that these modifications potentially increase $|E^*_{r^*}|$ and $\alpha$. However, note that in both cases we strictly decrease $\beta$. From now on assume that Cases A-1-i and A-1-ii do not apply.
	
	\medskip
	
	We start with three further subcases. The analysis of each one follows from the observation that the labels in $P_3'$ are larger than the ones in $P_1$.
	
	\noindent\textbf{Case A-1-iii.} Assume there is a path $\hat{P}$ in $H^+_{r^*}$ starting at a vertex that is visited by $P_1$ but is different from its tail and extended head and ending at some $\hat{r}\in R\setminus\{r^*,r\}$. Then
	the temporal path $P^*$ in $(G,\lambda)$ from $r'$ to $\hat{r}$ needs at least one label that is not contained in $\lambda^+_{r^*}$ or $\lambda^-_{r^*}$. 
	More specifically, $P^*$ either uses no labels from $\lambda^+_{r^*}$ or no from $\lambda^-_{r^*}$.
	
	\noindent\textbf{Case A-1-iv.} Assume there is a path $\hat{P}$ in $H^-_{r^*}$ starting at some $\hat{r}\in R\setminus\{r^*,r'\}$ and ending at a vertex that is visited by $P'_3$ but is different from its extended tail and head. Then
	the temporal path $P^*$ in $(G,\lambda)$ from $\hat{r}$ to $r$ needs at least one label that is not contained in $\lambda^+_{r^*}$ or $\lambda^-_{r^*}$. 
	More specifically, $P^*$ either uses no labels from $\lambda^+_{r^*}$ or no from $\lambda^-_{r^*}$.
	
	\noindent\textbf{Case A-1-v.} Assume there is a path $\hat{P}_1$ in $H^+_{r^*}$ starting at a vertex that is visited by $P_2$ but is different from its tail and extended head and ending at some $\hat{r}_1\in R\setminus\{r^*,r\}$. Furthermore, assume there is a path $\hat{P}_2$ in $H^-_{r^*}$ starting at some $\hat{r}_2\in R\setminus\{r^*,r'\}$ and ending at a vertex that is visited by $P'_2$ but is different from its extended tail and head. Then, if $\hat{r}_2\neq \hat{r}_1$ and $P_2\neq P_2'$, or the starting vertex of $\hat{P}_1$ is by at least two edges closer to $e$ than the starting vertex of $\hat{P}_2$, the temporal path $P^*$ in $(G,\lambda)$ from $\hat{r}_2$ to $\hat{r}_1$ needs at least one label that is not contained in $\lambda^+_{r^*}$ or $\lambda^-_{r^*}$. 
	More specifically, $P^*$ either uses no labels from $\lambda^+_{r^*}$ or no from $\lambda^-_{r^*}$.
	
	\begin{figure}[!htbp]
		\centering
		\begin{subfigure}{0.48\textwidth}
			\centering
			\includegraphics[width=0.88\linewidth]{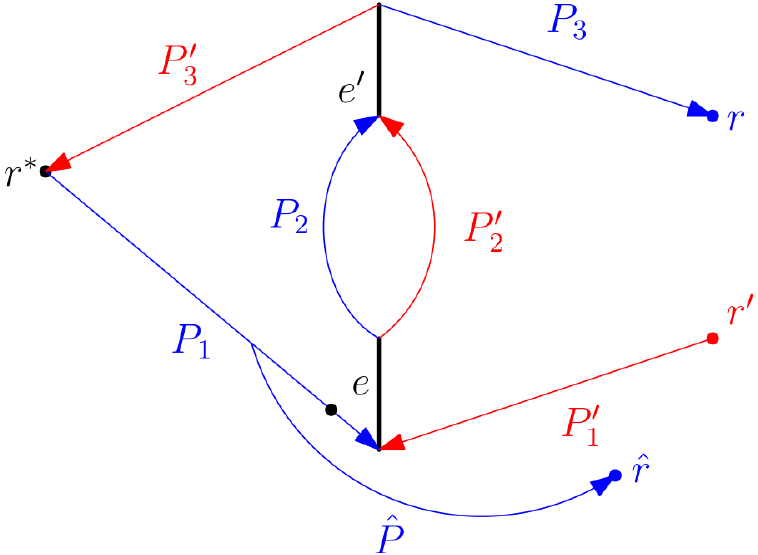}
			\caption{Case A-1-iii: $P^*$ from $r'$ to $\hat{r}$ either uses no labels from $\lambda^+_{r^*}$ or no from $\lambda^-_{r^*}$.}
			\label{fig:MSL-CaseA-1-iii}
		\end{subfigure}%
		\quad
		\begin{subfigure}{0.48\textwidth}
			\centering
			\includegraphics[width=0.88\linewidth]{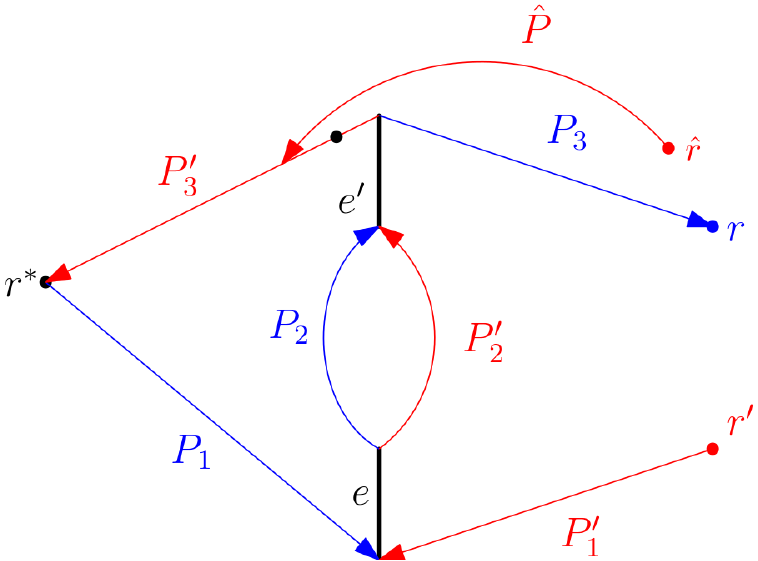}
			\caption{Case A-1-iv: $P^*$ from $\hat{r}$ to $r$ either uses no labels from $\lambda^+_{r^*}$ or no from $\lambda^-_{r^*}$.}
			\label{fig:MSL-CaseA-1-iv}
		\end{subfigure}
		
		\begin{subfigure}{0.48\textwidth}
			\centering
			\includegraphics[width=0.88\linewidth]{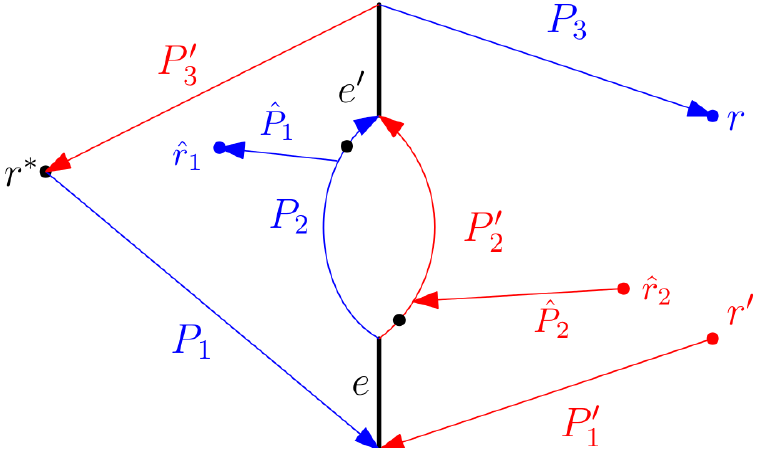}
			\caption{Case A-1-v: $P^*$ from $\hat{r}_2$ to $\hat{r}_1$ either uses no labels from $\lambda^+_{r^*}$ or no from $\lambda^-_{r^*}$.}
			\label{fig:MSL-CaseA-1-v}
		\end{subfigure}%
		\quad
		\begin{subfigure}{0.48\textwidth}
			\centering
			\includegraphics[width=0.88\linewidth]{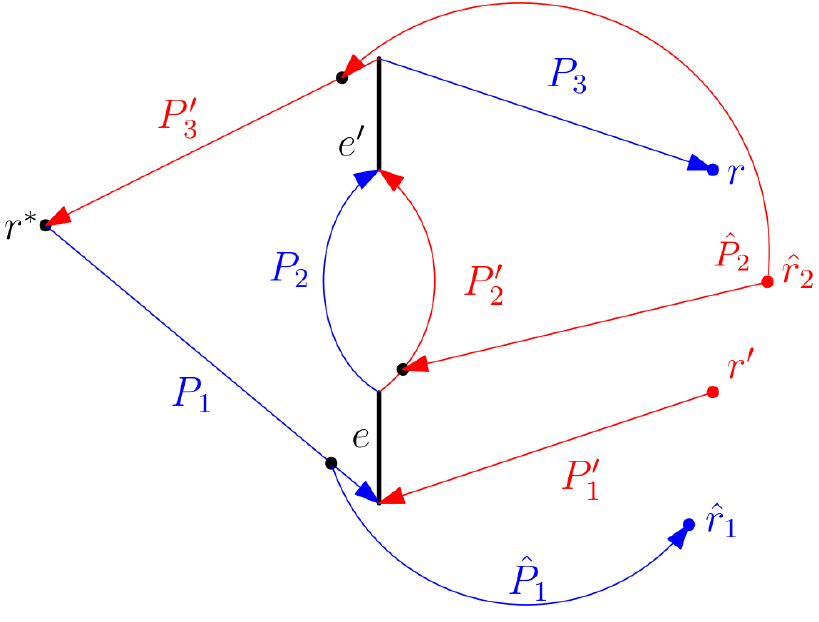}
			\caption{Case A-1-vi: $P^*$ from $\hat{r}_2$ to $\hat{r}_1$ either uses no labels from $\lambda^+_{r^*}$ or no from $\lambda^-_{r^*}$.}
			\label{fig:MSL-CaseA-1-vi}
		\end{subfigure}%
		
		\begin{subfigure}{0.48\textwidth}
			\centering
			\includegraphics[width=0.88\linewidth]{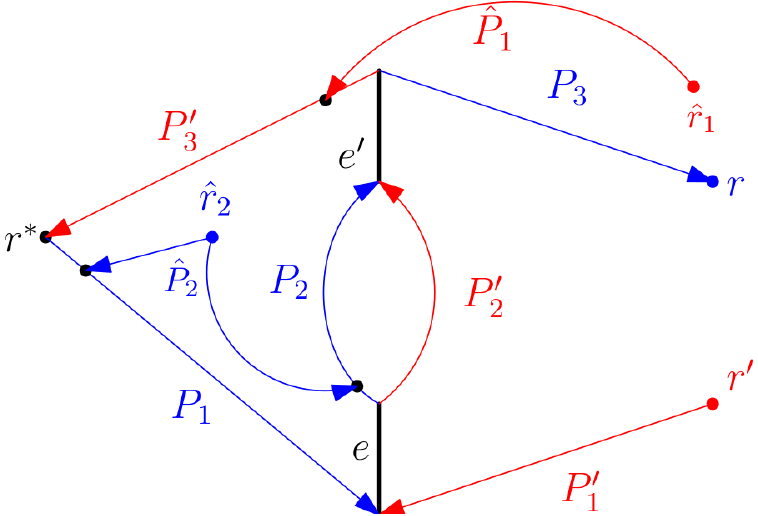}
			\caption{Case A-1-vii: $P^*$ from $\hat{r}_1$ to $\hat{r}_2$ either uses no labels from $\lambda^+_{r^*}$ or no from $\lambda^-_{r^*}$.}
			\label{fig:MSL-CaseA-1-vii}
		\end{subfigure}
		\quad
		\begin{subfigure}{0.48\textwidth}
			\centering
			\includegraphics[width=0.88\linewidth]{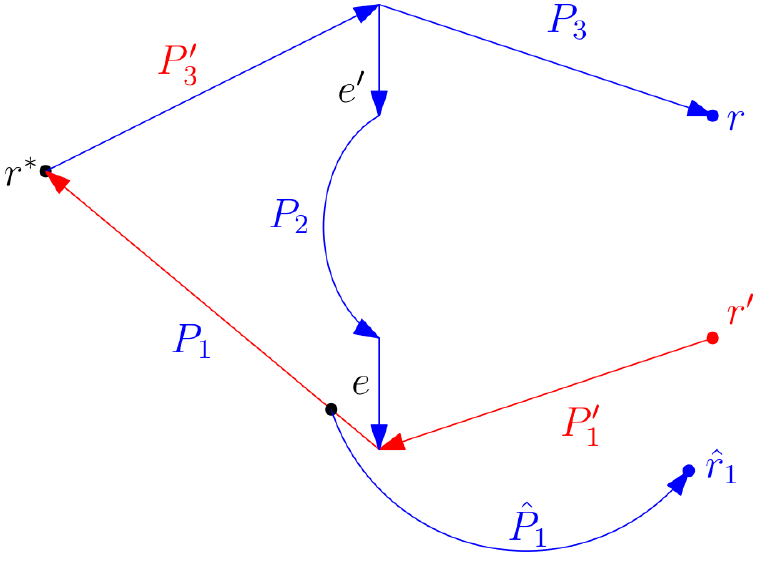}
			\caption{Modification of Case A-1-vi.}
			\label{fig:MSL-CaseA-1-vi_modification}
		\end{subfigure}%
		
		\begin{subfigure}{0.48\textwidth}
			\centering
			\includegraphics[width=0.88\linewidth]{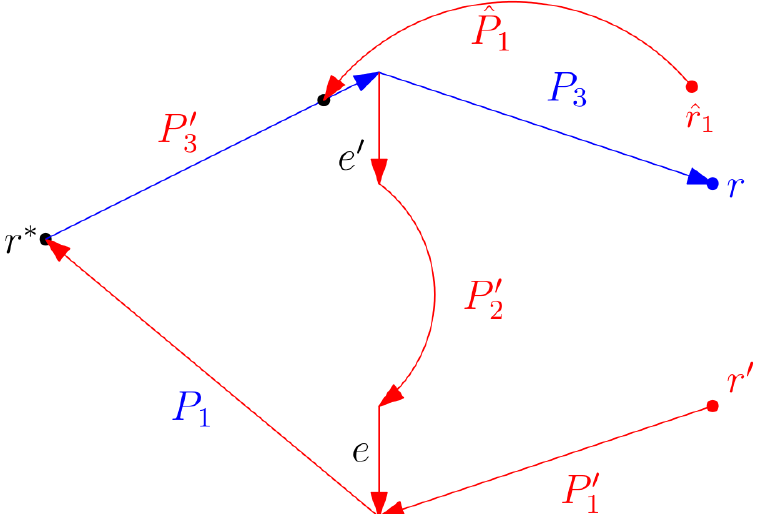}
			\caption{Modification of Case A-1-vii.}
			\label{fig:MSL-CaseA-1-vii_modification}
		\end{subfigure}
		\quad
		\begin{subfigure}{0.48\textwidth}
			\centering
			\includegraphics[width=0.88\linewidth]{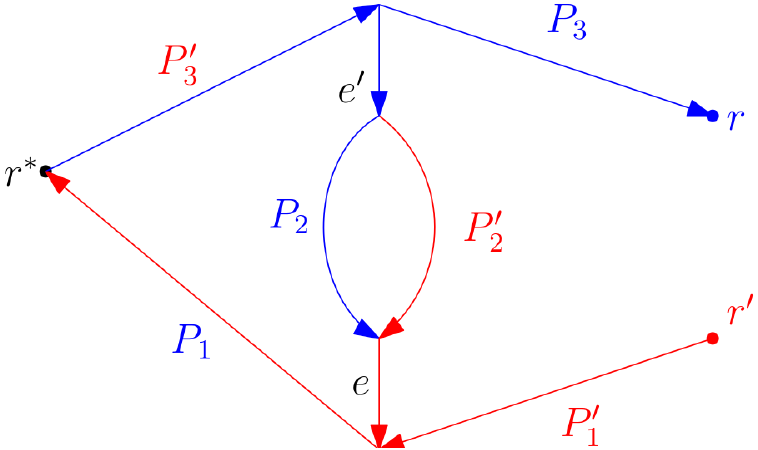}
			\caption{Modification when none of the cases A-1-vi nor A-1-vii apply.}
			\label{fig:MSL-caseA-1_modification}
		\end{subfigure}
		
		\caption{Cases A-1-iii -- A-1-vii, where \emph{blue} color corresponds to the labeling $\lambda^+_{r^*}$ and \emph{red} to $\lambda^-_{r^*}$. \label{fig:MSL-caseA-iii-vii}}
	\end{figure}
	
	\medskip
	
	In the above three Cases A-1-iii to A-1-v we do not make any modifications, since we can directly deduce that we need labels outside of $\lambda^+_{r^*}$ and $\lambda^-_{r^*}$.
	For the remainder of this case distinction, we assume that Cases A-1-iii to A-1-v do not apply.
	
	We can further observe the following using analogous arguments as above.
	
	\noindent\textbf{Case A-1-vi.}  Assume there is a path $\hat{P}_1$ in $H^+_{r^*}$ starting at the extended head of $P_1$ and ending at some $\hat{r}_1\in R\setminus\{r^*,r,r'\}$. If there is a path $\hat{P}_2$ in $H^-_{r^*}$ starting at some $\hat{r}_2\in R\setminus\{r^*,r'\}$ and ending at a vertex from $P_2'$ that is not its tail or a vertex from $P_3'$, then, if $\hat{r}_2\neq \hat{r}_1$, the temporal path $P^*$ in $(G,\lambda)$ from $\hat{r}_2$ to $\hat{r}_1$ either uses no labels from $\lambda^+_{r^*}$ or no from $\lambda^-_{r^*}$.
	
	\noindent\textbf{Case A-1-vii.}
	Assume there is a path $\hat{P}_1$ in $H^-_{r^*}$ starting at some $\hat{r}_1\in R\setminus\{r^*,r,r'\}$ and ending at the extended tail of $P_3'$. If there is a $\hat{P}_2$ in $H^+_{r^*}$ starting at a vertex from $P_1$ or a vertex from $P_2$ that is not its head and ending at some $\hat{r}_2\in R\setminus\{r^*,r\}$, then, if $\hat{r}_1\neq\hat{r}_2$, the temporal path $P^*$ in $(G,\lambda)$ from $\hat{r}_1$ to $\hat{r}_2$ either uses no labels from $\lambda^+_{r^*}$ or no from $\lambda^-_{r^*}$.
		
	\medskip
	
	First, assume that Case A-1-vi or Case A-1-vii or none of them apply. 
	Then we modify $\lambda$ in the following way without changing its connectivity properties. 
	First, we scale all labels in $\lambda$ by a factor of~$|V|$.
		
	The idea is first to essentially switch the roles of $P_1$ and $P_3'$.
	\begin{itemize}
		\item We remove $P_1$'s edges and labels from $H^+_{r^*}$ and  $\lambda^+_{r^*}$, respectively, add $P_1$'s edges to $H^-_{r^*}$, and add new labels for the edges of $P_1$ to $\lambda^-_{r^*}$ such that there are temporal paths from both endpoints of $e$ to $r^*$ that only use the new labels.
		\item We remove $P_3'$'s edges and labels from $H^-_{r^*}$ and  $\lambda^-_{r^*}$, respectively, add $P_3'$'s edges to $H^+_{r^*}$, and add new labels for the edges of $P_3'$ to $\lambda^+_{r^*}$ such that there are temporal paths from $r^*$ to both endpoints of $e$ that only use the new labels.
	\end{itemize}
	In both modification above, we assume w.l.o.g.\ that the smallest and the largest label assigned to an edge of $P_1$ by $\lambda^+_{r^*}$ before the modification are equal the smallest and the largest label, respectively, assigned to an edge of $P_3'$ by $\lambda^+_{r^*}$ after the modification. Symmetrically, we assume w.l.o.g.\ that the smallest and the largest label assigned to an edge of $P_3'$ by $\lambda^-_{r^*}$ before the modification are equal the smallest and the largest label, respectively, assigned to an edge of $P_1$ by $\lambda^-_{r^*}$ after the modification. 
	Note that now there is a path from $r^*$ to $r$ in $(H^+_{r^*},\lambda^+_{r^*})$ that does not use edges $e$ and $e'$. Furthermore, there is a path from $r'$ to $r^*$ in $(H^-_{r^*},\lambda^-_{r^*})$ that does not use edges $e$ and $e'$.
	
	Now we have to adjust labels on $e$, $e'$, $P_2$, and $P_2'$, depending on whether Case A-1-vi, Case A-1-vii or none of them apply.
	\begin{itemize}
		\item If Case A-1-vi applies, then we remove $e$, $e'$, and the edges of $P_2'$ and their labels from $H^-_{r^*}$ and  $\lambda^-_{r^*}$, respectively. Furthermore, we exchange the labels of $e$ and $e'$ and the edges of $P_2$ assigned by $\lambda^+_{r^*}$ in a way that there is a temporal path from $r^*$ to $\hat{r}_1$ (see Case A-1-vi) in $(H^+_{r^*},\lambda^+_{r^*})$.
		\item If Case A-1-vii applies, then we remove $e$, $e'$, and the edges of $P_2$ and their labels from $H^+_{r^*}$ and  $\lambda^+_{r^*}$, respectively. Furthermore, we exchange the labels of $e$ and $e'$ and the edges of $P_2'$ assigned by $\lambda^-_{r^*}$ in a way that there is a temporal path from $\hat{r}_1$ (see Case A-1-vii) to $r^*$ in $(H^-_{r^*},\lambda^-_{r^*})$.
		\item If none of the Cases A-1-vi and A-1-vii apply, then we remove $e$ its labels from $H^+_{r^*}$ and  $\lambda^+_{r^*}$, respectively, and we remove $e'$ its labels from $H^-_{r^*}$ and  $\lambda^-_{r^*}$, respectively.
		We modify the labels of $P_2$ assigned by $\lambda^+_{r^*}$ is a way that all terminals that were reachable from $r^*$ before the modifications can now be reached via $e'$.
		We modify the labels of $P_2'$ assigned by $\lambda^-_{r^*}$ is a way that all terminals that could reach $r^*$ before the modifications can now reach $r^*$ via $e$.
	\end{itemize}
	Note that after the modifications $H^+_{r^*}$ and $H^-_{r^*}$ are still trees, and $\lambda^+_{r^*}$ and $\lambda^-_{r^*}$ still assign at most one label per edge. Furthermore, we have that the modification do not increase the sum of edges in both trees $|E(H^+_{r^*})\cup E(H^-_{r^*})|$. Lastly, and most importantly, we have that at least one of $H^+_{r^*}$ and $H^-_{r^*}$ does contain both edges $e$ and $e'$. It follows that we strictly decrease $|E^*_{r^*}|$ without increasing $\alpha$.
	
	It follows that after exhaustively performing the above modifications we have that if Case A-1 applies, then one of the Cases A-1-iii to A-1-v has to apply.
	
	\medskip
	
	\noindent\textbf{Case A-2.} There are two temporal paths $P',P''$ from some $r',r''\in R\setminus\{r^*\}$, respectively, to $r^*$ in $H^-_{r^*}$ such that $P'$ traverses $e$ and $P''$ traverses $e'$. We consider several different subcases. Let $e=uv$ and let $u$ be the vertex that is closer to $r^*$ in $H^+_{r^*}$. Partition $P'$ into $P_1'$ from $r'$ to $e$, then $e$, and then $P_2'$ from $e$ to $r^*$.
	
	\begin{figure}[!htbp]
		\centering
		\begin{subfigure}{0.48\textwidth}
			\centering
			\includegraphics[width=0.88\linewidth]{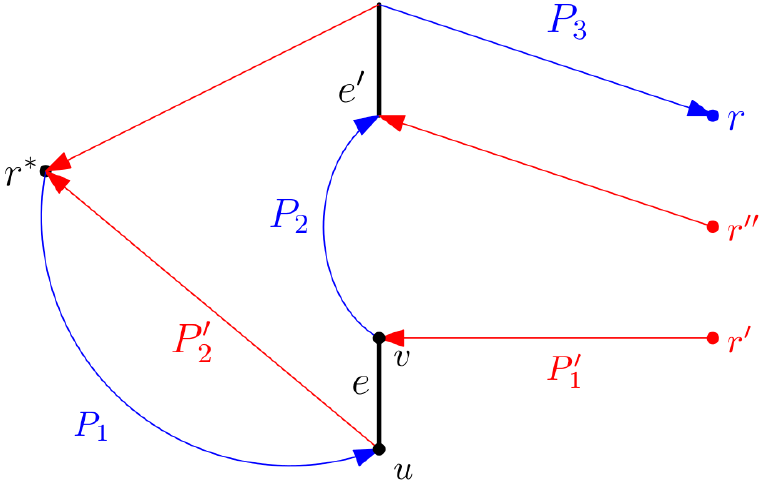}
			\caption{Example of Case A-2-i.}
			\label{fig:MSL-CaseA-2-i}
		\end{subfigure}%
		\quad
		\begin{subfigure}{0.48\textwidth}
			\centering
			\includegraphics[width=0.88\linewidth]{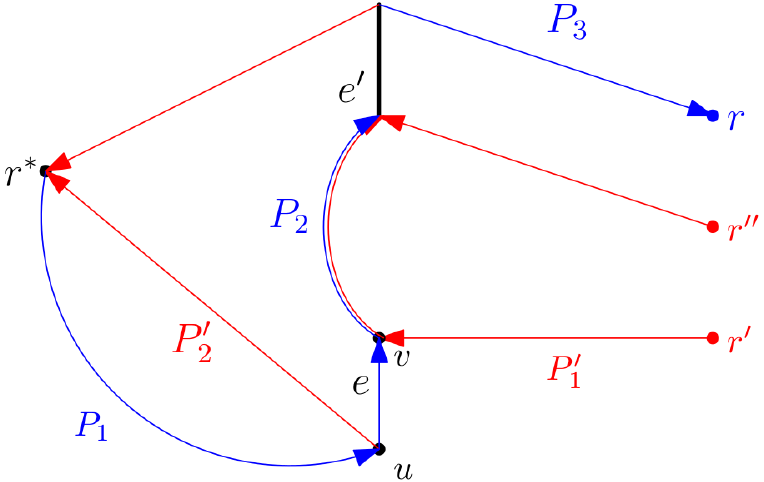}
			\caption{Modification of Case A-2-i.}
			\label{fig:MSL-CaseA-2-i_modification}
		\end{subfigure}
		
		\begin{subfigure}{0.48\textwidth}
			\centering
			\includegraphics[width=0.88\linewidth]{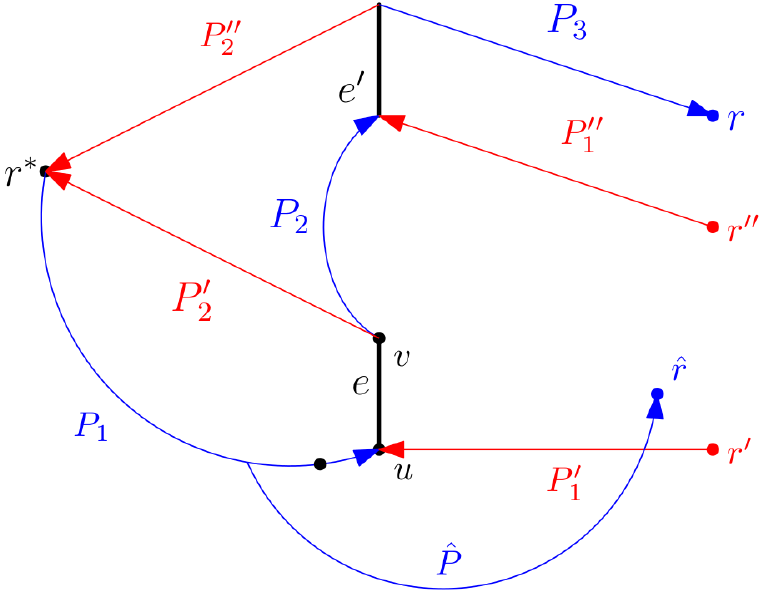}
			\caption{Example of Case A-2-ii.}
			\label{fig:MSL-CaseA-2-ii}
		\end{subfigure}%
		\quad
		\begin{subfigure}{0.48\textwidth}
			\centering
			\includegraphics[width=0.88\linewidth]{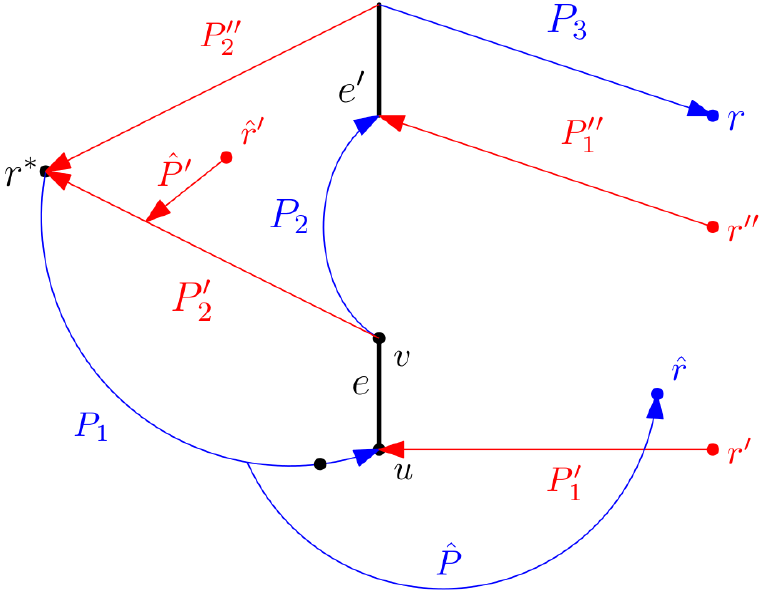}
			\caption{Case A-2-ii(a): $P^*$ from $\hat{r}'$ to $r'$ either uses no labels from $\lambda^+_{r^*}$ or no from $\lambda^-_{r^*}$.
				\label{fig:MSL-CaseA-2-ii-a}}
		\end{subfigure}%
		
		\begin{subfigure}{0.48\textwidth}
			\centering
			\includegraphics[width=0.88\linewidth]{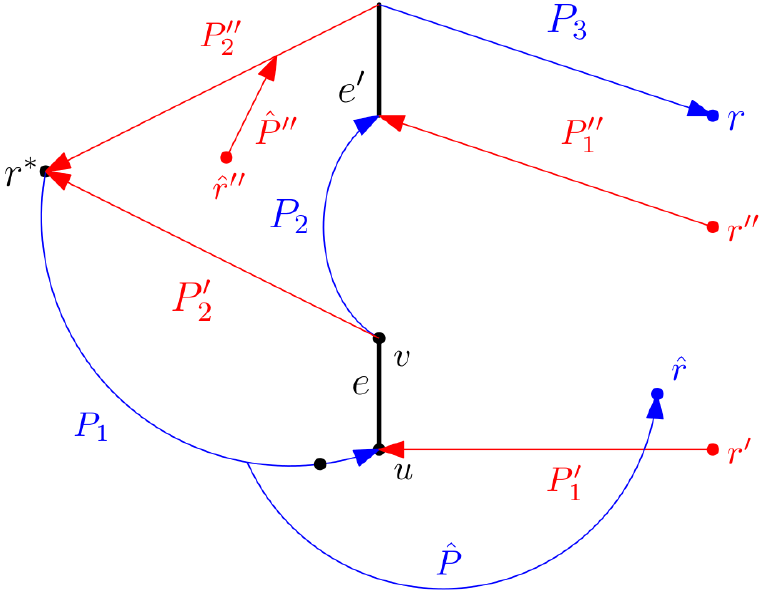}
			\caption{Case A-2-ii(b): $P^*$ from $\hat{r}''$ to $r'$ either uses no labels from $\lambda^+_{r^*}$ or no from $\lambda^-_{r^*}$.}
			\label{fig:MSL-CaseA-2-ii-b}
		\end{subfigure}%
		\quad
		\begin{subfigure}{0.48\textwidth}
			\centering
			\includegraphics[width=0.88\linewidth]{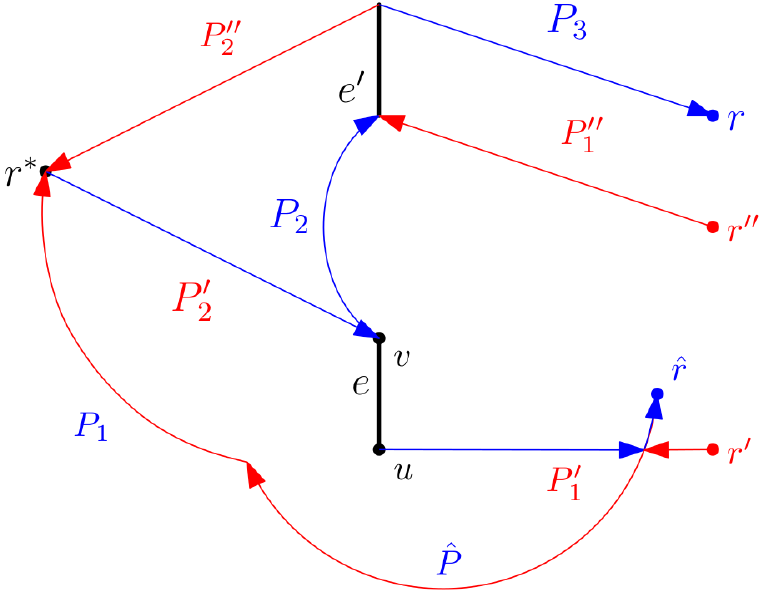}
			\caption{Modification of Case A-2-ii when A-2-ii(a) and A-2-ii(b) do not apply.}
			\label{fig:MSL-CaseA-2-ii_modification}
		\end{subfigure}%
		
		\begin{subfigure}{0.48\textwidth}
			\centering
			\includegraphics[width=0.88\linewidth]{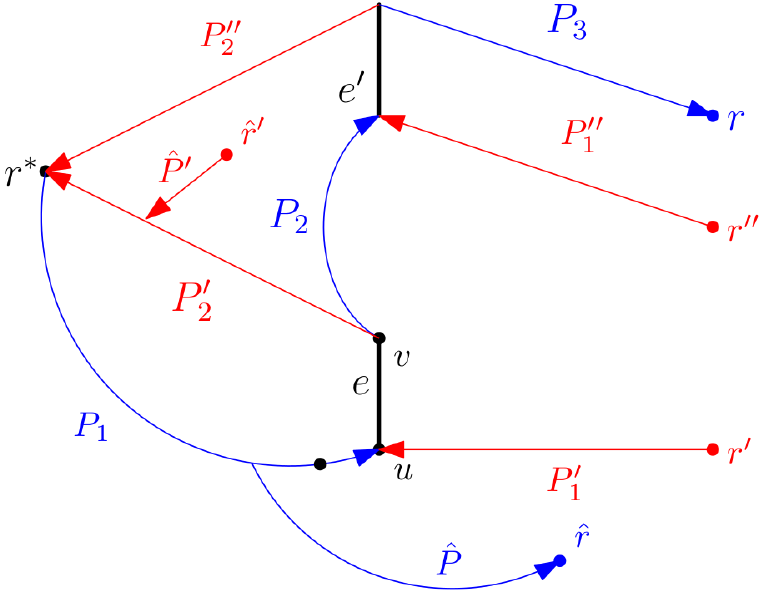}
			\caption{Case A-2-iii: $P^*$ from $\hat{r}'$ to $r$ either uses no labels from $\lambda^+_{r^*}$ or no from $\lambda^-_{r^*}$.}
			\label{fig:MSL-CaseA-2-iii}
		\end{subfigure}%
		\quad
		\begin{subfigure}{0.48\textwidth}
			\centering
			\includegraphics[width=0.88\linewidth]{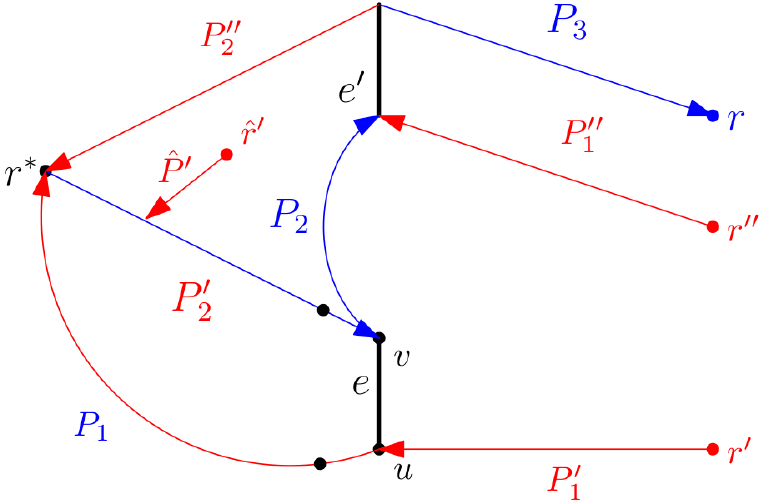}
			\caption{Modification of Case A-2-iii.}
			\label{fig:MSL-CaseA-2-iii_modification}
		\end{subfigure}%
		
		\caption{Cases A-2-i -- A-2-iii, where \emph{blue} color corresponds to the labeling $\lambda^+_{r^*}$ and \emph{red} to $\lambda^-_{r^*}$. \label{fig:MSL-caseA2}
		}
	\end{figure}
	
	\noindent\textbf{Case A-2-i.}  Assume the head of $P'_1$ is $v$. 
	
	We remove $e$ and its labels from $H^-_{r^*}$ and  $\lambda^-_{r^*}$, respectively. To obtain a new path in $(H^-_{r^*},\lambda^-_{r^*})$, we traverse $P'_1$, then traverse $P_2$ (by modifying $\lambda^-_{r^*}$ on $P'_1$ accordingly) which lets us reach $P''$ and then we traverse $P''$ to reach~$r^*$.
	
	Note that after the modifications $H^-_{r^*}$ is still a tree and $\lambda^-_{r^*}$ still assign at most one label per edge. However, the size of $E^*_{r^*}$ changes, in particular it can increase, but the maximal number $\alpha$ of edges between two edges from $E^*_{r^*}$ in $P$ decreases by one.
	
	\noindent\textbf{Case A-2-ii.} Assume the head of $P'_1$ is $u$. Assume there is a path $\hat{P}$ in $H^+_{r^*}$ starting at a vertex that is visited by $P_1$ but is different from its tail and extended head and ending at some $\hat{r}\in R\setminus\{r^*,r\}$, such that $\hat{r}=r'$ or $\hat{P}$ and $P'_1$ intersect in a vertex.
	For our analysis, we treat these two cases the same since in both cases we can assume that $r'$ can reach $\hat{r}$, in the latter through the intersection point. 
	
	\noindent\textbf{Case A-2-ii(a).}
	Furthermore, assume there is a path $\hat{P}'$ in $H^-_{r^*}$ starting at some $\hat{r}'\in R\setminus\{r^*,r'\}$ and ending at a vertex that is visited by $P'_2$. 
	Then the temporal path $P^*$ in $(G,\lambda)$ from $\hat{r}'$ to $r'$ either uses no labels from $\lambda^+_{r^*}$ or no from $\lambda^-_{r^*}$.
	
	\noindent\textbf{Case A-2-ii(b).}
	Furthermore, assume there is a path $\hat{P}''$ in $H^-_{r^*}$ starting at some $\hat{r}''\in R\setminus\{r^*,r'\}$ and ending at a vertex that is visited by $P''_2$. 
	Then the temporal path $P^*$ in $(G,\lambda)$ from $\hat{r}''$ to $r'$ either uses no labels from $\lambda^+_{r^*}$ or no from $\lambda^-_{r^*}$.
	
	Assume that Cases A-2-ii(a) and~(b) do not apply. Then we modify $\lambda$ in the following way without changing its connectivity properties.
	First, we scale all labels in $\lambda$ by a factor of~$|V|$.
	
	The idea is to essentially switch the roles of $\hat{P}$ and $P_2'$.
	\begin{itemize}
		\item We remove $P_1$'s and $\hat{P}$'s edges and labels from $H^+_{r^*}$ and  $\lambda^+_{r^*}$, respectively, add $\hat{P}$'s edges to $H^-_{r^*}$. Add the edges from the tail of $\hat{P}$ to $r^*$ to $H^-_{r^*}$ and add labels for those edges to $\lambda^-_{r^*}$ such that there is a path from $r'$ to $r^*$ in $(H^-_{r^*},\lambda^-_{r^*})$ that uses the newly added labels. 
		\item We remove $P_1'$'s and $P_2'$'s edges and labels from $H^-_{r^*}$ and  $\lambda^-_{r^*}$, respectively, add $P_1'$'s and $P_2'$'s edges to $H^+_{r^*}$, and add new labels for the edges of $P_1'$ and $P_2'$ to $\lambda^+_{r^*}$ such that there is temporal path from $r^*$ to $\hat{r}$ in $(H^+_{r^*},\lambda^+_{r^*})$.
	\end{itemize}
	Note that after the modifications $H^+_{r^*}$ and $H^-_{r^*}$ are still trees, and $\lambda^+_{r^*}$ and $\lambda^-_{r^*}$ still assign at most one label per edge. Furthermore, we have that the modification do not increase the sum of edges in both trees $|E(H^+_{r^*})\cup E(H^-_{r^*})|$. Lastly, and most importantly, we have that the path from $r^*$ to $r$ in $H^+_{r^*}$ does not contain both edges $e$ and $e'$. It follows that we decreased $\alpha$.
	
	\noindent\textbf{Case A-2-iii.}  Assume the head of $P'_1$ is $u$. Assume there is a path $\hat{P}$ in $H^+_{r^*}$ starting at a vertex that is visited by $P_1$ but is different from its tail and extended head and ending at some $\hat{r}\in R\setminus\{r^*,r,r'\}$. Then the temporal path $P^*$ in $(G,\lambda)$ from $r'$ to $\hat{r}$ 
	either uses no labels from $\lambda^+_{r^*}$ or no from $\lambda^-_{r^*}$.
	Furthermore, assume there is a path $\hat{P}'$ in $H^-_{r^*}$ starting at some $\hat{r}'\in R\setminus\{r^*,r'\}$ and ending at a vertex that is visited by $P'_2$ but is different from its extended tail and head. Then the temporal path $P^*$ in $(G,\lambda)$ from $\hat{r}'$ to $r$ 
	either uses no labels from $\lambda^+_{r^*}$ or no from $\lambda^-_{r^*}$.
	
	We again modify $\lambda$ in a way that does not change its connectivity properties. 
	First, we scale all labels in $\lambda$ by a factor of $|V|$.
	We essentially switch the roles of $P_1$ and~$P_2'$.
	
	We remove $P_1$'s edges and labels from $H^+_{r^*}$ and  $\lambda^+_{r^*}$, respectively, add $P_1$'s edges to $H^-_{r^*}$, and add new labels for the edges of $P_1$ to $\lambda^-_{r^*}$ such that there are temporal paths from both endpoints of $e$ to $r^*$ that only use the new labels.
	We remove $P_2'$'s edges and labels from $H^-_{r^*}$ and  $\lambda^-_{r^*}$, respectively, add $P_2'$'s edges to $H^+_{r^*}$, and add new labels for the edges of $P_2'$ to $\lambda^+_{r^*}$ such that there are temporal paths from $r^*$ to both endpoints of $e$ that only use the new labels.
	
	Note that now there is a path from $\hat{r}'$ to $r^*$ in $(H^-_{r^*},\lambda^-_{r^*})$ that does not use edge $e$.
	Further note that after the modifications $H^+_{r^*}$ and $H^-_{r^*}$ are still trees, and $\lambda^+_{r^*}$ and $\lambda^-_{r^*}$ still assign at most one label per edge. Furthermore, we have that the modification do not increase the sum of edges in both trees $|E(H^+_{r^*})\cup E(H^-_{r^*})|$.
	It follows that we strictly decrease $|E^*_{r^*}|$ without increasing $\alpha$.
	
	\medskip
	
	Now consider the case where we have a temporal path $P$ from some $r\in R\setminus\{r^*\}$ to $r^*$ in $H^-_{r^*}$ that traverses both $e$ and $e'$ and two temporal paths $P_1,P_2$ from $r^*$ to some $r_1,r_2\in R\setminus\{r^*\}$, respectively, in $H^+_{r^*}$ such that $P_1$ traverses $e$ and $P_2$ traverses $e'$. This case is analogous to the previously discussed case.
	
	From now on we assume that Case A-2 does not apply.
	
	\medskip
	
	\noindent\textbf{Case B.} From now on we assume that none of the above described cases apply. This means that there is no path from $r^*$ to some $r\in R\setminus\{r^*\}$ in $H^+_{r^*}$ that traverses both $e$ and $e'$ and there is no path from some $r'\in R\setminus\{r^*\}$ to $r^*$ in $H^-_{r^*}$ that traverses both $e$ and $e'$. It follows that for every $e\in E^*_{r^*}$ we have a path in $H^+_{r^*}$ from $r^*$ to some $r\in R\setminus\{r^*\}$ that only traverses $e$ from the edges in $E^*_{r^*}$ and we have a path in $H^-_{r^*}$ from some $r'\in R\setminus\{r^*\}$ to $r^*$ that only traverses $e$ from the edges in $E^*_{r^*}$.
	All the following cases are illustrated in~\cref{fig:MSL-CaseB-all}.
	
	\begin{figure}[!htb]
		\centering
		\begin{subfigure}{0.49\textwidth}
			\centering
			\includegraphics[width=0.88\linewidth]{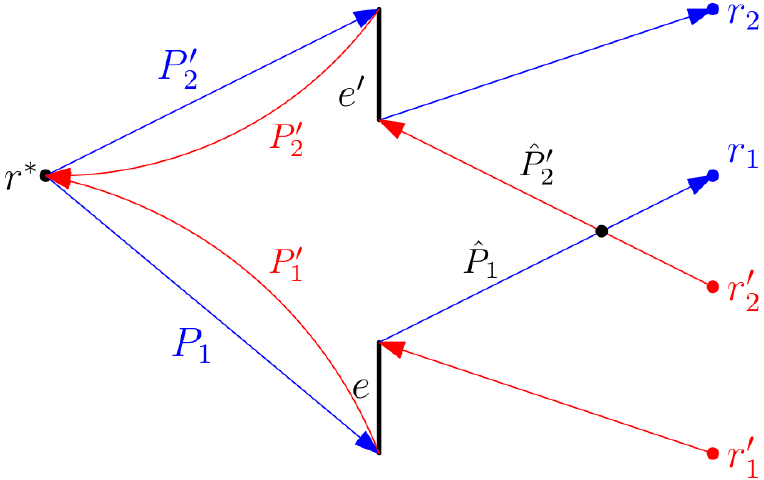}
			\caption{Case B-1-i: $|\hat{P}_1|\le |\hat{P}_2'|$ and $|\hat{P}_1|+ |\hat{P}_2'|\ge 3$ .}
			\label{fig:MSL-CaseB-1-i}
		\end{subfigure}%
		\quad
		\begin{subfigure}{0.48\textwidth}
			\centering
			\includegraphics[width=0.88\linewidth]{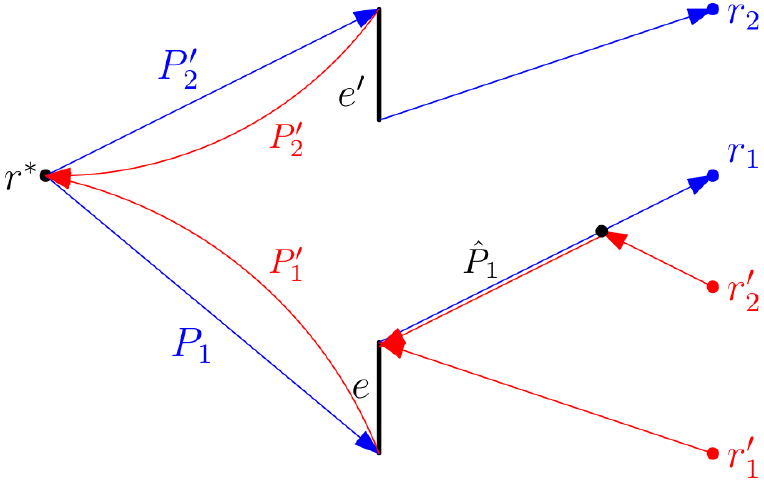}
			\caption{Modification of Case B-1-i.}
			\label{fig:CaseB-1-i_modification}
		\end{subfigure}
		
		\begin{subfigure}{0.48\textwidth}
			\centering
			\includegraphics[width=0.75\linewidth]{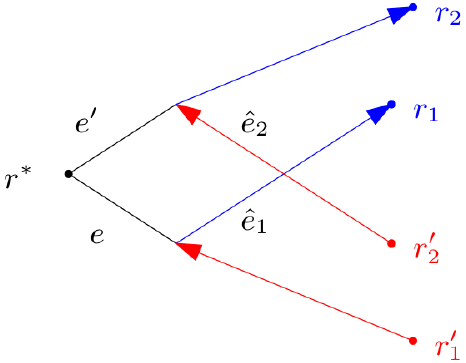}
			\caption{Case B-1-ii(a): The edges $e$, $e'$, $\hat{e}_1, \hat{e}_2$ form a~$C_4$.}
			\label{fig:MSL-CaseB-1-ii-a}
		\end{subfigure}
		\quad
		\begin{subfigure}{0.48\textwidth}
			\centering
			\includegraphics[width=0.88\linewidth]{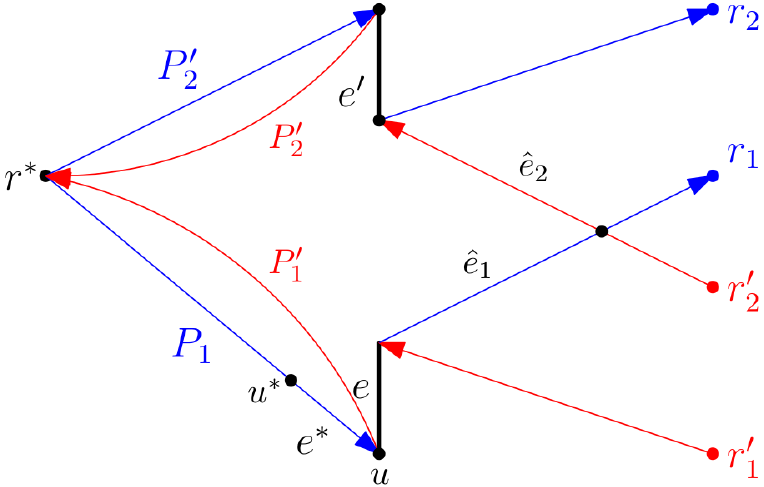}
			\caption{Case B-1-ii(b): $e$, $e'$, $\hat{e}_1, \hat{e}_2$ do not form a~$C_4$.}
			\label{fig:CaseB-1-ii-b}
		\end{subfigure}
		
		\begin{subfigure}{0.48\textwidth}
			\centering
			\includegraphics[width=0.88\linewidth]{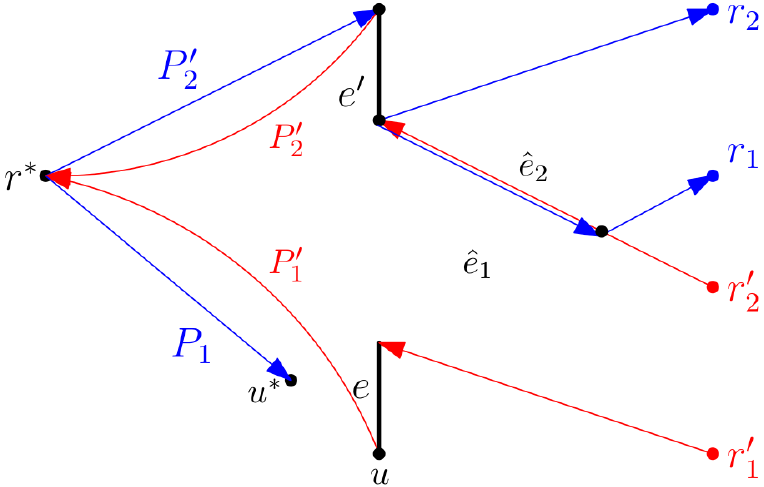}
			\caption{Modification of Case B-1-ii-b.}
			\label{fig:MSL-CaseB-1-ii-b_modification}
		\end{subfigure}
		\begin{subfigure}{0.48\textwidth}
			\centering
			\includegraphics[width=0.88\linewidth]{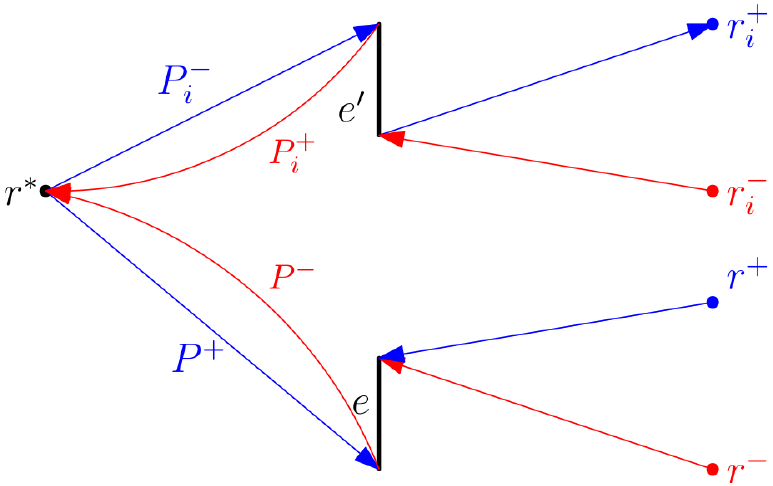}
			\caption{Case B-2: $r^+ \neq r_i^+$ and $r^- \neq r_i^-$.}
			\label{fig:MSL-CaseB-2}
		\end{subfigure}
		
		\caption{Cases B-1 -- B-2, where \emph{blue} color corresponds to the labeling $\lambda^+_{r^*}$ and \emph{red} to $\lambda^-_{r^*}$. \label{fig:MSL-CaseB-all}
		}
	\end{figure}

	\medskip
	
	\noindent\textbf{Case B-1.}
	Let $e,e'\in E^*_{r^*}$ and let $P_1$ be a path in $H^+_{r^*}$ from $r^*$ to some $r_1\in R\setminus\{r^*\}$ that only traverses $e$ from the edges in $E^*_{r^*}$ and let $P_2$ be a path in $H^-_{r^*}$ from some $r_2\in R\setminus\{r^*\}$ to $r^*$ that only traverses $e$ from the edges in $E^*_{r^*}$.
	Let $P_1'$ be a path in $H^+_{r^*}$ from $r^*$ to some $r_1'\in R\setminus\{r^*\}$ that only traverses $e'$ from the edges in $E^*_{r^*}$ and let $P_2'$ be a path in $H^-_{r^*}$ from some $r_2'\in R\setminus\{r^*\}$ to $r^*$ that only traverses $e'$ from the edges in $E^*_{r^*}$.

	Consider the case where all choices of $P_1,P_2,P_1',P_2'$ with the above properties we have $r_1=r_2'$ or $P_1$ and $P_2'$ intersect in a vertex after they traversed $e$ and $e'$, respectively. 
	Again for our analysis, we treat these two cases the same since in both cases we can assume that $r_2'$ can reach $r_1$, in the latter through the intersection point.
	The case where all choices of $P_1,P_2,P_1',P_2'$ with the above properties we have $r_1'=r_2$ or $P_1'$ and $P_2$ intersect in a vertex after they traversed $e'$ and $e$, respectively, is symmetric.
	
	Fix temporal paths $P_1,P_2,P_1',P_2'$ with the above properties and $r_1=r_2'$ or $P_1$ and $P_2'$ intersect in a vertex after they traversed $e$ and $e'$, respectively. Let $\hat{P}_1$ be the path segment from $e$ to the first vertex included in $P_2'$ (excluding $e$) and let $\hat{P}_2'$ be the path segment from the last vertex included in $P_1$ to $e'$ (excluding $e'$).
	
	\noindent\textbf{Case B-1-i.}
	Assume $|\hat{P}_1|\le |\hat{P}_2'|$ (the opposite case is symmetric) and $|\hat{P}_1|+ |\hat{P}_2'|\ge 3$ (not both paths are only a single edge). We remove $\hat{P}_2'$'s edges and $e$ and the corresponding labels from $H^-_{r^*}$ and  $\lambda^-_{r^*}$, respectively,
	such that there is a temporal path from $r_2'$ to $e$ that uses the new labels.
	
	Note that after the modifications $H^+_{r^*}$ and $H^-_{r^*}$ are still trees, and $\lambda^+_{r^*}$ and $\lambda^-_{r^*}$ still assign at most one label per edge. Furthermore, we have that the modification do not increase the sum of edges in both trees $|E(H^+_{r^*})\cup E(H^-_{r^*})|$. Lastly, and most importantly, we have that at least one of $H^+_{r^*}$ and $H^-_{r^*}$ does contain both edges $e$ and $e'$.
	
	\noindent\textbf{Case B-1-ii.}
	Assume $|\hat{P}_1|= |\hat{P}_2'|=1$, that is, both paths are only a single edge $\hat{e}_1$ and $\hat{e}_2$, respectively.
	
	\noindent\textbf{Case B-1-ii(a).}
	The edges $e$, $e'$, $\hat{e}_1$, and $\hat{e}_2$ form a $C_4$. Then we are in the case that $k$ is even.
	In this case we set $\hat{e}_1$ to be $e^+$ and we set $\hat{e}_2$ to be $e^-$. One of these two edges will be used to close the $C_4$, depending on whether which of $H^+_{r^*}$ and $H^-_{r^*}$ has fewer edges. The edges $e$ and $e'$ stay in $E_{r^*}^*$ and will be the only two edges for which we cannot account a label in $\lambda$ that is not present in $\lambda^+_{r^*}$ or $\lambda^-_{r^*}$. In this case we have that $|E_{r^*}^*|\le x'+2$ is fulfilled.
	
	\medskip
	
	If Case B-1-ii(a) never applies, then we are in the case that $k$ is odd and we have to be able to account a label in $\lambda$ that is not present in $\lambda^+_{r^*}$ or $\lambda^-_{r^*}$ for all but one edge in $E_{r^*}^*$.
	
	\medskip
	
	\noindent\textbf{Case B-1-ii(b).}
	The edges $e$, $e'$, $\hat{e}_1$, and $\hat{e}_2$ do not form a $C_4$.
	Let $e=uv$ and $e'=u'v'$ and let $u$ and $u'$ be the vertices closer to $r^*$ in $P_1$ and $P_2'$, respectively.
	Then this means there is either at least one edge $e^*$ between $r^*$ and $u$ or between $r^*$ and $u'$. Consider the case where $e^*$ is between $r^*$ and $u$ and let $e^*=uu^*$ for some vertex $u^*$. In this case $e^*$ is contained in $H^+_{r^*}$. The other case is symmetric. Let $e^{**}$ be the edge between $r^*$ and $u$ in $H^-_{r^*}$, that is incident with $u$.
	
	Note that $\lambda^+_{r^*}(e^*)<\lambda^+_{r^*}(e)<\lambda^-_{r^*}(e^{**})$.
	We now make the following modification. We remove label $\lambda^+_{r^*}(e^*)$ and add a new label to $\hat{e}_2$ in $\lambda^+_{r^*}$ that is chosen in a way that allows for a temporal path from $r^*$ to $r_1$ via $e'$ and then $\hat{e}_2$.
	
	\medskip
	
	\noindent\textbf{Case B-2.} 
	Fix some $e\in E^*_{r^*}$ and let $P^+$ be a path in $H^+_{r^*}$ from $r^*$ to some $r^+\in R\setminus\{r^*\}$ that only traverses $e$ from the edges in $E^*_{r^*}$ and let $P^-$ be a path in $H^-_{r^*}$ from some $r^-\in R\setminus\{r^*\}$ to $r^*$ that only traverses $e$ from the edges in $E^*_{r^*}$.
	For all $e_i\in E^*_{r^*}\setminus \{e\}$ let $P^+_i$ be a path in $H^+_{r^*}$ from $r^*$ to some $r^+_i\in R\setminus\{r^*\}$ that only traverses $e_i$ from the edges in $E^*_{r^*}$ and let $P_i^-$ be a path in $H^-_{r^*}$ from some $r^-_i\in R\setminus\{r^*\}$ to $r^*$ that only traverses $e_i$ from the edges in $E^*_{r^*}$. Note that for all $i\neq i'$ we have that $r^+_i\neq r^+_{i'}$ and $r^-_i\neq r^-_{i'}$. Now consider edge $e_i$. If $\lambda^+_{r^*}(e)\le \lambda^+_{r^*}(e_i)$, then the temporal path in $(G,\lambda)$ from $r^-_i$ to $r^+$ needs at least one label that is not contained in $\lambda^+_{r^*}$ or $\lambda^-_{r^*}$. 
	If $\lambda^+_{r^*}(e)> \lambda^+_{r^*}(e_i)$, then the temporal path in $(G,\lambda)$ from $r^-$ to $r^+_i$ needs at least one label that is not contained in $\lambda^+_{r^*}$ or $\lambda^-_{r^*}$.
	This implies, if Case B-1-ii(a) does not apply, that $\lambda$ contains at least $|E^*_{r^*}|-1$ labels that are not contained in $\lambda^+_{r^*}$ or $\lambda^-_{r^*}$ and hence $|E^*_{r^*}|\le x'+1$.
	If Case B-1-ii(a) applies, then $\lambda$ contains at least $|E^*_{r^*}|-2$ labels that are not contained in $\lambda^+_{r^*}$ or $\lambda^-_{r^*}$ and hence $|E^*_{r^*}|\le x'+2$.
	
	This finishes the proof.
\end{proof}

\medskip

Having \cref{lem:MSLstructure}, we can now give our algorithm for \MSL. As mentioned before, it uses an FPT-algorithm for \ST\ parameterized by the number of terminals~\cite{Dreyfus1971Steiner} as a subroutine. Recall the definition of \ST.

\problemdef{\ST}
{A static graph $G = (V,E)$, a subset of vertices $R \subseteq V$ and a positive integer $k$.}
{Is there a subtree of $G$ that includes all the vertices of $R$ and that contains at most $k$ edges.}

Let $(G,R,k)$ be an instance of \MSL. Note that if $G$ is $C_4$-free, then \cref{lem:MSLstructure} immediately implies that we can use an algorithm for \ST\ on the same input graph $G$ with the same terminal vertices $R$ and check whether the resulting solution subtree has at most $k^* = \lceil (k+1)/2 \rceil$ edges.
In the case where $G$ contains $C_4$s, we have to determine first whether there is a $C_4$ in $G$ that can be labeled in an optimal labeling. Formally, we show the following.

\begin{theorem}\label{thm:MinCoreFPT}
	\MSL\ is in FPT when parameterized by the number of terminals.
\end{theorem}

\begin{proof}
	Assume we have access to an algorithm $\mathcal{A}$ for \ST\ that on input $(G,R)$ outputs the size of a minimum solution, that is, an integer $k$ such that $(G,R,k)$ is a \textsc{YES} instance of \ST\ and $(G,R,k-1)$ is a \textsc{NO} instance of \ST.
	
	Let $(G,R,k)$ be an instance of \MSL\ and let $k^*=\mathcal{A}(G,R)$. For all $C_4$'s in $G$ let $k_{C_4}=\mathcal{A}(G,R\cup V(C_4))$. If there exist an $C_4$ in $G$ such that $k_{C_4}=k^*$, then $(G,R,k)$ is a \textsc{YES} instance of \MSL\ if and only if $k\ge 2k^*-2$. Otherwise $(G,R,k)$ is a \textsc{YES} instance of \MSL\ if and only if $k\ge 2k^*-1$.
	
	We first show correctness, then we analyse the running time. 
	
	($\Leftarrow$): Assume there exist a $C_4$ in $G$ such that $k_{C_4}=k^*$. Then there exist a subtree of $G$ connecting all terminal vertices and containing three edges of the $C_4$. We add the missing edge of the $C_4$ and label the subgraph using \cref{thm:optLabGossipC4}. This requires $2k^*-2$ labels and clearly afterwards all terminals can pairwise reach each other. Hence, we have that if $k\ge 2k^*-2$, then $(G,R,k)$ is a \textsc{YES} instance of \MSL. Assume there is no $C_4$ in $G$ such that $k_{C_4}=k^*$. Then there exist a subtree of $G$ connecting all terminal vertices and containing $k^*$ edges. We label this tree using \cref{thm:optLabGossipC4}. This requires $2k^*-1$ labels and clearly afterwards all terminals can pairwise reach each other. Hence, we have that if $k\ge 2k^*-1$, then $(G,R,k)$ is a \textsc{YES} instance of \MSL.
	
	($\Rightarrow$): Assume that $(G,R,k)$ is a \textsc{YES} instance of \MSL\ and let $k_\text{opt}\le k$ such that $(G,R,k_\text{opt})$ is a \textsc{YES} instance of \MSL\ and $(G,R,k_\text{opt}-1)$ is a \textsc{NO} instance of \MSL. By \cref{lem:MSLstructure}, we have that if $k_\text{opt}$ is odd, then there is a labeling $\lambda$ of size $k_\text{opt}$ for $G$ such that the edges labeled by $\lambda$ form a tree $H$, and every leaf of $H$ is a vertex in $R$. It is easy to see that $H$ is a solution for the \ST\ instance $(G,R)$. Hence, $\mathcal{A}(G,R)$ outputs a lower bound $k^*$ for the number of edges in $H$. Furthermore, since all leafs of $H$ are terminals, we have that every vertex in $(H,\lambda)$ can temporally reach every other vertex. By \cref{thm:optLabGossipC4} we know that then $\lambda$ needs $2k^*-1$ labels. This implies that $k\ge k_\text{opt}\ge 2k^*-1$.
	
	Now assume that $k_\text{opt}$ is even. Then by \cref{lem:MSLstructure} we have that there is a labeling $\lambda$ of size $k^*$ for $G$ such that the edges labeled by $\lambda$ form a graph $H$ that is a tree $H'$ with one additional edge that forms a $C_4$, and every leaf of $H'$ is a vertex in $R$.
	For the $C_4$ that is formed we have that $\mathcal{A}(G,R\cup V(C_4))$ outputs a lower bound $k^*$ for the number of edges in $H'$.
	Note that we have $k^*\le \mathcal{A}(G,R)$, since otherwise $2k^*-2>2\mathcal{A}(G,R)-1$, which means by \cref{thm:optLabGossipC4} that $k_\text{opt}< 2k^*-2$. However, since all leafs of $H'$ are terminals, we have that every vertex in $(H,\lambda)$ can temporally reach every other vertex. Hence, \cref{thm:optLabGossipC4} implies that $k_\text{opt}\ge 2k^*-2$. It follows that $k_\text{opt}< 2k^*-2$ leads to a contradiction and we have $k\ge k_\text{opt}\ge 2k^*-2$.
	
	\emph{Running time:} We can use the FPT-algorithm for \ST\ parameterized by the number of terminals by Dreyfus and Wagner~\cite{Dreyfus1971Steiner} for algorithm $\mathcal{A}$. Note that we need to iterate over all $C_4$s in $G$ (there are at most $n^4$ of them). Each time we invoke $\mathcal{A}(G,R\cup V(C_4))$, we increase the number of terminals by at most four. It follows that overall we obtain an FPT running time for the number of terminals as a parameter.
\end{proof}

\subsection{Parameterized Hardness of \MASL}\label{MASL-W1-hard-subsec}

Note that, since \MASL\ generalizes both \MSL\ and \MAL, NP-hardness of \MASL\ is already implied by both \cref{thm:MinCoreNP,thm:NPDiameterStatic}.
In this section, we prove that \MASL\ is W[1]-hard when parameterized by the number $|R|$ of the terminals, even if the restriction $a$ on the age is a constant. To this end, we provide a parameterized reduction from \MCC. This, together with \cref{thm:MinCoreFPT}, implies that \MASL\ is strictly harder than \MSL\ (parameterized by the number $|R|$ of terminals), unless FPT$=$W[1]. 

\begin{theorem}\label{thm:MACLW1hard}
	\MASL\ is W[1]-hard when parameterized by the number $|R|$ of the terminals, even if the restriction $a$ on the age is a constant.
\end{theorem}

\begin{proof}
	To prove that the \MASL\ is W[1]-hard when parameterized by the combination of the number $|R|$ of the terminals and the number $k$ of labels, even if the restriction $a$ on the age is a constant, we provide a parameterized polynomial-time reduction from \MCC\ parameterized by the number of colors, which is W[1]-hard~\cite{Fellows2009Parameterized}.
	
	\problemdef{\MCC}
	{A static graph $G = (V,E)$, a positive integer $k$, a vertex-coloring $c : V (G) \rightarrow \{1, 2, \dots , k\}$.}
	{Does $G$ have a clique of size $k$ including vertices of all $k$ colors?}
	
	Let $(G,k,c)$ be an input of the \MCC\ problem and denote $|V(G)| = n, |E(G)|=m$. 
	We construct $(G^*,R^*,a^*,k^*)$, the input of \MASL\ using the following procedure.
	The vertex set $V(G^*)$ consists of the following vertices:
	\begin{itemize}
		\item a ``color-vertex'' corresponding to every color of $V(G)$: $C = \{c_i | i \in \{1, 2, \dots , k\} \text{ a color}$ $\text{of } V(G)  \}$, 
		\item a ``vertex-vertex'' corresponding to every vertex of $G$: $U_V = \{u_v | v \in V(G)\}$,
		\item an ``edge-vertex'' corresponding to every edge of $G$: $U_E = \{u_e | e \in E(G)\}$,
		\item a ``color-combination-vertex'' corresponding to a pair of two colors of $V(G)$: $W = \{c_{i,j} | i,j \in \{1, 2, \dots , k\}, i < j, \text{ colors of } V(G) \}$, and
		\item $2n + 4m + 5m + \frac{11}{8}(k^4 - 2k^3 - k^2 + 2k) + \frac{11}{2}(k^3 - 3k^2 + 2k)$ ``dummy'' vertices.
	\end{itemize}
	The edge set $E(G^*)$ consists of the following edges:
	\begin{itemize}
		\item a path of length $3$ (using $2$ dummy vertices) between a color-vertex $c_i$, corresponding to the color $i$, and every vertex-vertex $u_v \in U_V$, where $v$ is of color $i$ in $V(G)$, \ie $c(v) = i$,
		\item for every edge $e=vw \in E(G)$, where $c(v) = i$ and $c(w)=j$, we connect the corresponding edge-vertex $u_e$ with
		\begin{itemize}
			\item[-] the vertex-vertices $u_v$ and $u_w$, each with a path of length $3$ (using $2$ dummy vertices),
			\item[-] the color-combination-vertex $c_{i,j}$, with a path of length $6$ (using $5$ dummy vertices),
		\end{itemize}
		\item a path of length $12$ (using $11$ dummy vertices), between each pair of color-combination-vertices, and
		\item a path of length $12$ (using $11$ dummy vertices), between all pairs of color-vertices $c_i$ and color-combination-vertices $c_{jk}$, where $i \notin \{j , k\}$, \ie we connect the color-vertex of color $i$ with all color-combination vertices of pairs of color that do not include $i$.
	\end{itemize}
	We set $R^* = C \cup W$ (note that $|R^*|\in O(k^2)$), $a^* = 12$ and $k^* = 6k + 6(k^2-k) + 6(k^2-k) + 3 (k^4 -2k^3 - k^2 + 2k) + 12(k^3 - 3k^2 + 2k)$. This finishes the construction. 
	It is not hard to see that this construction can be performed in polynomial time.
	For an illustration see \cref{fig:MASLfromMCC-W1}.
	At the end $G^*$ is a graph with $3n + 10m + \frac{1}{2}(k^2 + k) + \frac{11}{8}(k^4 -2k^3 - k^2 + 2k) + \frac{11}{2}(k^3 - 3k^2 + 2k)$ vertices 
	and $3n + 12m + \frac{3}{2}(k^4 - 2k^3 - k^2 + 2k) + 6(k^3 - 3k^2 + 2k)$ edges. 
	
	\begin{figure}[ht]
		\centering
		\includegraphics[width=0.8\linewidth]{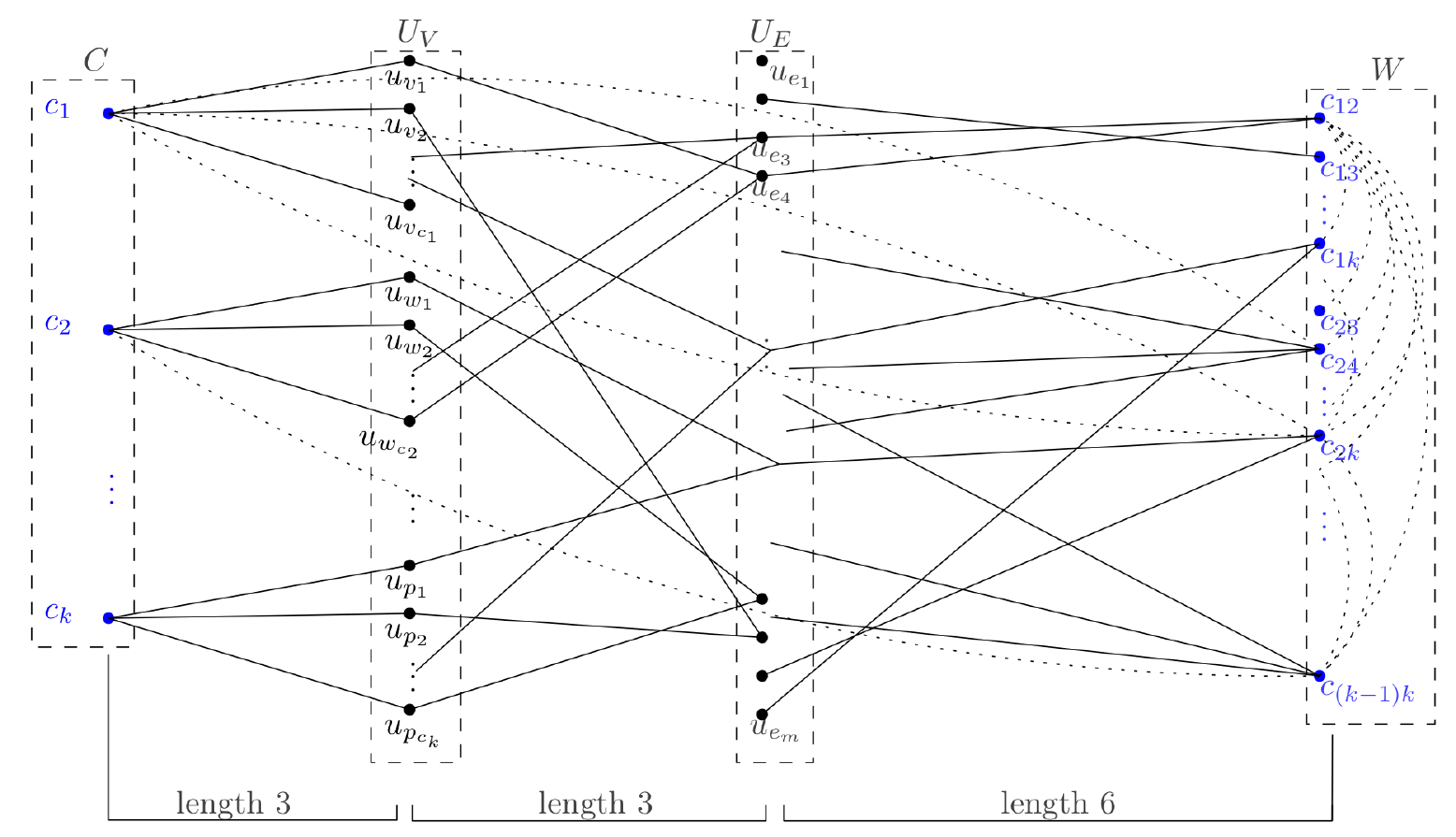}
		\caption{An example of the construction of the input graph for \MASL. For better readability, some paths among the vertices in $W$ and paths among $c_i \in C$ and $c_{jk} \in W$ ($i \neq j \neq k$), are not depicted.}
		\label{fig:MASLfromMCC-W1}
	\end{figure}
	
	We claim that $(G,k,c)$ is a \textsc{YES} instance of the \MCC\ if and only if $(G^*,R^*,a^*,k^*)$ is a \textsc{YES} instance of the \MASL.
	
	($\Rightarrow$):
	Assume $(G,k,c)$ is a \textsc{YES} instance of the \MCC. Let $S \subseteq V(G)$ be the set of vertices that form a multicolored clique in $G$.
	We construct a labeling $\lambda$ for $G^*$ that uses $k^*$ labels, which are not larger than $a^* = 12$, and admits a temporal path between all vertices from $R^*$ as follows.

	Let $U_S$ be the set of corresponding vertices to $S$ in $G^*$.
	For each $v \in S$ of color $i$ we label the three edges connecting $c_i$ to $u_v$ 
	with labels $1,2,3$, one per each edge, in order to create temporal paths starting in $c_i$ and
	with labels $12,11,10$, one per each edge, in order to create temporal paths that finish in $c_i$.
	For every edge $vw=e \in E$ with endpoints in $S$
	we label the path from both of its endpoint vertex-vertices $u_v, u_w$ to the edge-vertex $u_e$ with labels $4,5,6$, one per each edge, and
	with labels $9,8,7$, one per each edge.
	This ensures the existence of both temporal paths between $c_i$ and $c_j$.
	More precisely, $(c_i,c_j)$-temporal path (resp.~$(c_j,c_i)$-temporal path) uses labels $1,2,3$ to reach $u_v$ (resp.~$u_w$), from where it continues with $4,5,6$ to $u_e$, then with $7,8,9$ reaches $u_w$  (resp.~$u_v$) and finally with $10,11,12$ it finishes in $c_j$  (resp.~$c_i$).
	Note, since $S$ is a multicolored clique then each vertex $v' \in S$ is of a unique color $i'$
	and all vertices in $S$ are connected.
	Therefore, using the above construction for all vertices in $S$, vertex $c_i$ reaches and is reached by every other color vertex $c_j$ through the vertex-vertex $u_v$.
	Even more, since there is an edge $e$ connecting any two vertices $v,w \in S$,
	there is a unique edge vertex $u_e$ (and consequently a unique path), that is used for both temporal paths between vertex-vertices $u_v, u_w$ and their corresponding color-vertices.
	The above construction clearly produces a temporal path (of length $12$) between any two color-vertices.
	This construction uses $2 \cdot 3$ labels between every color-vertex $c_i$ and its unique vertex-vertex $u_v$, where $v \in S$ and $c(c)=i$,
	and $2 \cdot 6$ labels from each edge vertex $u_e$ to both of its endpoint vertex-vertices, where $e$ is an edge of the multicolored clique formed by the vertices in $S$.
	All in total we used $6k + 12 \binom{k}{2} = 6k + 6(k^2 -k)$ labels, to connect all edge-vertices corresponding to edges formed by $S$ with their endpoints vertex-vertices.
	
	Now, let $c_{ij}$ and $c_{i'j'}$ be two arbitrary color-combination-vertices.
	By the construction of $G^*$ there is a unique path of length $12$ connecting them, which we label with labels $1,2, \dots , 12$ in both directions.
	This labeling uses $2 \cdot 12$ labels for each pair of color-combination-vertices, hence all together we use $24 \frac{|W|(|W| - 1)}{2}$ labels, since $|W| = \binom{k}{2}$ this equals to $3(k^4-2k^3-k^2+2k)$.
	
	Finally, let $c_{i'}$ and $c_{ij}$ be two arbitrary color and color-combination-vertices, respectively.
	In the case when $i' \notin \{i,j\}$ there is a unique path of length $12$ in $G^*$ between them (that uses only the dummy vertices).
	We label this path with labels $1,2, \dots , 12$ in both directions.
	This procedure uses $2 \cdot 12$ labels for each pair of such vertices, hence all together we use $24 k \binom{k-1}{2}$ labels, which equals to $12(k^3 - 3k^2 + 2k)$.
	In the case when $i' \in \{i,j\}$ (w.l.o.g.\ $i' = i$)
	we connect the vertices using the following path.
	In $S$ exists a unique vertex of color $i$, denote it $v$. By the definition of $S$ there is also vertex $w$ of color $j$, which is connected to $v$ with some edge, denote it $e$.
	Therefore, to obtain a $(c_i, c_{ij})$-temporal path, we first reach $u_v$ from $c_i$ with labels $1,2,3$, then continue to $u_e$, using labels $4,5,6$, from where we continue to $c_{i,j}$ using the labels $7,8, \dots, 12$.
	The $(c_{ij}, c_{i})$-temporal path uses the same edges, with labels in reversed order.
	This construction introduced $2 \cdot 6$ new labels on the path of length $6$ between the edge-vertex $u_e$ and the color-class-vertex $c_{ij}$
	and reused all labels on the $(c_i, u_e)$-temporal paths.
	Repeating this for every color-class-vertex we use $2 \cdot 6 |W|$ new labels, since $|W| = \binom{k}{2}$ this equals to $6(k^2-k)$.
	
	All together $\lambda$ uses $6k + 6(k^2-k) + 6(k^2-k) + 3(k^4-2k^3-k^2+2k) + 12(k^3-3k^2+2k)$ labels.
	
	($\Leftarrow$):
	Assume that $(G^*,R^*,a^*,k^*)$ is a \textsc{YES} instance of the \MASL\ and let $\lambda$ be the corresponding labeling of $G^*$. 
	Before we construct a multicolored clique for $G$,
	we prove that the distance between any two terminal vertices from $R^*$ in $G^*$ is $12$.
	
	\medskip
	
	\noindent\textbf{Case A.} Let $c_i, c_j \in C$ be two arbitrary color-vertices
	and let $e$ be an edge in $G$ with endpoints of color $i$ and $j$, \ie $e = vw \in E(G)$ and $c(v) = i, c(w)=j$.
	There are two options how to reach $c_j$ from $c_i$.
	One when the path connecting them passes through the set $E$ and the other, when it passes through the set $W$.
	
	\noindent\textbf{Case A-1.}
	If the path passes through the set $E$,
	we must first go through a vertex-vertex $u_v$, then we go to the edge-vertex $u_e$, continue to the vertex-vertex $u_w$ and finish in $c_j$.
	Since all these vertices are connected with a path of length $3$, we get that the distance of the whole $(c_i,c_j)$-path is $12$.
	
	\noindent\textbf{Case A-2.}
	If the path passes through the set $W$,
	then we must go through the color-class-vertex $c_{ij}$.
	Since the path between any color-vertex and color-class-vertex is of length $12$ (we prove this in the following paragraph), the whole $(c_i,c_j)$-path is of length $24$.
	
	\medskip
	
	Therefore, the shortest path connecting two color-vertices is of length $12$ and must go through the appropriate edge-vertex.
	
	\medskip
	
	\noindent\textbf{Case B.} Let $c_{ij}$ and $c_{i'}$ be two arbitrary vertices from the color-class-vertices and color-vertices.
	We distinguish two cases. 
	
	\noindent\textbf{Case B-1.}
	First, when $i' \notin \{i,j\}$.
	Then, by the construction of $G^*$, there exists a direct path of length $12$, connecting them.
	Any other $(c_{i'}, c_{ij})$-path must 
	either go from $c_{i'}$ to some color-class-vertex $c_{i'j'}$, which is then connected with a path of length $12$ to the $c_{ij}$,
	or go to one of the color-vertices and then continue to the $c_{ij}$. 
	In both cases the constructed path is strictly longer than $12$.
	
	\noindent\textbf{Case B-2.}
	Second, when $i' \in \{i,j\}$.
	Let $c(v) = i$ and $vw = e \in E(G)$. Then there is a path from $c_i$ to $c_{ij}$ that goes through the vertex-vertex $u_v$ (using a path of length $3$), continues to the edge-vertex $u_e$ (using a path of length $3$), which is connected to the color-class-vertex $c_{ij}$ (using a path of length $6$).
	Hence the constructed $(c_i, c_{ij})$-path is of length $12$.
	There exists also another $(c_i, c_{ij})$-path, that goes through some other $c_{ij'}$ color-class-vertex, but it is longer than~$12$.
	
	\medskip
	
	\noindent\textbf{Case C.} Let $c_{ij}$ and $c_{i' j'}$ be two arbitrary color-class-vertices.
	By construction of $G^*$, there is a path of length $12$ connecting them.
	Any other $(c_{ij} ,c_{i' j'})$-path, must use at least one vertex-vertex, which is on the distance $9$ from the color-class-vertices (therefore the path through it would be of length at least $18$), 
	or a color-vertex, which is on the distance $12$ from the color-class-vertices.
	In both cases the constructed path is strictly longer than $12$.
	
	\medskip
	
	It follows that the distance between any two terminal vertices in $R^*$ is $12$, hence a temporal path connecting them must use all labels from $1$ to $12$.
	Using this property we know that any labeling that admits a temporal path among all terminal vertices must definitely use all labels $1,2, \dots, 12$ on the temporal paths
	among any two color-combination-vertices $c_{ij}$ and $c_{i'j'}$, and
	among a color-vertex $c_{i'}$ and a color-combination-vertex $c_{ij}$, where $i' \notin \{i,j\}$.
	This is true as there are unique paths of length $12$ among them.
	For these temporal paths we must use $2 \cdot 12 \frac{|W|(|W| - 1)}{2}$ labels (since $|W| = \binom{k}{2}$ this equals to $3(k^4-2k^3-k^2+2k)$)
	and $2 \cdot 12 k \binom{k-1}{2}$ labels (which equals to $12(k^3 - 3k^2 + 2k)$).
	Therefore, the labeling $\lambda$ can use only $6k + 6(k^2-k) + 6(k^2-k)$ labels to connect all other terminals.
	
	Let us now observe what happens with the temporal paths connecting remaining temporal vertices.
	To create a temporal path starting in a color-vertex $c_i$ and ending in some other color-vertex (or color-combination-vertex), $\lambda$ must label at least $3$ edges, to allow $c_i$ to reach one of its corresponding vertex-vertices $u_v$. Similarly it holds for a temporal path ending in $c_i$.
	Since the path connecting $c_i$ to some other terminal is of length $12$, the labels used on the temporal paths starting and ending in $c_i$ cannot be the same. In fact the labels must be $1,2,3$ for one direction and $12,11,10$ for the other.
	Therefore, $\lambda$ uses at least $6k$ labels on edges between vertices of $C$ and $U_V$.
	Extending the arguing from above, for $c_i$ to reach some (suitable) edge vertex $u_e$ the path needs to continue from $u_v$ to $u_e$ and must use the labels $4,5,6$ (or $9,8,7$ in case of the path in the opposite direction). 
	From $u_e$ the path can continue to the corresponding color-combination-vertex $c_{i,j}$ where it must use the labels $7,8, \dots, 12$, or to the vertex-vertex corresponding to the other endpoint of $e$.
	This finishes the construction of the temporal path from a color-vertex to the color-class-vertex
	and the temporal paths among color-vertices.
	The remaining thing is to connect a color-class-vertex with its corresponding color-vertices.
	The temporal path must go through some edge vertex $u_e$, that is on the distance $6$ from it, therefore the labeling must use the labels $1,2, \dots , 6$. From $u_e$ the path continues to the suitable vertex-vertex and then to the color-vertex.
	Using the above labeling we see that $\lambda$ must use at least $2 \cdot 6 |W|$ labels (which equals to $6(k^2-k)$ labels) on the edges between the color-class-vertices in $W$ and the edge vertices in $U_E$
	and at least $2 \cdot 6 \binom{k}{2}$ labels (which equals to $6(k^2-k)$ labels) on the edges between the edge-vertices in $U_E$ and vertex-vertices in $U_V$.
	Since all this together equals to $k^*$, all of the bounds are tight, \ie labeling cannot use more labels.
	
	We still need to show that for every color-vertex $c_i$ there exists a unique vertex-vertex $u_v$ connected to it such that all temporal paths to and from $c_i$ travel only through $u_v$.
	By the arguing on the number of labels used, we know that there can be at most two vertex-vertices that lie on temporal paths to or from $c_i$.
	More precisely, one that lies on every temporal path starting in $c_i$ and the other that lies on every temporal path that finishes in $c_i$.
	Let now $u_v, u_{v'}$ be two such vertex-vertices. Suppose that $u_v$ lies on all temporal paths that start in $c_i$ and $u_{v'}$ on all temporal paths that end in $c_i$.
	Now let $u_e$ be the edge-vertex on a temporal path from $c_i$ to $c_j$, and let $u_w$ be the vertex-vertex connected to $c_j$ and $u_e$.
	Therefore the $(c_i,c_j)$-temporal path has the following form: it starts in $c_i$, uses the labels $1,2,3$ to reach $u_v$, then continues to $u_e$ with $4,5,6$, then with $7,8,9$ reaches $u_w$ and with $10,11,12$ ends in the $c_j$.
	To obtain the $(c_i,c_{ij})$-teporal path we must label the edges from $u_e$ to $c_{ij}$ with the labels $6,7, \dots, 12$,
	since the edge-vertex $u_e$ is the only edge-vertex connected to the color-class-vertex $c_{ij}$ that can be reached from $c_i$ (if there would be another such edge-vertex, then the labeling $\lambda$ would use too many labels on the edges between $U_V$ and $U_E$).
	Now, for the color-vertex $c_j$ to be able to reach the color-class-vertex $c_{ij}$, it must use the same labels between the $u_e$ and $c_{ij}$ (using the same reasoning as before).
	Therefore the path from $c_j$ to $u_e$ (through) $u_{w}$ uses also the labels $1,2, \dots, 6$.
	But then for $c_j$ to reach $c_i$ the temporal path must use the vertex-vertex $u_w$, even more
	it must use the edge vertex $u_e$ and consequently the vertex-vertex $u_v$, from where it would reach $c_i$.
	But this is in the contradiction with the assumption that the path from $c_i$ to $u_v$ uses only labels $1,2,3$.
	Therefore, every color-vertex $c_i$ admits a unique vertex-vertex $u_v$ that lies on all $(c_i,c_j)$ and $(c_j,c_i)$-temporal paths.
	For the conclusion of the proof we claim that all vertices $v$ corresponding to these unique vertex-vertices $u_v$ of color-vertices $c_i$, form a multicolored clique in $G$.
	This is true as, by construction, a temporal path between two vertex-vertices $u_v, u_w$ corresponds to the edge $vw = e \in E(G)$.
	Since every vertex-vertex is connected to exactly one color-vertex, this corresponds to the vertex coloring of $V(G)$.
	In $G^*$ there is a temporal path among any two color vertices, therefore the vertex-vertices used in these temporal paths can be reached among each other,
	which means that they really do form a multicolored clique.
\end{proof}

\medskip

Note here that, in the constructed instance of \MASL\ in the proof of~\cref{thm:MACLW1hard}, 
the number of labels is also upper-bounded by a function of the number of colors in the instance of \MCC. Therefore the proof of~\cref{thm:MACLW1hard} implies also the next result, which is even stronger (since in every solution of \MASL\ the number of time-labels is lower-bounded by a function of the number $|R|$ of terminals).

\begin{corollary}
	\MASL\ is W[1]-hard when parameterized by the number $k$ of time-labels, even if the restriction $a$ on the age is a constant.
\end{corollary}

\section{Concluding remarks}

Several open questions arise from our results. As we pointed out in~\cref{thm:BoundsForCycles}, $\kappa(C_n,d) = \Theta(n^2)$, 
while $\kappa(G,d) = O(n^2)$ for every graph $G$ by Observation~\ref{kappa-bound-obs}. 
For which graph classes $\mathcal{G}$ do we have $\kappa(G,d) =o(n^2)$ (resp.~$\kappa(G,d) =O(n)$) for every $G\in\mathcal{G}$?

As we proved in~\cref{thm:NPDiameterStatic}, \MAL\ is NP-complete when the upper age bound is equal to the diameter $d$ of the input graph $G$. In other words, it is NP-hard to compute $\kappa(G,d)$. On the other hand, $\kappa(G,2r)$ can be easily computed in polynomial time, where $r$ is the \emph{radius} of $G$. Indeed, using the results of~\cref{gossip-polynomial-subsec}, it easily follows that, if $G$ contains (resp.~does not contain) a $C_4$ then $\kappa(G,2r)=2n-4$ (resp.~$\kappa(G,2r)=2n-3$). For which values of an upper age bound $a$, where $d\leq a \leq 2r$, can $\kappa(G,a)$ computed efficiently? In particular, can $\kappa(G,d+1)$ or $\kappa(G,2r-1)$ be computed in polynomial time for every undirected graph $G$?

With respect to parameterized algorithmics, is \MAL\ FPT with respect to the number $k$ of time-labels?

\end{document}